\documentclass[10pt]{article}
\usepackage[margin=1in]{geometry}
\usepackage[utf8]{inputenc}
\usepackage[usenames, dvipsnames]{xcolor}

\usepackage{float}

\usepackage{amsmath, amsthm, amssymb, thmtools, calc}
\usepackage{thm-restate}
\usepackage{fourier, bm, bbm}  % AC: If you don't like Fourier, change to Times or Pazo, but please not CMR!!
\usepackage[T1]{fontenc}
\usepackage{mleftright, xspace, paralist, enumitem, multirow, booktabs}
\usepackage[colorlinks, linkcolor=BrickRed, citecolor=blue]{hyperref}
\usepackage{url}
\usepackage{cite, cleveref}
\usepackage{algorithm}
\usepackage[noend]{algpseudocode}

\usepackage{tikz}
\usetikzlibrary{arrows.meta,
                chains,
                positioning}

\makeatletter
\def\th@plain{%
  \thm@notefont{}% same as heading font
  \itshape % body font
}
\def\th@definition{%
  \thm@notefont{}% same as heading font
  \normalfont % body font
}
\makeatother
\newtheorem{theorem}{Theorem}[section]
\newtheorem{lemma}[theorem]{Lemma}
\newtheorem{corollary}[lemma]{Corollary}

\theoremstyle{definition}\newtheorem{definition}[lemma]{Definition}

\theoremstyle{remark}  
\newlist{thmparts}{enumerate}{1}
\setlist[thmparts]{labelindent=\parindent,leftmargin=*,itemsep=2pt,font=\normalfont,label=(\thetheorem.\arabic*)}
\Crefname{thmpartsi}{Subresult}{Subresults} % any better name? subtheorem?
\declaretheoremstyle[%
  spaceabove=-6pt,%
  spacebelow=6pt,%
  headfont=\normalfont\itshape,%
  postheadspace=1em,%
  qed=\qedsymbol%
]{mystyle} 

\RequirePackage{caption, float}
\algrenewcomment[1]{\hfill \textcolor{blue}{$\triangleright$ #1}}
\algnewcommand{\LineComment}[1]{\State \textcolor{blue}{$\triangleright$ #1}}
\algdef{SE}[SUBALG]{Indent}{EndIndent}{}{\algorithmicend\ }%
\algtext*{Indent}
\algtext*{EndIndent}

\DeclareMathOperator{\EE}{\mathbb{E}}
\newcommand{\FF}{\mathbb{F}}
\newcommand{\NN}{\mathbb{N}}

\newcommand{\cA}{\mathcal{A}}
\newcommand{\cB}{\mathcal{B}}
\newcommand{\cC}{\mathcal{C}}

\newcommand{\cF}{\mathcal{F}}
\newcommand{\cG}{\mathcal{G}}
\newcommand{\cH}{\mathcal{H}}
\newcommand{\cI}{\mathcal{I}}

\newcommand{\cK}{\mathcal{K}}

\newcommand{\cN}{\mathcal{N}}

\newcommand{\cS}{\mathcal{S}}
\newcommand{\cT}{\mathcal{T}}

\newcommand{\tO}{\widetilde{O}}

\newcommand{\tOmega}{\widetilde{\Omega}}

\newcommand{\hatn}{\hat{n}}

\DeclareMathOperator{\poly}{poly}

\newcommand{\ceil}[1]{{\left\lceil{#1}\right\rceil}}
\newcommand{\floor}[1]{{\left\lfloor{#1}\right\rfloor}}

\newcommand{\eps}{\varepsilon}

\makeatletter
\newcommand\squeezepar{\@startsection{paragraph}{4}{\z@}{1.5ex \@plus1ex \@minus.2ex}{-1em}{\normalfont\normalsize\bfseries}}
\allowdisplaybreaks
\renewcommand{\paragraph}[1]{\squeezepar{{#1}.}}
\makeatother
\newcommand{\pell}{{(\ell)}} % parenthesized \ell
\newcommand{\pellxi}{{(\ell,\xi)}} % parenthesized \ell,\xi

 \newcommand{\mypar}[1]{\medskip\noindent{\bfseries #1.}~}

\title{Low-Memory Algorithms for Online and W-Streaming Edge Coloring}
\author{Prantar Ghosh\thanks{DIMACS, Rutgers University. Research supported in part by a grant (820931) to DIMACS from the Simons Foundation.} \and 
Manuel Stoeckl\thanks{Department of Computer Science, Dartmouth College. This work was supported in part by the National Science Foundation under award 2006589.}}
\date{}

\begin{document}

\maketitle

\begin{abstract}
For edge coloring, the online and the W-streaming models
seem somewhat orthogonal: the former needs edges to be 
assigned colors immediately after insertion, typically 
without any space restrictions, while the latter limits 
memory to sublinear in the input size but allows an 
edge's color to be announced any time after its 
insertion. We aim for the best of both worlds by 
designing small-space online algorithms for edge 
coloring. We study the problem under both (adversarial) 
edge arrivals and vertex arrivals. Our results 
significantly improve upon the memory used by prior 
online algorithms while achieving an $O(1)$-competitive 
ratio. In particular, for $n$-node graphs with maximum 
vertex-degree $\Delta$ under edge arrivals, we obtain an 
online $O(\Delta)$-coloring in $\tilde{O}(n\sqrt{\Delta})$ space. 
This is also the first W-streaming edge-coloring algorithm using $O(\Delta)$ colors (in sublinear memory). All prior works either used linear memory or $\omega(\Delta)$ colors. We also achieve a smooth color-space tradeoff: for any $t=O(\Delta)$, we get an $O(\Delta t (\log^2 \Delta))$-coloring in $\tilde{O}(n\sqrt{\Delta/t})$ space, improving upon the state of the art that used $\tilde{O}(n\Delta/t)$ space for the same number of colors (the $\tO(.)$ notation hides polylog$(n)$ factors). The improvements stem from extensive use of random permutations that enable us to avoid previously used colors. Most of our algorithms can be derandomized and extended to multigraphs, where edge coloring is known to be considerably harder than for simple graphs.
\end{abstract}

\section{Introduction}

A proper edge-coloring of a graph or a multigraph colors its edges such that no two adjacent edges share the same color. The goal is to use as few colors as possible. Any graph with maximum vertex-degree $\Delta$ trivially requires $\Delta$ colors to be properly edge-colored. A celebrated theorem of Vizing \cite{Vizing64} says that $\Delta+1$ colors suffice for any simple graph.\footnote{For multigraphs, $3\Delta/2$ colors are necessary and sufficient. \cite{Shannon49}} There are constructive  polynomial time algorithms that achieve a $(\Delta+1)$-edge-coloring in the classical offline setting \cite{MisraG92}. These algorithms are likely to be optimal with respect to the number of colors: distinguishing between whether the edge-chromatic number (i.e., the minimum number of colors needed to edge-color a graph) of a simple graph is $\Delta$ or $\Delta+1$ is NP-hard \cite{Holyer81}.  

The edge-coloring problem has several practical 
applications, including in switch routing 
\cite{aggarwal2003switch}, round-robin tournament 
scheduling \cite{JanuarioURD16}, call scheduling 
\cite{ErlebachJ01}, optical networks \cite{RaghavanU94}, 
and link scheduling in sensor networks 
\cite{GandhamDP05}. In many of these applications, such 
as in switch routing, the underlying graph is built 
gradually by a sequence of edge insertions and the color 
assignments need to be done instantly and irrevocably. 
This is modeled by the \emph{online} edge coloring 
problem. Due to its restrictions, an online algorithm 
cannot obtain a $(\Delta+1)$-coloring \cite{BarNoyMN92}. 
Consider, however, the simple greedy algorithm that 
colors every edge with the first available color that is 
not already assigned to any of its neighbors. Since each 
edge can have at most $2\Delta-2$ adjacent edges, this 
algorithm achieves a $(2\Delta-1)$-coloring, i.e., a 
competitive ratio of $2-o(1)$ (since the optimum is 
$\Delta$ or $\Delta+1$). Bar-Noy, Motwani, and Naor 
\cite{BarNoyMN92} showed that no online algorithm can 
perform better than this greedy algorithm. However, they 
proved this only for graphs with max-degree 
$\Delta=O(\log n)$. They conjectured that for 
$\Delta=\omega(\log n)$, it is possible to get better 
bounds, and that, in particular, a $(1+o(1))\Delta$-coloring
is possible. Several works \cite{aggarwal2003switch, BahmaniMM12, CohenPW19, BhattacharyaGW21, SaberiW21, KulkarniLSST22, NaorSW23} have studied online edge 
coloring with the aim of beating the greedy algorithm and/or resolving the said conjecture. Other variants of the problem have also been studied \cite{FavrholdtN03, Mikkelsen16, FavrholdtM18}. However, all prior works assume that all graph edges are always stored in the memory along with their colors.

With the ubiquity of big data in the modern world, this assumption often seems fallacious. The graphs that motivate the study of edge coloring, such as communication and internet routing networks, turn out to be large-scale or massive graphs in today's world, making it expensive for servers to store them entirely in their memory. This has led to big graph processing models such as \emph{graph streaming} that, similar to the online model, have sequential access to the graph edges, but can only store a small summary of the input graph so as to solve a problem related to it. There is an immediate barrier for the edge coloring problem in this setting: the output size is as large as the input, and hence an algorithm must use space linear in the input size to present the output as a whole. To remedy this, one can consider the natural extension of the model where the output is also reported in streaming fashion: in the context of edge coloring, think of the algorithm having a limited working memory to store information about both the input graph and the output coloring; it periodically streams or announces the edge colors before deleting them from its memory. This is the so called W-streaming model. Unlike the online model, here we don't need to assign a color to the incoming edge right away, and can defer it to some later time. However, due to the space restriction, we are not able to remember all the previously announced colors. Note that this makes even the greedy $(2\Delta-1)$-coloring algorithm hard (or maybe impossible) to implement in this model.

 In this work, we aim to get the best of both worlds of the online and the streaming models: we focus on designing \emph{low-memory} online algorithms for edge coloring. This is motivated by modern practical scenarios that demand immediate color assignment as well as space optimization. We succeed in designing such algorithms and at the same time, the quality of our algorithms is close to optimal: we achieve an $O(1)$-competitive ratio, i.e., a color bound of $O(\Delta)$. Note that no prior work studying edge-coloring in the sublinear-space setting could attain an $O(\Delta)$-coloring W-streaming algorithm, let alone online. For adversarial edge-arrival streams, we get an online $O(\Delta)$-coloring in $O(n\sqrt{\Delta})$ space, significantly reducing the space used by prior online algorithms at the cost of only a constant factor in the number of colors. We can smoothly tradeoff space with colors to get an $O(\Delta t)$-coloring in $\tO(n\sqrt{\Delta/t})$ space. This improves upon the state of the art \cite{CharikarL21,AnsariSZ22} which obtained the same color bound using $\tO(n\Delta/t)$ space. Furthermore, for the natural and well-studied settings of vertex-arrival in general graphs and one-sided vertex arrival in bipartite graphs, we can improve the space usage to $O(n \text{ polylog }n)$, i.e., semi-streaming, which is the most popular memory regime for graph streaming problems. Most of our algorithms generalize to multigraphs and can be made deterministic.

\subsection{Our Results and Contributions}

We study edge-coloring in the online model with sublinear (i.e., $o(n\Delta)$) memory as well as in the W-streaming model and improve upon the state of the art.
These results are summarized in \Cref{table:online} and \Cref{table:wstream}. They also mention the state of the art, for comparison.

\begin{table}[h]
\centering
\begin{tabular}{ c c c c c l }
\toprule
  Arrival & Algorithm & Colors & Space & Graph & Reference \\ 
  \midrule
   Edge & Randomized & $\left(\frac{e}{e-1} + o(1)\right)\Delta$ & $\tO(n \Delta)$& Simple & \cite{KulkarniLSST22}\\
  Edge & Randomized & $O(\Delta )$& $\tO(n \sqrt{\Delta})$ & Simple & \Cref{thm:online-rand-ea-edgecol} \vspace{1mm}\\ 
  Edge  & Deterministic & $(2\Delta-1) t$ & $O(n\Delta/t)$& Multigraph &\cite{AnsariSZ22}\\
     
  Edge & Deterministic & $\tO(\Delta t)$ & $\tO(n \sqrt{\Delta/t})\star$ & Multigraph & \Cref{thm:online-det-ea-edgecol}\vspace{2mm}\\  
  
  Vertex & Randomized & $(1.9+o(1))\Delta$ & $\tO(n\Delta)$ & Simple & \cite{SaberiW21}\\
  Vertex & Randomized & $O(\Delta)$ & $\tO(n)\star$ & Multigraph & \Cref{thm:online-rand-va-edgecol} \vspace{1mm}\\
  Vertex & Deterministic & $2\Delta-1$& $O(n\Delta)$ & Multigraph & Greedy folklore\\
  Vertex & Deterministic & $O(\Delta)$& $\tO(n)\star$ & Multigraph & \Cref{thm:online-det-va-edgecol}\vspace{2mm}\\

 One-sided vertex & Randomized & $(1+o(1))\Delta $ & $\tO(n\Delta)$ & Simple & \cite{CohenPW19}\\
  One-sided vertex & Randomized & $1.533\Delta$ & $\tO(n\Delta)$ & Multigraph & \cite{NaorSW23}\\
  One-sided vertex & Randomized & $5\Delta$ & $\tO(n)\star$ & Multigraph & \Cref{lem:online-bpt-rand-va-edgecol}\vspace{1mm}\\
  One-sided vertex & Deterministic & $2\Delta-1$ & $O(n\Delta)$ & Multigraph & Greedy folklore\\
  One-sided vertex & Deterministic & $O(\Delta)$ & $\tO(n)\star$ & Multigraph & \Cref{lem:online-bpt-det-va-edgecol}\\
\bottomrule
\end{tabular}
\caption{Our results in the online model. Here, $t=O(\Delta)$ is any positive integer. Algorithms marked with a $\star$ require oracle randomness for randomized algorithms and advice computable in exponential time for deterministic.}
\label{table:online}
\end{table}

\begin{table}[h]
\centering
\begin{tabular}{ c c c c c l }
\toprule
 Algorithm & Colors & Space & Graph & Reference\\
\midrule
 Randomized & $O(\Delta^2/s)$ & $\tO(ns)$ & Simple &\cite{CharikarL21}\\
 Randomized & $O(\Delta^2/s)$& $\tO(n \sqrt{s})$ & Simple & \Cref{cor:wstream-rand-ea}\\
 Randomized & $O(\Delta^2/s)$& $\tO(n \sqrt{s})\star$ & Multigraph & \Cref{thm:Wstream-rand-ea-edgecol}\vspace{1mm}\\
 Deterministic & $(1-o(1))\Delta^2/s$ & $O(ns)$& Simple &\cite{AnsariSZ22}\\
 Deterministic & $\tO(\Delta^2/s)$ & $\tO(n \sqrt{s})\star$ & Multigraph & \Cref{cor:wstream-det-ea}\\
\bottomrule
\end{tabular}
\caption{Our results in the W-streaming model. Here, $s\leq \Delta/2$ is any positive integer. Results marked with $\star$ require oracle randomness for randomized algorithms and advice computable in exponential time for deterministic.}
\label{table:wstream}
\end{table}

We consider the problem under (adversarial) edge-arrivals as well as vertex-arrivals. We give an account of our results in each of these models below.   

\mypar{Edge-arrival model} Here we design both online and W-streaming algorithms. 

\begin{theorem}[Formalized in \Cref{thm:online-rand-ea-formal}]\label{thm:online-rand-ea-edgecol}
  Given any adversarial edge-arrival stream of a simple graph, there is a randomized algorithm for online $O(\Delta)$-edge-coloring using $\tO(n \sqrt{\Delta})$ bits of space. % note: only for nonadaptive streams
\end{theorem}

Previously, there was no sublinear space online algorithm known for $O(\Delta)$-coloring. As observed in \Cref{table:online}, all prior algorithms need $\Theta(n\Delta)$ space in the worst case to achieve a color bound of $O(\Delta)$.

Note that \Cref{thm:online-rand-ea-edgecol} immediately implies a randomized W-streaming algorithm with the same space and color bounds. Although immediate, we believe that it is important to note it as a corollary.

\begin{corollary}\label{cor:wstream-rand-ea}
Given an adversarially ordered edge stream of any simple graph, there is a randomized W-streaming algorithm for $O(\Delta)$-edge-coloring using $\tO(n \sqrt{\Delta})$ bits of space.
\end{corollary}

The above result improves upon the state of the art algorithms of \cite{CharikarL21,AnsariSZ22} which, as implied by \Cref{table:wstream}, only obtain $\omega(\Delta)$-colorings for $o(n\Delta)$ space (the non-trivial memory regime in W-streaming). In fact, we improve upon them by a factor of $\Omega(\sqrt{\Delta})$ in space for $O(\Delta)$-coloring.

We show that the above W-streaming algorithm can be made to work for multigraphs and against adaptive adversaries at the cost of $\tO(n\Delta)$ bits of oracle randomness.

\begin{theorem}[Formalized in \Cref{thm:Wstream-rand-ea-formal}]\label{thm:Wstream-rand-ea-edgecol}
  Given an adversarially ordered edge stream of any multigraph, there is a randomized W-streaming algorithm for $O(\Delta)$ edge-coloring using $\tO(n \sqrt{\Delta})$ bits of space and $\tO(n\Delta)$ bits of oracle randomness. The algorithm works even against adaptive adversaries.  
\end{theorem}

Further, we prove that we can make the above algorithms deterministic at the cost of only a polylogarithmic factor in space. Once again, the online algorithm immediately implies a W-streaming algorithm.  

\begin{theorem}[Formalized in \Cref{thm:online-det-ea-formal}]\label{thm:online-det-ea-edgecol}
  Given an adversarial edge-arrival stream of edges of any multigraph, there is a deterministic algorithm for online $O(\Delta (\log^2 \Delta))$-edge-coloring using $\tO(n \sqrt{\Delta})$ bits of space. 
\end{theorem}

\begin{corollary}\label{cor:wstream-det-ea}
    Given an adversarially ordered  edge stream of any multigraph, there is a deterministic W-streaming algorithm for $O(\Delta (\log^2 \Delta))$-edge-coloring using $\tO(n \sqrt{\Delta})$ bits of space.
\end{corollary}

Furthermore, in each case, we can achieve a smooth tradeoff between the number of colors and the memory used. This is implied by a framework captured in the following lemma.

\begin{lemma}[Formalized and generalized in \Cref{lem:space-color-tradeoff-formal}]\label{lem:space-color-tradeoff}
  Suppose that we are given an $f(n,\Delta)$-space streaming algorithm $\cA$ for $O(\Delta)$-coloring any $n$-node multigraph with max-degree $\Delta$ under adversarial edge arrivals. Then, for any $s\geq 1$, there is a streaming algorithm $\cB$ for $O(s \Delta)$-coloring the same kind of graphs under adversarial edge arrivals using $f(n/s, s\Delta)+ \tO(n)$ bits of space.
\end{lemma}

For the online model, the above lemma combined with \Cref{thm:online-det-ea-edgecol} immediately gives the tradeoff of $\tO(\Delta t)$ colors and $\tO(n\sqrt{\Delta/t})$ space for any $t= O(\Delta)$, as claimed in \Cref{table:online}.  In other words, combined with \Cref{cor:wstream-rand-ea}, it implies the W-streaming bounds of $O(\Delta^2/s)$ colors and $O(n\sqrt{s})$ space for any $s= O(\Delta)$, as claimed in \Cref{table:wstream}. Note that our results match the tradeoff obtained by the state of the art for $t=\Theta(\Delta)$ and $s=O(1)$, and strictly improve upon them for $t=o(\Delta)$ and $s=\omega(1)$.

\mypar{Vertex-Arrival Model} We now turn to the weaker vertex-arrival model.  The online edge-coloring problem has been widely studied in this setting as well (see \Cref{sec:related} for a detailed discussion). Our online algorithms obtain significantly better space bounds than the edge-arrival setting. 

\begin{theorem}[Formalized in \Cref{thm:online-rand-va-formal}]\label{thm:online-rand-va-edgecol}
  Given any adversarial vertex-arrival stream of a multigraph, there is a randomized online $O(\Delta)$-edge coloring algorithm using $\tO(n)$ bits of space. It works even against an adaptive adversary and uses $\tO(n \Delta)$ oracle random bits.
\end{theorem}

% Formal version of thm:online-rand-va-edgecol
%\begin{theorem}\label{thm:online-rand-va-edgecol}
%   There is a randomized online $O(\Delta)$-edge coloring algorithm for vertex arrival streams over multigraphs using $O(n \log (n \Delta / \delta))$ bits of space, with error $\le \delta$ against any adaptive adversary. It uses $O(n \Delta \log \Delta)$ oracle random bits.
  
%   (This follows immediately by combining \Cref{lem:cvt-gen-to-bpt} with \Cref{lem:online-bpt-rand-va-edgecol}.)
% \end{theorem}

Thus, at the cost of only a constant factor in the number of colors, we can improve the memory usage from $\tO(n\Delta)$ to $\tO(n)$ for vertex-arrival streams. Since this algorithm immediately implies a W-streaming algorithm with the same bounds, we see that for vertex-arrival streams,
$O(\Delta)$-coloring can be achieved in semi-streaming space, the most popular space regime for graph streaming. Behnezhad et al.~\cite{BehnezhadDHKS19} mentioned that ``a major open question is whether [the number of colors for W-streaming edge-coloring] can be improved to $O(\Delta)$ while also
keeping the memory near-linear in $n$.'' Our results answer the question in the affirmative for vertex-arrival streams, which is a widely studied model in the streaming literature as well. 

Further, we show that the algorithm can be made deterministic using $\tO(n)$ bits of \emph{advice} instead of $\tO(n \Delta)$  bits of oracle randomness. By picking a uniformly random advice string, the same algorithm can alternatively be used as a robust algorithm with $1/\poly(n)$ error; the advice can also be computed in exponential time.

\begin{theorem}[Formalized in \Cref{thm:online-det-va-formal}]\label{thm:online-det-va-edgecol}
    Given any adversarial vertex-arrival stream of a multigraph, there is a deterministic online $O(\Delta)$-edge-coloring algorithm using $\tO(n)$ bits of space, using $\tO(n)$ bits of advice. 
\end{theorem}

An interesting special case of the vertex-arrival model is the one-sided vertex-arrival setting for bipartite graphs. Here, the vertices on one side of the bipartite graph are fixed, while the vertices on the other side arrive in a sequence along with their incident edges. A couple of works \cite{CohenPW19, NaorSW23} have studied online edge-coloring specifically in this model. We design low-memory online algorithms in this model (see \Cref{alg:bpt-rand-va-edgecol,alg:bpt-det-va-edgecol}) and use them as building blocks for our algorithms in the more general settings of vertex-arrival and edge-arrival. These algorithms maybe of independent interest due to practical applications of the one-sided vertex-arrival model; moreover, the randomized algorithm in this model uses only $5\Delta$ colors (as opposed to our algorithms where the hidden constant in $O(\Delta)$ is rather large).

Finally, we present a lower bound on the space requirement of a deterministic online edge-coloring algorithm.

\begin{theorem}[Formalized in \Cref{thm:det-online-va-lb-formal}]\label{thm:det-online-va-lb}
  For $\Delta\leq \eps n$ for a sufficiently small constant $\eps$, any deterministic online algorithm that edge-colors a graph using $(2-o(1))\Delta$ colors requires $\Omega(n)$ space. 
\end{theorem}

To the best of our knowledge, this is the first non-trivial space lower bound proven for an online edge-coloring algorithm.

An outline of how the several building blocks are put together to obtain the above results is given in
\Cref{fig:overview}.

\begin{figure}[!htb]
  \centering
  \begin{tikzpicture}[
node distance = 4mm and 9mm,
  start chain = going below,
   arr/.style = {{Straight Barb[scale=0.8]}-, rounded corners=1ex, semithick},
   arrd/.style = {{Straight Barb[scale=0.8]}-, rounded corners=1ex, semithick, dashed},
     N/.style = {draw, thick, fill=#1,
                 text width=10em, align=center, inner ysep=2ex},
   N/.default = white,
every edge/.style = {draw, arr}]   
\node (t1) [N=red!30] {\Cref{thm:online-rand-va-edgecol}};
\node (t2) [N=red!30, below=of t1, yshift=-9ex] {\Cref{thm:online-det-va-edgecol}};
\node (t3) [N=red!30, below=of t2, yshift=-9ex] {\Cref{thm:Wstream-rand-ea-edgecol}};
\node (t4) [N=red!30, below=of t3] {\Cref{thm:online-rand-ea-edgecol}};
\node (t5) [N=red!30, below=of t4] {\Cref{thm:online-det-ea-edgecol}};
\node (t6) [N=red!30, below=of t5] {\Cref{lem:space-color-tradeoff}};
\node (t7) [N=red!30, below=of t6] {\Cref{thm:det-online-va-lb}};

\node (l1c) [N=orange!30, right=of t1] {\Cref{lem:online-bpt-rand-va-edgecol}};
\node (l2c) [N=orange!30, right=of t2, yshift=-9ex] {\Cref{lem:online-bpt-det-va-edgecol}};

\node (l2cs) [N=orange!30, right=of l2c, yshift=-4ex] {\Cref{lem:approx-indep-small-anticoncentration}};
\node (l2cb) [N=orange!30, right=of l2c, yshift=+4ex] {\Cref{lem:approx-indep-large-anticoncentration}};

\node (emoment) [N=orange!30, right=of l1c] {\Cref{lem:moment-gen-func}};

\node (lve) [N=orange!30, right=of t3] {\Cref{lem:va-to-ea-conversion}};
\node (lbpt) [N=orange!30, below=of l1c] {\Cref{lem:cvt-gen-to-bpt}};
\node (lbpte) [N=orange!30, below=of lbpt] {\Cref{lem:cvt-gen-to-bpt-ea}};

\node (extcode) [N=yellow!30, right=of lbpt, yshift=-4.5ex] {\cite{SipserS96}, as \Cref{cor:practical-binary-codes}};

\node (lpartial) [N=orange!30, right=of t5, yshift=-4ex] {\Cref{lem:ea-partial-coloring}};
\node (lperm) [N=orange!30, right=of t5, yshift=+4ex] {\Cref{lem:fast-permutations}};
\node (loffline) [N=orange!30, right=of lpartial] {\Cref{lem:offline-ea-coloring}};
\node (mor) [N=yellow!30, right=of lperm] {\cite{Morris13}};

\node (l7c) [N=orange!30, right=of t7] {\Cref{lem:lb-rand-part-matching}}; 

\draw[arr]  (t1)-- (l1c);
\draw[arr]  (l1c)-- (emoment);
\draw[arr]  ([yshift=-1ex] t2.east) -- ++ (0.3,0) |- (l2c);
\draw[arr]  ([yshift=-1ex] l2c.east) -- ++ (0.5,0) |- (l2cs);
\draw[arr]  ([yshift=+1ex] l2c.east) -- ++ (0.5,0) |- (l2cb);
\draw[arr]  (lbpt.east) -- ++ (0.5,0) |- (extcode);
\draw[arr]  (lbpte.east) -- ++ (0.5,0) |- (extcode);
\draw[arr]  ([yshift=-1ex] t1.east) -- ++ (0.3,0) |- (lbpt);
\draw[arr]  ([yshift=+1ex] t2.east) -- ++ (0.3,0) |- (lbpt);
\draw[arr]  (t3)-- (lve);
\draw[arr]  (lperm)-- (mor);
\draw[arr]  ([yshift=+2ex] t3.east) -- ++ (0.5,0) |- (lbpte);
\draw[arr]  ([yshift=+1ex] t3.east) -- ++ (0.6,0) |- (l1c);

% skipping arrow from lperm to lemma 3.8 -- technically it's not needed for the deterministic result
%\draw[arrd]  ([yshift=-1ex] t2.east) -- ++ (0.7,0) |- (lperm);
\draw[arr]  ([yshift=-1ex] t4.east) -- ++ (0.5,0) |- (lperm);
\draw[arrd]  ([yshift=+1ex] t5.east) -- ++ (0.5,0) |- (lperm);
\draw[arr]  ([yshift=-1ex] t5.east) -- ++ (0.5,0) |- (lpartial);
\draw[arr]  (lpartial)-- (loffline);
\draw[arr]  (t7)-- (l7c);

    \end{tikzpicture}
    \caption{Overview of how the results in this paper fit together. Primary results are in red; main supporting lemmas in orange; and specific external results in yellow. \label{fig:overview}}
\end{figure}
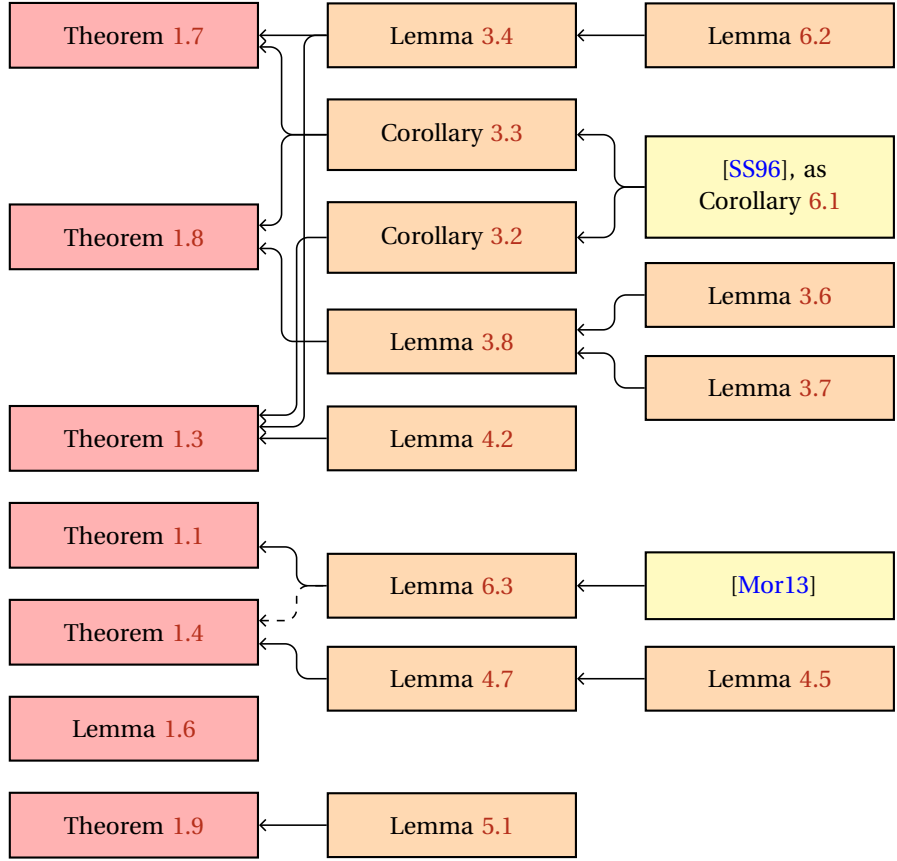

\subsection{Related Work}\label{sec:related}

\mypar{Online model} The edge-coloring problem has a rich literature in the online model \cite{aggarwal2003switch, AnsariSZ22, BarNoyMN92, BahmaniMM12, BhattacharyaGW21, CohenPW19, FavrholdtM18, FavrholdtN03, Mikkelsen15, Mikkelsen16, NaorSW23, KulkarniLSST22, SaberiW21}. The seminal work of Bar-Noy, Motwani, and Naor \cite{BarNoyMN92} showed that no online algorithm can do better than the greedy algorithm that obtains a $(2\Delta-1)$-coloring by assigning each edge the first available color that's not already used by any of its adjacent edges. However, this lower bound applies only to graphs with $\Delta=O(\log n)$. They conjectured that for $\Delta=\omega(\log n)$, there exist online $(1+o(1))\Delta$-coloring algorithms. Although this conjecture remains unresolved, there has been significant progress on it over the years. A number of works \cite{aggarwal2003switch, BahmaniMM12, BhattacharyaGW21} considered the problem under \emph{random-order} edge arrivals: Aggarwal et al.~\cite{aggarwal2003switch} showed that if $\Delta=\omega(n^2)$, then a $(1+o(1))\Delta$-coloring is possible. For $\Delta=\omega(\log n)$ (the bound in the said conjecture), Bahmani et al.~\cite{BahmaniMM12} obtained a $1.26\Delta$-coloring. Bhattacharya et al.~\cite{BhattacharyaGW21} then attained the ``best of both worlds'' by designing a $(1+o(1))\Delta$-coloring algorithm for $\Delta=\omega(\log n)$, resolving the conjecture for random-order arrivals. 

More relevant to our work is the setting of \emph{adversarial-order} edge arrivals. Cohen et al.~\cite{CohenPW19} were the first to make progress on \cite{BarNoyMN92}'s conjecture in this setting: they obtained a $(1+o(1))\Delta$-coloring for bipartite graphs under one-sided vertex arrivals (i.e., the nodes on one side are fixed, and the nodes on the other side arrive one by one with all incident edges). Their algorithm assumes a priori knowledge of the value of $\Delta$. For unknown $\Delta$, they prove that no online algorithm can achieve better than a $(e/(e-1))\Delta$-coloring, and also complement this result with a $(e/(e-1) +o(1))\Delta$-coloring algorithm for unknown $\Delta$. For  bipartite \emph{multigraphs} with one-sided vertex arrivals, Naor et al.~\cite{NaorSW23} very recently prove that $1.533\Delta$ colors suffice, while at least $1.207\Delta$ colors are necessary even for $\Delta=2$. Saberi and Wajc \cite{SaberiW21} showed that it is possible to beat the greedy algorithm for $\Delta=\omega(\log n)$ under vertex arrivals in general graphs: they design a $(1.9+o(1))\Delta$-coloring algorithm. Recently, Kulkarni et al.~\cite{KulkarniLSST22} made the first progress on the said conjecture in the general setting of adversarial edge arrivals: they obtained a $(e/(e-1) +o(1))\Delta$-coloring in this model. Note that the focus of all these works was on resolving \cite{BarNoyMN92}'s conjecture without any space limitations. Our focus is on designing low-memory online algorithms while staying within a constant factor of the optimal number of colors. The only prior sublinear-space online edge-coloring algorithm we know was given by Ansari et al.~\cite{AnsariSZ22}: a (deterministic) online $2\Delta t$-coloring in $O(n\Delta/t)$ space for any $t\leq \Delta$.

A number of works \cite{FavrholdtN03, EhmsenFKM10, FavrholdtM18} have studied the variant of the problem where given a fixed number of colors, the goal is to color as many edges as possible. Mikkelsen \cite{Mikkelsen15, Mikkelsen16} considered online edge-coloring with limited advice for the future.

\mypar{W-Streaming model} The W-streaming model \cite{DemetrescuFR06} is a natural extension of the classical streaming model for the study of problems where the output size is very large, possibly larger than our memory. While prior works have considered several graph problems in this model \cite{DemetrescuFR06, DemetrescuEMR10, LauraS11, GlazikSS22}, we are only aware of three papers \cite{BehnezhadDHKS19, CharikarL21, AnsariSZ22} that have studied edge-coloring here. Behnezhad et al.~\cite{BehnezhadDHKS19} initiated the study of W-streaming edge-coloring algorithms. They considered the problem for both adversarial-order and random-order streams: using $\tO(n)$ bits of working memory, they gave an $O(\Delta^2)$-coloring in the former setting, and a $(2e\Delta)$-coloring in the latter setting. Charikar and Liu \cite{CharikarL21} improved these results: for adversarial-order streams, for any $s=\Omega(\log n)$, they gave an $O(\Delta^2/s)$-coloring algorithm that uses $\tO(ns)$ space; and for random-order streams, they gave a $(1+o(1))\Delta$-coloring algorithm using $\tO(n)$ space. Both of the aforementioned algorithms for adversarial-order streams are, however, randomized. Ansari et al.~\cite{AnsariSZ22} gave simple deterministic algorithms achieving the same bounds of $O(\Delta^2/s)$ colors and $\tO(ns)$ space. Their algorithm can also be made online at the cost of a factor of $2$ in the number of colors. Note that parameterizing our results in \Cref{table:wstream} appropriately, our algorithms achieve $O(\Delta^2/s)$-colorings in $\tO(n\sqrt{s})$ space, matching the state of the art for $s=O(1)$, and strictly improving upon it for $s=\omega(1)$.

\mypar{Concurrent work} In an independent and parallel work, Behnezhad and Saneian \cite{BehnezhadS23} have designed a randomized $\tO(n\sqrt{\Delta})$-space W-streaming algorithm for $O(\Delta)$-edge-coloring for edge-arrival streams in simple general graphs. This matches our \Cref{cor:wstream-rand-ea}. Their result generalizes to give, for any $s \in [\sqrt{\Delta}]$, an $O(\Delta^{1.5} / s)$ coloring algorithm in $\tO(n s)$ space, while we achieve an $O(\Delta^2 / s)$-coloring in the same space. They also get an $O(\Delta)$-edge-coloring algorithm for vertex-arrival streams using $\tO(n)$ space, similar to our \Cref{thm:online-rand-va-edgecol}. Note that some of our edge-arrival algorithms have the additional strong feature of being online, while it is not clear if their edge-arrival algorithm can also be implemented in the online setting. 
In terms of techniques, while both works have some high level ideas in common, e.g., using random offsets/permutations to keep track of colors, or designing a one-sided vertex-arrival algorithm first and building on it to obtain the edge-arrival algorithm, the final algorithms and analyses in the two papers are fairly different.

Another independent work by Chechik, Mukhtar, and Zhang \cite{ChechikMZ23} obtains a randomized W-streaming algorithm that edge-colors an edge-arrival stream on general multi-graphs using $O(\Delta^{1.5} \log \Delta)$ colors in expectation\footnote{While \cite{ChechikMZ23} does not claim this, one can prove their algorithm uses $O(\Delta^{1.5} \log \Delta)$ colors with $\ge 1 - 1/\poly(n)$ probability.}, and $\tO(n)$ bits of space in expectation. Unlike us, they make no claims in the online model.

\section{Preliminaries}
\subsection{Notation}

Throughout the paper, logarithms are in base $2$. The notation $[t]$ indicates the set of integers $\{1,\ldots,t\}$. The notation 
 $\tO(x)$ ignores $\poly(\log(n),\log(\Delta))$ factors in $x$. $A \sqcup B$ gives the disjoint union of $A$ and $B$. $S_t$ is the set of permutations over $[t]$, and for any permutation $\sigma \in S_t$ and $X\subseteq [t]$, we denote $\sigma[X] := \{\sigma_i : i \in X\}$. For any set $X$, $\binom{X}{k}$ denotes the set of all $k$-sized subsets of $X$.

If not otherwise stated, $n$ is the number of vertices in a graph $G$, $V$ the set of vertices (or $A \sqcup B$ if the graph is bipartite), $E$ the (multi-)set of edges, and $\Delta$ is the maximum degree of the graph.

\subsection{Basic Definitions}

\begin{definition} 
  A random permutation $\sigma$ in $S_n$ is $k$-wise independent if, for all
  distinct $a_1,\ldots,a_k$ in $[n]$, and distinct $b_1,\ldots,b_k$ in $[n]$,
  we have:
  \begin{align*}
    \Pr\left[\bigwedge_{i \in [k]} \{ \sigma(a_i) = b_i \} \right] = \frac{1}{\prod_{i \in [k]} (n - i + 1)}
  \end{align*}
  
  A family of permutations is $k$-wise independent if the random variable for 
  a uniformly randomly chosen element of that family is $k$-wise independent.
  
  Per \cite{AlonL12}, while it is not known if there are nontrivial
  $k$-wise independent families of permutations for large $k$ and $n$, one can
  always construct weighted distributions which have support of size $n^{O(k)}$ and provide
  $k$-wise independence.

  A random permutation $\sigma$ is $(\epsilon,k)$-wise independent if for all distinct $a_1,\ldots,a_k$ in $[n]$, the distribution of $\sigma$ on $a_1,\ldots,a_k$ has total variation distance $\le \epsilon$ from uniform. In other words,
  \begin{align*}
      \frac{1}{2} \sum_{\text{distinct $b_1,\ldots,b_k$ in $[n]$}} \left| \Pr\left[\bigwedge_{i \in [k]} \{ \sigma(a_i) = b_i \} \right] - \frac{1}{\prod_{i \in [k]} (n - i + 1)} \right| \le \epsilon
  \end{align*}

  We say almost $k$-wise independent, when the random permutation is $(\epsilon,k)$-wise independent for sufficiently small $\epsilon$.
\end{definition}

\subsection{Models}

This paper will use the following models of presenting edges to an algorithm to be colored. In all cases, the set of vertices for the graph is known in advance. For general graphs, we call the set of vertices $V$; for bipartite graphs, $V$ is partitioned into two disjoint sets, which we typically call $A$ and $B$. Let $G$ be the (multi-) graph formed by taking the union of all edges in the stream.

We assume that the maximum degree $\Delta$ of $G$ is known in advance. An edge-coloring algorithm for which $\Delta$ is not known in advance can be converted to one which is, although one way to do this conversion (by running a new $2 \Delta$-coloring algorithm with a fresh set of colors whenever the maximum degree of graph formed by the input stream doubles) increases the total number of colors used by a constant factor, and requires  $O(n \log \Delta)$ bits of space to keep track of the maximum degree.
Since the algorithms in this paper already have large constant factors on number of colors used, it is not worth it to optimize the algorithms for the case where $\Delta$ is not known in advance.

\begin{definition}
    With an \emph{edge arrival stream}, the algorithm is given a sequence of edges in the graph. Each edge is provided as an ordered pair $\{x,y\}$ of vertices in $V$. In this paper, online algorithms processing edge arrival streams will implement a method $\textsc{Process}(\{x,y\})$ which returns the color assigned to the edge. For example, see  \Cref{alg:simple-greedy}, an implementation of the greedy edge coloring algorithm using $O(n \Delta)$ bits of space. W-streaming algorithms may assign the color for an edge at any time, although all edges must be given a color at the end of the stream.
\end{definition}

\begin{definition}
    In a \emph{vertex arrival stream}, the algorithm is given a sequence of (vertex,edge-set) pairs $(v,M_v)$, where the edge set $M_v$ contains all
    edges from $v$ to vertices that have been seen earlier in the stream. Online algorithms should report colors for all edges in $M_v$ when  $(v,M_v)$ is processed.

    A \emph{one-sided vertex arrival stream} on a bipartite graph with parts $A,B$ is like a vertex arrival stream, if which the vertices for one part ($B$) were all presented first, and then all the (vertex,edge-set) pairs for the other part ($A$) are given. For one-sided vertex arrival, we assume that the algorithm knows parts $A$ and $B$ in advance, and receives the
    (vertex,edge-set) pairs for $B$. The stream consists of pairs $(v,M_v)$, where each $v \in A$, and $M_v$ contains all edges from $v$ to $B$.
\end{definition}

An algorithm is said to be robust if it works with $\ge 1- \delta$ probability
even when its input streams are adaptively generated. By "adaptively generated",
we mean that the input is produced by an adaptive adversary that sees all outputs
of the online (or W-streaming) algorithm, and repeatedly 
chooses the next element of the stream based on what the algorithm has output so
far. See \cite{BenEliezerJWY20} for a more detailed explanation.

\begin{algorithm}[htb!]
  \caption{An implementation of a greedy $2\Delta-1$ online edge-coloring algorithm using $O(n \Delta)$ bits of space \label{alg:simple-greedy}}
  \begin{algorithmic}[1]
  \Statex \textbf{Input}: Stream of edges in an $n$-vertex graph $G = (V,E)$
    \Statex
    \Statex \underline{\textbf{Initialize}:}
    \For{$v \in V$}
    \State $U_v \gets \emptyset$ is a subset of $[2\Delta-1]$
    \EndFor
    \Statex
    
  \Statex \underline{\textbf{Process}(edge $\{x,y\}$) $\rightarrow$ color} 
    \State Let $c$ be arbitrary color in $[2\Delta-1] \setminus U_x \setminus U_y$
    \State Add $c$ to $U_x$ and to $U_y$
    \State \textbf{return} color $c$
  \end{algorithmic}
\end{algorithm}

\section{Edge coloring on vertex arrival streams}

\begin{lemma}[Deterministic general-to-bipartite partitioning]\label{lem:edge-gen-to-bpt}
   % see Lemma 38 or URN4
   For sufficiently large $n$, there is a set of $t = 4 \ceil{\log n}$ bipartite graphs $F_1,\ldots,F_t$, and an online algorithm $\cA$, which processes a stream of edges and assigns each edge to one of the $t$ graphs. The algorithm ensures that at each time, for each vertex $v$, $\deg_{F_i}(v) \le \frac{300}{\log n} \deg_G(x) + 1$. It uses $O(n (\log n) (\log \Delta))$ bits of space.
\end{lemma}

Using \Cref{lem:edge-gen-to-bpt} to route edges to $O(\log n)$ instances of an algorithm that
$f(\hat{\Delta})$ colors bipartite graphs of max degree $\le \hat{\Delta}$ implies the following corollary.

\begin{corollary}[Of \Cref{lem:edge-gen-to-bpt}]\label{lem:cvt-gen-to-bpt-ea}
  Say $f : \NN \mapsto \NN$ is a function for which $f(x) / x$ is monotonically
  increasing. Then given an algorithm $\cA$ for edge coloring with $f(\Delta)$
  colors on edge arrival streams over bipartite graphs of max degree $\Delta$,
  which uses $g(n,\Delta)$ bits of space, one can implement an algorithm
  $\cB$ for edge coloring with $O(f(\Delta))$ colors on general graphs,
  using $O((g(n,\Delta) + n \log \Delta) \log n)$ space.
\end{corollary}

Combining the previous corollary with that fact that one can convert an algorithm for one-sided vertex arrival streams on bipartite graphs to general ("two-sided") vertex arrival streams on bipartite graphs, only doubling the number of colors used, gives the following:

\begin{corollary}[Of \Cref{lem:edge-gen-to-bpt}]\label{lem:cvt-gen-to-bpt}
  Say $f : \NN \mapsto \NN$ is a function for which $f(x) / x$ is monotonically
  increasing. Then given an algorithm $\cA$ for edge coloring with $f(\Delta)$ 
  colors on one-sided vertex arrival streams over bipartite graphs of max degree
  $\Delta$, which uses $g(n,\Delta)$ bits of space,
  one can implement an algorithm $\cB$ for edge coloring under vertex arrivals of general graphs using 
  $O(f(\Delta))$ colors and $O((g(n,\Delta) + n \log \Delta) \log n)$ space.
\end{corollary}

\begin{proof}[Proof of \Cref{lem:edge-gen-to-bpt}]
  We claim \Cref{alg:gen-to-bpt} works for sufficiently large $n$.
    
  \begin{algorithm}[htb!]
    \caption{Algorithm to partition general graph edges into bipartite graphs\label{alg:gen-to-bpt}}
    \begin{algorithmic}[1]
    \Statex \textbf{Input}: Stream of edges in an $n$-vertex graph $G=(V,E)$
      \Statex
      \Statex \underline{\textbf{Initialize}:}
      \For{$v \in V$}
        \State $\deg_{F_i}(v) \gets 0$
      \EndFor
      \State Setup binary code $\cC$ of length $t := 4 \ceil{\log n}$ from \Cref{cor:practical-binary-codes}
      \For{$i \in [t]$}
        \State Let $F_i$ be the bipartition with parts $A_i = \{v \in V : \cC(v)_i = 0 \}$ and $B_i = \{v \in V : \cC(v)_i = 1 \}$
      \EndFor

    \Statex 
    \Statex\underline{\textbf{Process}(edge $\{x,y\}$)}
      \State Let $\deg(x) = \sum_{i \in [t]} \deg_{F_i}(x)$ and $\deg(y) = \sum_{i \in [t]} \deg_{F_i}(y)$.
      \For{$i \in [t]$}
          \If{$\cC(x)_i \ne \cC(y)_i$ and $\deg_{F_i}(x) \le 1200 \deg(x) / t$ and $\deg_{F_i}(y) \le 1200 \deg(y) / t$}\label{line:gen-to-bpt-edge-cond}
              \State Increase $\deg_{F_i}(x)$ and $\deg_{F_i}(y)$ by 1
              \State Assign edge $\{x,y\}$ to $F_i$
              \State \textbf{return}
          \EndIf
      \EndFor
      \State unreachable \label{line:gen-to-bpt-unreachable}
    \end{algorithmic}
  \end{algorithm}

  It is clear that at each point in time, for all $v \in V$ and $i \in [t]$, the algorithm will have $\deg_{F_i}(v)$ be the number of edges assigned to $F_i$ incident on $v$.
  
  Line \ref{line:gen-to-bpt-edge-cond} of the algorithm ensures that before
  edge $\{x,y\}$ is assigned, $\deg_{F_i}(x) \le 1200 \deg(x) / t$. Consequently, after the edge is assigned, $\deg_{F_i}(x) \le 1200 \deg(x) / t + 1 \le 300 \deg(x) / \log(n) + 1$. Similarly, we will have $\deg_{F_i}(y) \le 300 \deg(y) / \log(n) + 1$.
  
  It remains to prove that the algorithm will always assign an edge, and that Line \ref{line:gen-to-bpt-unreachable} is never reached. When processing edge an $\{x,y\}$,
  define $B_v := \{i \in [t]: \deg_{F_i}(v) > 1200 \deg(v) / t$. By Markov's inequality,
  since $\sum_{i\in[t]} \deg_{F_i}(v) = \deg(v)$, $|B_v| \le t / 1200$. Because the code
  $\cC$ has minimum distance $t / 400$, the set $K = \{i : \cC(x)_i \ne \cC(y)_i\}$ has size $\ge t/400$; and the set of $i \in [t]$ for which Line \ref{line:gen-to-bpt-edge-cond} passes has size $|K \setminus B_x \setminus B_y| \ge t/400 - t/1200 - t/1200 = t/1200$,
  and hence is nonempty.
\end{proof}

% note: this may be worth breaking up into correctness & probability of abort sublemmas

\begin{lemma}\label{lem:online-bpt-rand-va-edgecol}
  \Cref{thm:online-rand-va-edgecol} holds for one-sided vertex arrival streams on bipartite graphs.
\end{lemma}
% Similar to: URN4-36
\begin{proof}[Proof of \Cref{lem:online-bpt-rand-va-edgecol}]
  Consider \Cref{alg:bpt-rand-va-edgecol}. This algorithm will have the required properties if $\Delta \ge 6 \ln \frac{n}{\delta}$; if $\Delta$ is smaller, convert the vertex arrival stream to an edge arrival stream and pass it to \Cref{alg:simple-greedy}, which guarantees a $2\Delta-1$ coloring of the graph using $O(n \Delta) = O(n \ln \frac{n}{\delta})$ bits of space.

  \begin{algorithm}[htb!]
    \caption{Randomized algorithm for $5 \Delta$ edge coloring for (adversarial) one sided vertex arrival bipartite streams\label{alg:bpt-rand-va-edgecol}}
    \begin{algorithmic}[1]
    \Statex \textbf{Input}: Stream of vertex arrivals $n$-vertex graph $G=(A \sqcup B,E)$
      \Statex
      \Statex \underline{\textbf{Initialize}:}
      \State Let $C = 5 \Delta$.
      \For{$v \in B$}
        \State Let $\sigma_v$ be a uniformly random permutation over $[C]$ \Comment{constructed on demand from random oracle bits}.
        \State $h_v \gets 1$.
            
      \EndFor
      
    \Statex \underline{\textbf{Process}(vertex $x$ with multiset $M_v$ of edges to $B$)}
      \State Let $S \gets \emptyset$ \Comment{Set of colors $M_x$ will have used so far}
      \For{$e = \{x,y\}$ in $M_x$, in arbitrary order}
        \While{$h_y \le C \land \sigma_y[h_y] \in S$}\label{step:bpt-rand-va-edgecol-pick-loop}
          \State $h_y \gets h_y + 1$
        \EndWhile
        \If{$h_y > C$}\State abort \EndIf
        \State Assign color $\sigma_y[h_y]$ to $e$\label{step:bpt-rand-va-edgecol-color-assigned}
        \State $S \gets S \cup \{\sigma_y[h_y] \}$
        \State $h_y \gets h_y + 1$ \label{step:bpt-rand-va-edgecol-step-hy}
      \EndFor
    \end{algorithmic}
  \end{algorithm}

  This algorithm will never assign the same color to any pair of edges adjacent to
  the same vertex; at worst, it will abort. The condition on
  Line \ref{step:bpt-rand-va-edgecol-pick-loop} ensures
  that when a vertex $x$ is processed, no two edges will be assigned the same color.
  On the other hand, after Line \ref{step:bpt-rand-va-edgecol-color-assigned} assigns
  a color to an edge, Line \ref{step:bpt-rand-va-edgecol-step-hy} increases $h_y$;
  because $\sigma_v$ is a permutation, this prevents the algorithm from ever
  assigning the same color twice to edges incident on some vertex $y$ in $B$.
  
  For the rest of the proof, we will argue that the algorithm never aborts; equivalently,
  that $h_y \le C$ always holds for all $y \in B$. In fact, we shall prove the stronger
  claim, that $h_y \le C - 2 \Delta$ holds with probability $\ge 1 - \delta / n$
  for each individual $y \in B$. Consider a specific vertex $y \in B$.
  For each $i \in [\Delta]$, let $V_{y,i}$
  be the random variable counting the number of times that the loop starting at
  Line \ref{step:bpt-rand-va-edgecol-pick-loop} ran, when the $i$th edge adjacent
  to $y$ was processed. If there was no $i$th edge (or the algorithm already aborted), % alt: or if h_y >= C - 2Delta
  we set $V_{y,i} = 0$; then at the end of the stream, we will have
  $h_y \le \Delta + \sum_{i\in\Delta} V_{y,i}$.
  
  We now consider the distribution of $V_{y,i}$, conditioned on both the value of the
  variable $S$ at the time the $i$th edge was processed, and on the parts of
  the permutation $\sigma_y$ which the algorithm has read so far, $\sigma_y[1..{h_y-1}]$.
  \begin{align*}
    % V_{y,i} >= k implies that the next k elements are _all_ contained in S
    \Pr[V_{y,i} \ge k \mid S, \sigma_y[1..{h_y-1}]] &= \Pr[\sigma_y[h_y,\ldots,h_y + k - 1] \subseteq S \mid S, \sigma_y[1..{h_y-1}]] \\
      &= \binom{|S \cap \sigma_y[h_y,...,C]|}{k} / \binom{C - h_y - 1}{k}  \\
      & \le \binom{\Delta}{k} / \binom{2 \Delta}{k} \qquad\qquad \text{since $|S|\le\Delta$, $h_y \le C - 2\Delta$} \\
      &\le \frac{\Delta\cdot (\Delta-1) \cdots (\Delta -k +1)}{2\Delta \cdot (2\Delta-1) \cdots (2\Delta -k +1)} \le \frac{1}{2^k}
  \end{align*}
  
  Since this bound holds for all values of $S$ and all $\sigma_y[1..{h_y-1}]$, in particular we
  have
  \begin{align*}
    \Pr[V_{y,i} \ge k | (V_{y,j})_{j < i}] = \EE_{\text{$S,\sigma_y[1..{h_y-1}]$ compat with $(V_{y,j})_{j < i}$}} \Pr[V_{y,i} \ge k |S,\sigma_y[1..{h_y-1}]] \le \frac{1}{2^k}
  \end{align*}
  By \Cref{lem:moment-gen-func}, for $e^t \in [1,2)$, we have $\EE[e^{t V_{y,i} } \mid (V_{y,j})_{j < i}] \le 1 / (2 - e^t)$, and $\EE[V_{y,i} \mid (V_{y,j})_{j < i}] \le 1$.
  
  By a slight variation on the Chernoff bound:
  \begin{align*}
    \Pr[\sum_{i=1}^{\Delta} V_{y,i} \ge 2 \Delta] &\le \inf_{t\ge 0} \Pr\left[\prod_{i=1}^{\Delta} e^{t V_{y,i}} \ge e^{2 t \Delta}\right] \\
     &\le \inf_{t\ge 0} \frac{1}{e^{2 t \Delta}} \EE[e^{t V_{y,1}} \cdots \EE\left[e^{t V_{y,\Delta}} \mid (V_{y,j})_{j<\Delta}\right] \cdots ]  \\
     &\le \min_{t : e^t \in [1,2)} \frac{1}{e^{2 t \Delta}} \left(\frac{1}{2 - e^t}\right)^\Delta = \left( \min_{t : e^t \in [1,2)} \frac{1}{e^{2t} (2-e^t)}\right)^\Delta \\
     &= \left( \frac{1}{\max_{x \in [1,2)} x^2 (2-x)}\right)^\Delta = \left(\frac{27}{32}\right)^\Delta \le \exp(-\Delta/6)
  \end{align*}
  Since $h_v \le \Delta + \sum_{i\in\Delta} V_{y,i}$, this implies that
  \begin{align*}
    \Pr[h_v \ge C] \le \Pr[h_v \ge C - 2 \Delta] = \Pr[h_v \ge 3 \Delta] \le \Pr[\sum_{i=1}^{\Delta} V_{y,i} \ge 2 \Delta] \le  \exp(-\Delta/6)
  \end{align*}
  Thus, by a union bound, the probability that \emph{any} vertices $v \in B$ will have $h_v > C$ at the end of the algorithm will be $\le n \exp(-\Delta/6)$. In particular, if $\Delta \ge 6 \ln(n / \delta)$, the algorithm will the guaranteed to abort with probability $\le \delta$.
\end{proof}

Observe that \Cref{lem:cvt-gen-to-bpt} and \Cref{lem:online-bpt-rand-va-edgecol} collectively imply \Cref{thm:online-rand-va-edgecol}.

\begin{theorem}[Formal version of \Cref{thm:online-rand-va-edgecol}]\label{thm:online-rand-va-formal}
   There is a randomized online $O(\Delta)$-edge coloring algorithm for vertex arrival streams over multigraphs using $O(n \log (n \Delta / \delta))$ bits of space, with error $\le \delta$ against any adaptive adversary. It uses $O(n \Delta \log \Delta)$ oracle random bits.
\end{theorem}

Let us now turn to the deterministic version of the problem. We will introduce an a deterministic algorithm which uses advice, and show that if the advice is chosen randomly, the algorithm will work with high probability. This algorithm can also be used as a randomized algorithm, if we choose the advice uniformly at random, and we use a somewhat complicated analysis to show that this can be achieved using only $\tO(n)$ (random) bits of advice.

\begin{lemma}
\label{lem:approx-indep-large-anticoncentration}
  Let $C,t,w$ be integers, with $C \ge t \ge 512$, and $8 | C$. Let $\epsilon \le C^{-t - 1}$. Say that $F_1,\ldots,F_w$ are subsets of $[C]$, and define $s_i = |F_i|$ for all $i \in [w]$. We furthermore require $\min_{i \in [w]} s_i \ge \frac{1}{2} t$, and $\min_{i \in [w]} s_i \ge \frac{1}{2} \max_{i\in[w]} s_i$. Let $X \subseteq [C]$ satisfy $|X| \le \frac{1}{8} C$, and let $\sum_{i \in [w]} s_i \le \frac{1}{2} C$. Then if $\sigma_1,\ldots,\sigma_w$ are $(\epsilon,t)$-wise independent random permutations over $[C]$,
  \begin{align*}
    \Pr\left[\left| \bigcup_{i \in [w]} \sigma_i[F_i] \setminus X \right| < \frac{1}{8} \sum_{i \in [w]} s_i \right] \le \exp(- \frac{1}{2^9} t w) \,.
  \end{align*}
\end{lemma}

\begin{proof}[Proof of \Cref{lem:approx-indep-large-anticoncentration}]
  Since the $(\sigma_i)_{i\in [w]}$ are $(\epsilon,t)$-wise independent, in particular we have for any $i\in [w]$,$j \in [C]$, that $\Pr\left[j \in \sigma_i[F_i]\right] \le s_i / C + \epsilon$, and for any $Q \subseteq [C]$ with $|Q| \le t$, that 
  \begin{align}
    \Pr[Q \subseteq \sigma_i[F_i]] = \frac{s_i \cdot (s_i-1) \cdots (s_i - |Q| + 1)}{C \cdot (C-1) \cdots (C - |Q|+1)} + \epsilon \le \left(\frac{s_i}{C}\right)^{|Q|}  \label{eq:apx-comb-bound}
  \end{align}
  Let $U_i$ be a random subset of $[C]$ in which each element in $C$ is included independently with probability $s_i / C$. Eq. \ref{eq:apx-comb-bound} thus implies $\Pr[Q \subseteq \sigma_i[F_i]] \le \Pr[Q \subseteq U_i]$. Now, for any fixed set $H \subseteq [C]$, let $Y_{H,i} := |\sigma_i[F_i] \cap H|$, and $W_{H,i} = |U_i \cap H|$.  Then $\EE[Y_{H,i}] \le s_i |H| / C + \epsilon C = \EE[W_{H,i}] + \epsilon C$, and as a consequence of Eq. \ref{eq:apx-comb-bound}, we have for all $k \le t$, that $\EE[Y_{H,i}^k] \le \EE[W_{H,i}^k] + \epsilon C$. This lets us bound the moment generating function of $Y_{H,i}$, for nonnegative $z$:
  \begin{align}
    \EE e^{z Y_{H,i}} &\le \sum_{k = 0}^{\infty} \frac{1}{k!} (\EE (z Y_{H,i})^k) \nonumber\\
      &\le \sum_{k=0}^{t} \frac{1}{k!} (\EE (z Y_{H,i})^k) + \sum_{k=t+1}^{\infty} \frac{1}{k!} (\EE (z Y_{H,i})^k) \nonumber\\
      &\le \sum_{k = 0}^{t} \frac{1}{k!} (\EE (z W_{H,i})^k) + \sum_{k=t+1}^{\infty} \frac{1}{(k-t-1)!(t+1)!} ((z s_i)^{k}) + \epsilon t C \nonumber\\
      &\le \EE e^{z W_{H,i}} + \frac{(z s_i)^{t+1}}{(t+1)!} e^{z s_i}  + \epsilon t C \nonumber\\
      &\le \exp\left({\frac{s_i |H|}{C} (e^z - 1)} \right) +\frac{(z s_i)^{t+1}}{(t+1)!} e^{z s_i} +  \epsilon t C \label{eq:moment-bound}
  \end{align}
  
  We will use this to upper bound the probability that a supermartingale, which sums
  how many elements in each $\sigma_i[F_i]$ were present in $\bigcup_{j < i} \sigma_j[F_j]$, grows too large. Let $B_1$ be an arbitrary set of size $\frac{1}{8} C$ which contains $X$. (This definition will make the following analysis simpler than if we had set $B_1 = X$). For each $i \in \{2,\ldots, w\}$, define $B_i = B_{i-1} \cup \sigma_{i-1}[F_{i-1}]$. Note that $|B_w| \le \frac{5}{8} C$. We have $\EE[\sum_{i \in [w]} Y_{B_i,i}] \le \sum_{i \in [w]}{\left(\frac{|B_i|}{C} s_i + \epsilon C\right)} \le \frac{|B_w|}{C} \sum_{i \in [w]} s_i \le \frac{5}{8} \sum_{i \in [w]} s_i$.
  
  Note that
  \begin{align}
    \Pr\left[\left| \bigcup_{i \in [w]} \sigma_i[F_i] \setminus X \right|  < \frac{1}{8} \sum_{i \in [w]} s_i \right] & \le \Pr\left[\sum_{i \in [w]} Y_{B_i,i} \ge \frac{7}{8} \sum_{i\in[w]} s_i \right] \label{eq:ybi-excess}
  \end{align}
  Let $\gamma = (\frac{7}{8} \sum_{i\in[w]} s_i) / \EE[\sum_{i \in [w]} Y_{B_i,i}] $; this is $\ge 7/5$. Applying a modified proof of the Chernoff bound/Azuma's inequality to the right hand side of Eq. \ref{eq:ybi-excess} gives:
  \begin{align*}
    &= \inf_{z > 0} \Pr\left[\exp\left(z \sum_{i \in [w]} Y_{B_i,i}\right) \ge \exp\left(\gamma z \EE[\sum_{i \in [w]} Y_{B_i,i}]\right) \right] \\
    &\le  \inf_{z > 0} \frac{\exp\left(z \sum_{i \in [w]} Y_{B_i,i}\right) }{\exp\left(\gamma z \EE[\sum_{i \in [w]} Y_{B_i,i}]\right)} \qquad\qquad \text{by Markov's inequality} \\ 
    &= \inf_{z > 0} \exp\left(z \sum_{i \in [w]} (Y_{B_i,i} - \gamma \EE[Y_{B_i,i} ])\right) \\
    &= \inf_{z > 0} \EE\left[\exp\left(z (Y_{B_1,1} - \gamma \EE[Y_{B_1,1}])\right) \cdots \EE\left[\exp\left(z (Y_{B_w,w} - \gamma \EE[Y_{B_w,w}])\right) \big\vert Y_{B_1,1}, \ldots, Y_{B_{w-1},w-1}\right] \cdots  \right] \\
    &\le \inf_{z > 0} \prod_{i \in [w]} \max_{B: \frac{1}{8} C \le |B| \le \frac{5}{8} C} \EE \exp\left(z (Y_{B,i} - \gamma \EE[Y_{B,i}])\right) \qquad\qquad \text{bounding terms from the inside out} \\
    &= \inf_{z > 0} \prod_{i \in [w]} \max_{B: \frac{1}{8} C \le |B| \le \frac{5}{8} C} \frac{\EE \exp\left(z Y_{B,i}\right)}{\exp\left(\gamma z\EE[Y_{B,i}])\right) } \\
    &\le \inf_{z > 0} \prod_{i \in [w]} \max_{B: \frac{1}{8} C \le |B| \le \frac{5}{8} C} \frac{\exp\left({\frac{s_i |B|}{C} (e^z - 1)} \right) +\frac{(z s_i)^{t+1}}{(t+1)!} e^{z s_i} +  \epsilon t C}{\exp\left(\gamma z\EE[Y_{B,i}])\right) } \qquad\qquad \text{by Eq. \ref{eq:moment-bound} }
  \end{align*}
  Now set $z = \frac{t}{8 \frac{1}{w} \sum_{i\in[w]} s_i}$. Since $t \le 2 \min_{i \in [w]} s_i$, it follows $z \le \frac{2 \min_{i \in [w]} s_i}{8 \min_{i \in [w]} s_i} \le \frac{1}{4}$. Since $\max_{i \in [w]} s_i \le 2 \min_{i \in [w]} s_i$, we have $z \le t / (4 s_i)$ for all $i \in [w]$. This implies $(zs_i)^{t+1}/(t+1)! e^{z s_i} \le {(t/4)}^{t+1} e^{t/4} / (t+1)! \le \frac{1}{2}$.
  Since $\epsilon \le C^{-t-1}$, we also have $\epsilon t C \le \frac{1}{2}$. Continuing the upper bound of Eq. \ref{eq:ybi-excess}:
  \begin{align*}
    &\le \prod_{i \in [w]} \max_{B: \frac{1}{8} C \le |B| \le \frac{5}{8} C} \frac{\exp\left({\frac{s_i |B|}{C} (e^z - 1)} \right) + 1}{\exp\left(\gamma ((z s_i |B| / C) - \epsilon) \right) } \hspace{-2cm} \\
    &\le e^{\gamma w \epsilon} \prod_{i \in [w]} \frac{\exp\left({\frac{s_i}{8} (e^z - 1)} \right) + 1}{\exp\left(\gamma z s_i / 8 \right) }  &\text{maximum occurs at $|B| = \frac{1}{8} C$}\\
    &\le e^{w (\ln(2) + \gamma \epsilon)} \prod_{i \in [w]} \exp( (e^z - 1 - \gamma z) s_i / 8 )  &\text{since $(x+1)/x^\gamma \le 2/x^{\gamma -1}$} \\
    &\le e^{w (\ln(2) + \gamma \epsilon)} \prod_{i \in [w]} \exp\left( -\frac{1}{4} z \frac{s_i}{8} \right)  &\hspace{-4cm}\text{since $z \le \frac{1}{4}$ implies $\frac{1}{4} z \le \frac{7}{5} z - (e^z-1)$} \\
    & = \exp\left( - \frac{1}{4} \frac{t w}{8 \sum_{i \in [w]} s_i} \sum_{i \in [w]} \frac{s_i}{8}  + w (\ln 2 + \gamma \epsilon)\right) \hspace{-2cm} \\
    &\le \exp( - \frac{t w}{2^8}  + w (\ln 2 + \gamma \epsilon)) \\
    &\le \exp( - \frac{t w}{2^9})  &\text{since $\frac{t}{2^9} \ge 1 \ge \ln 2 + \gamma \epsilon$}
  \end{align*}
  This completes the proof of the lemma.
\end{proof}

\begin{lemma}\label{lem:approx-indep-small-anticoncentration}
  Let $C,s,w$ be integers, with $s \ge 4$. Let $\epsilon \le C^{-s - 1}$. Say that $F_1,\ldots,F_w$ are subsets of $[C]$, with each $|F_i| \le s$, $|F_i| \ge 2$ and $\sum_{i \in [w]} |F_i| \ge \frac{1}{2} C$. Furthermore, let $X \subseteq [C]$ satisfy $|X| \le \frac{1}{8} C$. Then if $\sigma_1,\ldots,\sigma_w$ are $(\epsilon,s)$-wise independent random permutations over $[C]$,
  \begin{align*}
    \Pr\left[\left| \bigcup_{i \in [w]} \sigma_i[F_i] \setminus X \right| < \frac{1}{16} C \right] \le \exp(- \frac{1}{2} \sum_{i \in[w]} |F_i|) \,.
  \end{align*}
  
  This lemma differs from \Cref{lem:approx-indep-large-anticoncentration} in that the sets $(\sigma_i[F_i])_{i \in [w]}$ are now smaller, and $\sum_{i\in[w]} |F_i|$ is $\Omega(C)$.
\end{lemma}

\begin{proof}[Proof of \Cref{lem:approx-indep-small-anticoncentration}]
  Because the $(\sigma_i)_{i \in [w]}$ are $(\epsilon,s)$-wise independent, for each $i \in [w]$, the random variable $\sigma_i[F_i]$ has total variation distance $\epsilon$ from being a uniform random subset of $[C]$ of size $|F_i|$. Let $\tau = \floor{\frac{1}{16} C}$. We will bound the probability that there exists any set $T \subseteq [C]$ of size $\tau$ for which $\bigcup_{i \in [w]} \sigma_i[F_i] \subseteq T \cup X$. Observe:
  \begin{align*}
    \Pr&\left[\exists T \in \binom{[C]}{\tau} : \bigcup_{i \in [w]} \sigma_i[F_i] \subseteq T \cup X \right] \\
      &\le \sum_{T \in \binom{[C]}{\tau}} \Pr\left[\bigcup_{i \in [w]} \sigma_i[F_i] \subseteq T \cup X \right] \\
      &= \sum_{T \in \binom{[C]}{\tau}} \prod_{i \in [w]} \Pr\left[\sigma_i[F_i] \subseteq T \cup X \right] & \text{since $\sigma_i$ are independent} \\
      &\le \sum_{T \in \binom{[C]}{\tau}} \prod_{i \in [w]} \left(\frac{\binom{|T \cup X|}{|F_i|}}{\binom{C}{|F_i|}} + \epsilon\right) & \hspace{-2cm}\text{since $F_i$ approximately uniform} \\
      &= \sum_{T \in \binom{[C]}{\tau}} \prod_{i \in [w]} \left(\frac{|T\cup X| \cdots (|T \cup X| - |F_i| + 1)}{C \cdots (C - |F_i| + 1)} + \epsilon\right)\hspace{-4cm} \\
      &\le \sum_{T \in \binom{[C]}{\tau}} \prod_{i \in [w]} \left(\frac{|T\cup X|}{C}\right)^{|F_i|} & \hspace{-2cm}\text{since $2 \le |F_i| \le s$ and $\epsilon \le C^{-s - 1}$} \\
      &\le \binom{C}{\tau} \left(\frac{|T\cup X|}{C}\right)^{\sum_{i \in [w]} |F_i|} \\
      &\le 2^{0.338 C} \left(\frac{3}{16}\right)^{\sum_{i \in [w]} |F_i|} & \text{using $|T|+|X| \le \left(\frac{1}{16} + \frac{1}{8}\right) C$ and $\tau \le \frac{C}{16}$}\\
      &\le \exp( 0.235 C - \ln(16/3) \sum_{i \in [w]} |F_i|) \\
      &\le \exp( - \frac{1}{2} \sum_{i \in [w]} |F_i| ) & \text{since $\sum_{i \in [w]} |F_i| \ge C / 2$ }
  \end{align*}
  This completes the proof.
\end{proof}

\begin{lemma}\label{lem:online-bpt-det-va-edgecol}
  \Cref{thm:online-det-va-edgecol} holds for one-sided vertex arrival streams on bipartite graphs.
\end{lemma}

% Compare with URN4-48
\begin{proof}[Proof of \Cref{lem:online-bpt-det-va-edgecol}]
  $\delta \in (0,1)$ is a parameter governing the probability that the  \Cref{alg:bpt-det-va-edgecol} will fail to be a correct deterministic algorithm, if its advice is chosen randomly; if one just wants a deterministic algorithm, setting $\delta = 1/2$ suffices. If $\Delta \le \log\frac{n \Delta}{\delta}$, use the simple greedy algorithm (\Cref{alg:simple-greedy}). Otherwise, use \Cref{alg:bpt-det-va-edgecol}.
  
  \begin{algorithm}[htb!]
    \caption{Deterministic algorithm for $O(\Delta)$ edge coloring for (adversarial) one sided vertex arrival bipartite streams, using $\tO(n)$ bits of advice\label{alg:bpt-det-va-edgecol}}
    \begin{algorithmic}[1]
    \Statex \textbf{Input}: Stream of vertex arrivals for $n$-vertex graph $G=(A \sqcup B,E)$ of max degree $\Delta$
      \Statex
      \Statex $\delta \in (0,1)$ is a parameter so that, if the advice is chosen randomly, it will work for all inputs with probability $\ge 1 - \delta$
      \Statex
      \Statex \underline{\textbf{Initialize}:}
      \Statex Let $C = 2^{18} \Delta$. % also force power of two, for bounded indep case?
      \Statex Let $s = \ceil{2^{18} \log\frac{n \Delta}{\delta}} $.
      \Statex Advice: $(\sigma_v)_{v \in B}$, where each $\sigma_v$ is a permutation over $[C]$. If chosen randomly, each is $(\epsilon,s)$-wise independent for $\epsilon \le C^{-s - 1}$
      
      \For{$v \in B$}
        \State $b_v \gets 1$.
        \State $Q_v \gets [s]$.
      \EndFor
      
    \Statex \underline{\textbf{Process}(vertex $x$ with multiset $M_x$ of edges to $B$)}
      \Statex Let $d_{x,y}$ be the number of times edge $\{x,y\}$ is in $M_x$
      \For{each $y \in B$ with $d_{x,y} > 0$}
        \If{$d_{x,y} < \frac{1}{16} s$}
          \State Let $F_y = (b_y -1) s + Q_y$\label{step:bpt-det-va-edgecol-low-deg-fy}
        \Else
          \State Let $F_y = ((b_y -1) s + Q_y) \sqcup [b_y s, b_y + \ceil{\frac{64 d_{a,b}}{s}} s]$\label{step:bpt-det-va-edgecol-high-deg-fy}
        \EndIf
      \EndFor
      \State Construct bipartite graph $H$ from $M_x$ to $[C]$, edge $e \in M_x$
      is linked to all $c \in \sigma_y[F_y]$.\label{step:bpt-det-va-edgecol-match-graph}
      \State Compute an $M_x$-saturating matching $P$ of $H$. \label{step:bpt-det-va-edgecol-compute-matching}
      \For{each $e \in M_x$}
        \State Assign color $P(e)$ to $e$
        \If{$d_{x,y} < \frac{1}{16} s$}
          \State Remove $\sigma_y^{-1} - (b_y - 1)s$ from $Q_y$
        \EndIf
      \EndFor
      \For{each $y \in B$ with $d_{x,y} > 0$}
        \If{$d_{x,y} < \frac{1}{16} s$}
          \If{$|Q_y| \le s - \frac{1}{2^{17}} s$}
            \State $b_y \gets b_y + 1$ \label{step:bpt-det-va-edgecol-small-block}
            \State $Q_y \gets [s]$
          \EndIf
        \Else
          \State $Q_y \gets [s]$
          \State $b_y \gets b_y + \ceil{\frac{2 d_{x,y}}{s}} + 1$\label{step:bpt-det-va-edgecol-large-block}
        \EndIf
      \EndFor

    \end{algorithmic}
  \end{algorithm}
  
  This algorithm maintains, for each vertex $v\in B$, two variables $b_v$ and $Q_v$ that
  indicate which colors in $[C]$ are certainly available for that vertex. It ensures that
  that none of the colors in the set $\Xi_v = \{\sigma_v[(b_v - 1) s)] : i \in Q_v\} \sqcup \{\sigma_v[j] : j > b_v s\}$ have been used. When a new vertex $x$ in $A$ arrives, along with the multiset $M_x$
  of edges adjacent to it, the algorithm selects a set $F_y$ indicating candidate colors $\sigma_y[F_y]$ for each $y$ adjacent to $x$, and computes a matching between the edges in $M_x$ and the
  set of all colors, allowing each edge in $M_x$ only the colors corresponding to the edges endpoint in $B$. This matching ensures that all edges incident to $x$ receive different colors;
  and for all $y \in B$, the use of the set $F_y$ to constrain the set of candidate colors
  to a subset of $\Xi_v$ ensures that all incident to $y$ receive different colors.
  
  For a given vertex $y$, as the algorithm runs, $b_y$ will be increased, by either Line \ref{step:bpt-det-va-edgecol-small-block} or \ref{step:bpt-det-va-edgecol-large-block}. Line \ref{step:bpt-det-va-edgecol-small-block} only triggers when $|Q_y| \le s - \frac{1}{2^{17}} s$, which requires that vertex $y$ has received $\ge \frac{1}{2^{17}} s$ incident edges since the last time $b_y$ was increased. Since there will be at most $\Delta$ edges incident to $y$, the total increase to $b_y$ from this line over the course of the algorithm will be $\le \Delta / (\frac{1}{2^{17}} s) = 2^{17} \Delta / s$. On the other hand, Line \ref{step:bpt-det-va-edgecol-small-block} only triggers when $d_{x,y}\ge \frac{1}{16} s$, and then increases $b_y$ by $\ceil{64 d_{x,y} / s} + 1$. Since $\sum_{x \in A} d_{x,y} \le \Delta$,
  \begin{align*}
    \sum_{x \in A : d_{x,y} > \frac{1}{16} s} \left(\ceil{\frac{64 d_{x,y}}{s}} + 1\right)
    \le \sum_{x \in A : d_{x,y} > \frac{1}{16} s} \left(\frac{64 d_{x,y}}{s} + 2\right) \le \sum_{x \in A : d_{x,y} > \frac{1}{16} s} \frac{96 d_{x,y}}{s}  \le \frac{96 \Delta}{s}
  \end{align*}
  Thus, the total increase in $b_y$ will be $\le (2^{17}+96) \Delta / s$, and $b_y$ will always be $\le (2^{17}+97) \Delta / s$. Looking at the construction
  of the set $F_y$ on Lines \ref{step:bpt-det-va-edgecol-low-deg-fy} and \ref{step:bpt-det-va-edgecol-high-deg-fy}, we see that it will only ever contain elements which are $\le (2^{17}+98) \Delta$. Since $C = 2^{18} \Delta \ge (2^{17}+98) \Delta$, it follows that computing $\sigma_y[F_y]$ will never index out of range.
  
  The only remaining way this algorithm could fail is if Line \ref{step:bpt-det-va-edgecol-compute-matching}
  were to report that no $M_x$-saturating matching exists. We will show that, if the 
  $(\sigma_v)_{v \in B}$ are drawn from $(\epsilon,s)$-wise independent distributions,
  then probability $\ge 1 - \delta$, for all possible
  sets $M_x$ and combinations of "free slots", $(F_y)_{y : d_{x,y} >0}$, Hall's condition will
  hold on the graph $H$ constructed on Line \ref{step:bpt-det-va-edgecol-match-graph}.
  
  Since whether the constructed graph $H$ has a matching does not depend on the value of $x$, only on the number of edges arriving at a given $y \in B$, we do not need to take a union bound over all possible sets $M_x$. Instead, define a configuration by a tuple $(S,(d_y)_{y \in S},(b_y)_{y \in S},(Q_y)_{y \in S})$. The set $S$ gives the neighborhood of $x$, and for each $y \in S$ we set $d_y = d_{x,y}$. The values $b_y,Q_y$ match the values from the algorithm at the time $M_x$ arrives. Note that the set $F_y$ is a function of $b_y$, $Q_y$, and $d_y$, and for fixed $b_y,Q_y$ is monotone increasing as a function of $d_y$. We do \emph{not} need any extra cases to handle Hall's condition for subsets of $M_x$; consider any subset $M_x' \subseteq M_x$, and let $S',d_{y}',F_y'$ correspond to $M_x'$. Then because $F_y' \subseteq F_y$ for each $y \in S'$, if Hall's condition holds for the configuration $(S',(d_y')_{y \in S'},(b_y)_{y \in S'},(Q_y)_{y \in S'})$, then
  \begin{align*}
    \left|\bigcup_{y \in S'} F_{y}\right| \ge \left|\bigcup_{y \in S'} F_{y}'\right| \ge \sum_{y \in S'} d_y'
  \end{align*}
  which implies that Hall's condition also holds for the subset $M_x'$ within the bipartite graph $H$ constructed for $M_x$.

%   While a subset $T$ of vertices on the left ($M_x$) side of $H$ does not exactly correspond to a configuration whose vector $(d_{x,y})_{y \in B}$ is smaller, if we let $R = \{y \in B : \{x,y\} \in T\}$, then Hall's condition will hold for $T$ because it holds for the original configuration restricted to vertices in $R$ (i.e, $(R,(d_y)_{y \in R},(b_y)_{y \in R},(Q_y)_{y \in R}$), and $T$ does not induce greater demand on any vertex in $R$ than did $M_x$.
  
  Let $d = \sum_{y \in S} d_y$. If the permutations $(\sigma_y)_{y\in B}$ were each chosen uniformly at random, it would be straightforward to prove that Hall's condition fails for each configuration $(S,(d_y)_{y \in S},(b_y)_{y \in S},(Q_y)_{y \in S}$ with probability $\le e^d (d/C)^{\Theta(|S| s + d)}$, after which a union bound over configurations gives a $\le \delta$ total failure probability. However, because we assume the $(\sigma_y)_{y\in B}$ are only $(\epsilon,s)$-wise independent, we will need a more precise argument.
  
  Each configuration $(S,(d_y)_{y \in S},(b_y)_{y \in S},(Q_y)_{y \in S}$ can be split into $O(\log \Delta)$ different "level configurations". Let $\tau = s / 16$; then for each vertex $y$, if $d_y < \tau$, the algorithm will choose $F_y$ using Line \ref{step:bpt-det-va-edgecol-low-deg-fy} to be a subset of size $\le s$ and $\ge s (1 - 2^{-17}) \ge s / 2$; and if $d_y \ge \tau$, the algorithm will choose $F_y$ using Line \ref{step:bpt-det-va-edgecol-high-deg-fy}, to be a subset of size $\ge 64 d_y$. Define $L_0 = \{y \in S : d_y < \tau\}$. For each $\ell \in \{1,\ldots,\lambda\}$, for $\lambda = \ceil{\log{\Delta}}$ let $L_\ell = \{y \in S : 2^{\ell - 1} \tau \le d_y < 2^{\ell} \tau\}$. Also write $L_{> \ell} = \bigcup_{j > \ell} L_j$. Within each "level", the values of $d_y$ are either all small ($< \tau$), or all within a factor $2$ of each other. We will show that with high probability, the following two conditions hold:
  \begin{align}
    \forall \ell \ge 1, \forall (L_\lambda,\ldots,L_{\ell+1}) \text{ where } \max_{i > \ell} {(3/2)^{i - \ell} |L_i|} \le |L_\ell|,& \forall {d_y,b_y,Q_y \text{ for } y \in \bigsqcup_{j \ge \ell} L_\ell} : \nonumber\\
    & \left| \left(\bigcup_{y \in L_\ell} \sigma_y[F_y]\right) \setminus \left(\bigcup_{y \in L_{> \ell}} \sigma_y[F_y]\right) \right| \ge 8 \sum_{y \in L_0} d_y \label{eq:hall-abundant} \\
    \forall (L_\lambda,\ldots,L_1) \text{ where } \max_{i > 0} {(3/2)^{i} |L_i|} \le |L_0|, & \forall {d_y,b_y,Q_y \text{ for } y \in \bigsqcup_{j \ge 0} L_\ell} : \nonumber \\
    & \left| \left(\bigcup_{y \in L_0} \sigma_y[F_y]\right) \setminus \left(\bigcup_{y \in L_{> 0}} \sigma_y[F_y]\right) \right| \ge \sum_{y \in L_0} d_y \label{eq:hall-end-ok}
  \end{align}
  Then for the specific configuration $(S,(d_y)_{y \in S},(b_y)_{y \in S},(Q_y)_{y \in S})$, let $A \subseteq \{0,\ldots,\lambda\}$ contain all $\ell$ for which $|L_\ell| \ge \max_{i > \ell} (3/2)^{i-\ell}|L_i|$. Because the condition of Eq. \ref{eq:hall-abundant} holds for all $i \in A$ with $i > 0$, these "levels" of the configuration are associated with enough entries of $C$ that they "pay for" all levels with smaller degrees that also do not have many more vertices. Level $L_0$ pays for itself if $0 \in A$, by Eq. \ref{eq:hall-end-ok}. Formally, we have:
  \begin{align*}
    |\bigcup_{y \in S} \sigma_y[F_y]| &= \sum_{i \in \{0,\ldots,\lambda \}} | (\bigcup_{y \in L_i} \sigma_y[F_y]) \setminus (\bigcup_{y \in L_{> i}} \sigma_y[F_y]) |  \\
      &\ge \sum_{i \in A} | (\bigcup_{y \in L_i} \sigma_y[F_y]) \setminus (\bigcup_{y \in L_{> i}} \sigma_y[F_y]) | \\
      &\ge 1_{0 \in A} \left(\sum_{y \in L_0} d_y\right) + \sum_{i \in A \setminus \{0\}} 8 \sum_{y \in L_i} d_y \\
      &\ge 1_{0 \in A} \left(\sum_{y \in L_0} d_y\right) + \sum_{i \in A \setminus \{0\}} 8 \cdot 2^{i - 1} \tau |L_i| \\
      &= 1_{0 \in A} \left(\sum_{y \in L_0} d_y\right) + \sum_{i \in A \setminus \{0\}} 4 \cdot 2^{i} \tau |L_i| \\
      &\ge 1_{0 \in A} \left(\sum_{y \in L_0} d_y\right) + \sum_{i \in A \setminus \{0\}} \sum_{j \le i} \left(\frac{3}{4}\right)^{i - j} \left(2^{i} \tau |L_i|\right)  \\
      &\ge 1_{0 \in A} \left(\sum_{y \in L_0} d_y\right) + \sum_{i \in A \setminus \{0\}} \sum_{j \le i : |L_j| \le (3/2)^{i - j}|L_i|}  2^{j} \tau \left(\frac{3}{2}\right)^{i - j} |L_i|  \\
      &\ge 1_{0 \in A} \left(\sum_{y \in L_0} d_y\right) + \sum_{i \in A \setminus \{0\}} \sum_{j \le i : |L_j| \le (3/2)^{i - j}|L_i|}  2^{j} \tau |L_j| \\
      &\ge 1_{0 \in A} \left(\sum_{y \in L_0} d_y\right) + \sum_{i \in A \setminus \{0\}} \sum_{j \le i : |L_j| \le (3/2)^{i - j}|L_i|}  \sum_{y \in L_j} d_y  \\
      &=  \sum_{i \in \{0,\ldots,\lambda\}} \sum_{y \in L_j} d_y = \sum_{y \in S} d_y
  \end{align*}
  
  We first observe that Eq. \ref{eq:hall-abundant} matches the conditions for \Cref{lem:approx-indep-large-anticoncentration}. Specifically, if $y \in L_\ell$ for $\ell > 1$, then $d_y \ge s / 16$, and we have both $|F_y| \ge s$ and $|F_y| \ge 64 d_y$. Also, since the $d_y \in [2^{\ell-1} \tau, 2^\ell \tau)$, we will have $\max_{y \in L_\ell} |F_y| \le 2 \min_{y \in L_\ell} F_y$. Letting $X = \bigcup_{y \in L_{> i}} \sigma_y[F_y]$, we have $|X| \le \sum_{y \in L_{> i}} |F_y| \le \sum_{y \in L_{> \ell}} (2s + 64 d_y) \le 96 \Delta \le \frac{1}{8} C$ . Similarly, $\sum_{y \in L_{\ell}} \sigma_y[F_y] \le 96 \Delta \le \frac{1}{2} C$. Thus, \Cref{lem:approx-indep-large-anticoncentration} applies, and gives an $\exp(- O(s |L_\ell|))$ upper bound on the probability that $|\bigcup_{y \in L_\ell} \sigma_y[F_y] \setminus X| \ge 4 \sum_{y \in L_0} d_y$.
  
  For Eq. \ref{eq:hall-end-ok}, if $\sum_{y \in L_0}{|F_y|} \ge \frac{1}{2} C$, we apply \Cref{lem:approx-indep-small-anticoncentration}. As argued above, the set $X = \bigcup_{y \in L_{> 0}} \sigma_y[F_y]$ will have size $\le \frac{1}{8} C$. If the bad event in \Cref{lem:approx-indep-small-anticoncentration} does not hold, then the condition in Eq. \ref{eq:hall-end-ok} will, since $\sum_{y \in L_0} d_y \le \Delta \le \frac{1}{16} C$. On the other hand, if $\sum_{y \in L_0}{|F_y|} < \frac{1}{2} C$, we apply \Cref{lem:approx-indep-large-anticoncentration}; this works because we have $|F_y| \ge (1 - 2^{-17}) s \ge \frac{1}{2} s$.

  Since \Cref{lem:approx-indep-large-anticoncentration}  and \Cref{lem:approx-indep-small-anticoncentration} both ensure a $\exp( - \Omega(s |L_i|))$ type upper bound for the probability of conditions from Eqs. \ref{eq:hall-abundant} and \ref{eq:hall-end-ok}, we can bound the probability that none of the individual event fails in a single sum. Since for each $y$, $Q_y$ is refreshed after at least $s/2^{17}$ elements are removed from it, there are only $\sum_{i=0}^{\ceil{s/2^{17}}} \binom{s}{i} \le 2^{s H(1/2^{16})} \le \exp(s / 2^{11})$ possible values for $Q_y$; here $H$ is the binary entropy function.
  \begin{align*}
    \Pr &[\text{Eqs. \ref{eq:hall-abundant} and \ref{eq:hall-end-ok} hold}]  \\
      &\le \sum_{\ell \in \{0,\ldots,\lambda\}} \sum_{\underset{\text{where $|L_\ell| = w$}}{w \in \{1,\ldots,\Delta\}}} \sum_{\underset{\text{all $|L_j| \le (2/3)^{j - \ell} w$}}{L_\lambda,\ldots,L_\ell}} \sum_{(d_y,b_y,Q_y)_{y \in \bigsqcup_{j \ge \ell} L_j}} \exp(-  s w / 2^9) \\
      &\le \sum_{\ell \in \{0,\ldots,\lambda\}} \sum_{\underset{\text{where $|L_\ell| = w$}}{w \in \{1,\ldots,\Delta\}}} \prod_{j \ge \ell} \left((n+1) C^2 \exp(s/2^{11})\right)^{(2/3)^{j - \ell} w} \exp(-  s w / 2^9) \\
      &\le \sum_{\ell \in \{0,\ldots,\lambda\}} \sum_{\underset{\text{where $|L_\ell| = w$}}{w \in \{1,\ldots,\Delta\}}} \left((n+1) C^2 \exp(s / 2^{11}) \right)^{3 w} \exp(- s w / 2^9)  \\ 
      &\le \sum_{\ell \in \{0,\ldots,\lambda\}} \sum_{\underset{\text{where $|L_\ell| = w$}}{w \in \{1,\ldots,\Delta\}}} \left(n C\right)^{6 w} \exp(- s w / 2^{11}) \qquad  \text{since $(n+1) \le n^2$ and $2^{-9} - 3 \cdot 2^{-11} = 2^{-11}$} \\ 
      &\le \sum_{\ell \in \{0,\ldots,\lambda\}} \sum_{\underset{\text{where $|L_\ell| = w$}}{w \in \{1,\ldots,\Delta\}}} \exp(- s w / 2^{12}) \qquad \text{since $s \ge 6 \cdot 2^{12} (18 + \log(n \Delta)) \ge 6 \cdot 2^{12} \ln(n C)$} \\
      &\le (\lambda + 1) \Delta \exp(-s / 2^{12}) \le \delta \qquad \text{since $s \ge 2^{13} \log(\Delta/\delta) \ge 2^{13} \ln(\Delta/\delta)$}
  \end{align*}
  Thus,
  \begin{align*}
    \Pr[\text{any configuration fails Hall's condition}]  \le \delta \,.
  \end{align*}

\end{proof}

If the $(\epsilon,s)$-wise random permutations over $[C]$ are constructed using \Cref{lem:fast-permutations} (assuming $\Delta$ is a power of two), then the total number of bits of randomness needed to sample advice for the algorithm will be $O(n s (\log C)^4 \log\frac{1}{\epsilon}) = O(n s^2 (\log C)^5) = O\left(n \left(\log \frac{n\Delta}{\delta}\right)^2 (\log \Delta)^5 \right) $.
% 
% In the above proof, the only place where the fact that the $\sigma_y$ are chosen uniformly at random from the set of all permutations are used is Eq. \ref{eq:use-of-rand-perm}, where
% it is argued that for any fixed $F_y$, $\sigma_y[F_y]$ will be uniformly random. If we specialize
% \Cref{alg:bpt-det-va-edgecol} to only work with simple graphs, then the algorithm will only
% ever construct sets $F_y$ using Line \ref{step:bpt-det-va-edgecol-low-deg-fy}, ensuring that
% $F_y \le s$. Consequently, for the proof to work we only require that $\sigma[F_y]$ is s-wise independent, so instead of making the $\{\sigma_y\}_{y\in B}$ be uniformly random permutations, we could instead use $(\frac{1}{C^s},s)$-wise independent permutations from constructed by \Cref{lem:fast-permutations},
% using only $O((\log n)^5 \log \frac{n}{\delta})$ bits of randomnness each.

% We suspect, but have not proven, that when the advice is chosen randomly from distributions of $O(\log n)$-wise almost independent permutations, with high probability \Cref{alg:bpt-det-va-edgecol} will still work for all multigraphs.

Combining \Cref{lem:cvt-gen-to-bpt} with \Cref{lem:online-bpt-det-va-edgecol}, we immediately get the following theorem.

\begin{theorem}[Formal version of \Cref{thm:online-det-va-edgecol}]\label{thm:online-det-va-formal}
  There is a deterministic online $O(\Delta)$-edge coloring algorithm for vertex arrival streams over multigraphs using $O(n \log (n\Delta))$ bits of space, using $\tO(n)$ bits of advice. (By picking a uniformly random advice string, the same algorithm can alternatively be used as a robust algorithm with $1/\poly(n)$ error; the advice can also be computed in exponential time.) 
  \end{theorem}

\section{Edge coloring on edge arrival streams}

First we prove the general version of \Cref{lem:space-color-tradeoff}.

\begin{lemma}[Generalized \Cref{lem:space-color-tradeoff}]\label{lem:space-color-tradeoff-formal}
  Let $f, g$ be functions from $\NN \mapsto \NN$.
 Given a streaming algorithm $\cA$ for $g(\Delta)$-coloring over edge arrival streams on multigraphs of max degree $\Delta$, using $f(N,\Delta)$ bits of space, for any positive integer $s$, there is a streaming algorithm $\cB$ for $(g(s \Delta) + s\Delta)$-coloring edge arrival streams for multigraphs of max degree $\Delta$, using $f(N / s, s \Delta) + O(n \log \Delta)$ bits of space.
\end{lemma}

\begin{proof} Pseudocode for algorithm $\cB$ is given by \Cref{alg:space-color-exch}.

  \begin{algorithm}[htb!]
    \caption{Adapting edge coloring algorithm $\cB$ to use more colors and less space, with parameter $s$ \label{alg:space-color-exch}}
    \begin{algorithmic}[1]
    \Statex \textbf{Input}: Stream of edge arrivals for $n$-vertex graph $G=(V,E)$
      \Statex Assume $V = [n]$
      \Statex
      \Statex \underline{\textbf{Initialize}:}
      \Statex Let $\chi : K_s \mapsto [s]$ give an $s$-edge coloring of $K_s$.\footnotemark
      \State $A \gets $ instance of $\cA(\ceil{n/s}, \Delta s)$.
      \For{$v \in [n]$}
          \State $d_v \gets 0$
      \EndFor
      
    \Statex \underline{\textbf{Process}(edge $\{x,y\}$) $\rightarrow$ \textbf{color}}
      \State $d_x \gets d_x + 1$
      \State $d_y \gets d_y + 1$

      \If{ $\ceil{x / s} = \ceil{y / s}$}
        \State Let $c \gets \Delta \cdot (\chi(\{x \bmod s, y \bmod s\}) - 1) + d_{\min(x,y)}$ 
        
        \State \textbf{return} color $(0, c)$
      \EndIf
      
      \State Let $c \gets A.\textsc{Process}( {\ceil{x/s}, \ceil{y/s}} )$
      \State \textbf{return} color $(1,c)$
    \end{algorithmic}
  \end{algorithm}
  
  This algorithm partitions the set of all vertices into sets $S_1,\ldots,S_{\ceil{n/\delta}}$,
  where set $S_i$ contains the $s$ vertices $\{s (i - 1) + 1, \ldots, s i - 1, s i\}$.It provides the nested algorithm instance $A$ with the (non-loop) edges in the graph $H$ formed by contracting these sets. Edges entirely inside one of the $S_i$ are colored using a separate set of $\Delta s$ colors.

  As the total number of edges incident on a set of $s$ vertices in $G$ is $\le \Delta s$,
  the maximum degree of $H$ will also be $\le \Delta s$. Since instance $A$ is guaranteed
  to correctly edge color all multigraphs on $[\ceil{n / s}]$ of maximum degree $\le \Delta s$,
  no two edges adjacent to a vertex in $H$ will be assigned the same color. Consequently, the
  edges from each individual vertex $v \in S_i$ to vertices outside $S-I$ will all be 
  given different colors.
  
  Consider one of the vertex sets $S_i$; a given edge $\{x,y\}$ with $x,y\in S_i$
  will be assigned a color which, due to the use of $\chi$ to partition edges, will differ from the colors assigned to all other edge
  types between vertices in $S_i$; and if the edge $\{x,y\}$ was processed in the past,
  this time will assign a different color since $d_{\min(x,y)}$ has been increased since then.
  
  The algorithm will require $f(\ceil{N/s}, \Delta s)$ bits of space to store $A$, and
  $n \log \Delta$ bits of state to keep track of all vertex degrees. The total number of
  colors used will be $g(s \Delta) + s \Delta$; if $g(x) = O(x)$, this will be $O(s \Delta)$.
  
  % typographical hack to get footnote to work inside algorithm. This should be moved to match the float
  \footnotetext{While it is possible to implement this more efficiently, this function can also be evaluated by running the Misra-Gries algorithm in $O(s^3)$ time..\cite{MisraG92}}
\end{proof}

% Like URN4-37
\begin{lemma}\label{lem:va-to-ea-conversion}
  Given a streaming algorithm $\cA$ for $O(\Delta)$ edge coloring for one-sided vertex
  arrival streams over bipartite multigraphs using $\le f(n, \Delta)$ space, we
  can construct a streaming algorithm $\cB$ for $O(\Delta)$
  edge coloring of edge arrival streams over bipartite multigraphs using 
  $O(\sqrt{\Delta} f(n, O(\sqrt{\Delta})) + n \sqrt{\Delta} (\log n \Delta) \log(n/\delta))$ bits of space. The new streaming algorithm
  $\cB$ is randomized, runs in polynomial time, and has additional $\le \delta$ probability of error,
  even if the input stream is adaptively generated.
\end{lemma}

\begin{proof}[Proof of \Cref{lem:va-to-ea-conversion}]
  The W-streaming edge-arrival algorithm is given by \Cref{alg:va-to-ea-conversion}. The algorithm uses $s = O(\sqrt{\Delta})$ instances of $\cA$. This algorithm maintains a pool $P$ of edges, and whenever it receives a new edge it adds it to the pool. Edges with high multiplicity ($\tOmega(\sqrt{\Delta}/\log(n^2/\delta))$ in $P$ are moved to a different pool $L$; since there are not many of this type, they can be stored using only $\tO(n\sqrt{\Delta})$ space. When a vertex $v$ reaches a high degree ($\ge\sqrt{\Delta}$) in the pool, it and its incident edges are removed from $P$ and assigned to a random instance of $\cA$ which has not yet received $v$. At the end of the stream, all edges still stored in either $P$ or $L$ are colored.
  
  \begin{algorithm}[htb!]
    \caption{W-streaming algorithm for $O(\Delta)$ edge coloring on edge-arrival stream given black-box access to algorithm $\cA$ for $C \Delta$ edge coloring on vertex-arrival stream\label{alg:va-to-ea-conversion}}
    \begin{algorithmic}[1]
    \Statex \textbf{Input}: Stream of edge arrivals for $n$-vertex graph $G=(A \sqcup B,E)$
      \Statex
      \Statex \underline{\textbf{Initialize}:}
      \State Let $s = 2 \ceil{\sqrt{\Delta}}$ % number of sketches
      \State Let $\tau = \floor{\frac{\sqrt{\Delta}}{9 \ln(n/\delta)}}$ % number of sketches
      
      \State $P \gets \emptyset$ is a multiset of edges -- used to cache all arriving edges
      \State $L \gets \emptyset$ is a multiset of edges -- used to efficiently store certain edge types which have high multiplicity % O(log nDelta + Delta) space used per edge
      
      \For{$i \in [s]$}
        \State $\cI^{(i)} \gets$ instance of algorithm $\cA$ for graphs of max degree $\ceil{4 \Delta / s}$; this will use $C \ceil{4 \Delta / s}$ colors
        \State $x^{(i)} \gets [0,\ldots,0] \in \{0,1\}^A$, tracks for which vertices $w$ in $A$ the instance $\cI^{(i)}$ has received $(w,M_w)$
      \EndFor
      
    \Statex \underline{\textbf{Process}(edge $\{x,y\}$)}
      \State $P \gets P \cup \{\{x,y\}\}$.
      
      \If{edge $\{x,y\}$ has multiplicity > $\tau$ in $P$}
        \State Remove all copies of $\{x,y\}$ from $P$, and add them to $L$\label{step:va-to-ea-conversion-hide-highmult}
        \State \textbf{return}
      \EndIf
      
      \If{$\exists v \in A$ with degree $\ge \ceil{\sqrt{\Delta}}$ in $P$}
          \State Pick random $i$ from $\{ j \in [s]: x^{(j)}_v = 0\}$\label{step:va-to-ea-conversion-pick-star}
          \State $x^{(i)}_v \gets 1$
          \State Let $M_v$ be edges incident on $v$ in $P$
          \State Send $(v,M_v)$ to $\cI^{(i)}$ to be colored\label{step:va-to-ea-conversion-send-to-va}
          \State Remove $M_v$ from $P$
      \EndIf
      
    \Statex \underline{\textbf{End of Stream}}
      \State Color edges in $P \cup L$ greedily using an independent set of $2\Delta - 1$ colors
    \end{algorithmic}
  \end{algorithm}
  
  \Cref{alg:va-to-ea-conversion} requires $s f(n)$ bits of space to store the instances $\cI^{(1)},\ldots,\cI^{(s)}$, and $s n$ bits to keep track of the vectors $x^{(1)},\ldots,x^{(s)}$. Since the edges adjacent to a vertex in $A$ are removed from $P$ as soon as it reaches degree $\ceil{\sqrt{\Delta}}$, the total number of edges in $P$, counting multiplicity, will be $\le |A| (\ceil{\sqrt{\Delta}} - 1) = O(n\sqrt{\Delta})$. Thus, $P$ can be stored using $O(n \sqrt{\Delta} \log(n^2 / \delta))$ bits of space. Finally, since $L$ receives only edges whose multiplicity was at least $\sqrt{\Delta} / \log n$, it will contain at most $(n \Delta / 2) / (\sqrt{\Delta} / \log n) = n \sqrt{\Delta} / (2 \log n)$ distinct edges; keeping track of them and their multiplicity can be done in $O(n \sqrt{\Delta} \log (n \Delta) / \log n)$ space. In total, \Cref{alg:va-to-ea-conversion}  will require $O(\sqrt{\Delta} (f(n) + n (\log (n \Delta)) \log(n/\delta)))$ bits of space in total.
  
  The total number of colors used is $2 \ceil{\sqrt{\Delta}} \cdot C \ceil{4 \Delta / s} + (2 \Delta - 1) = O(\Delta)$.
  
  Because \Cref{alg:va-to-ea-conversion} only sends a star around a vertex $v$ to an instance $\cI^{(i)}$ (Line \ref{step:va-to-ea-conversion-send-to-va}) when the vertex $v$ has degree $= \ceil{\sqrt{\Delta}}$ in $P$, the maximum degree of arriving vertices that any instance of $\cA$ will process will be $\ceil{\sqrt{\Delta}}$. However, it is still possible that for some sketch $\cI^{(i)}$, a vertex $v \in B$ will receive a too many edges from vertices in $A$ that the sketch $\cI^{(i)}$ receives later.
    
  For some pair $i \in [s]$, $z \in B$, we will show that sketch $\cI^{(i)}$ receives $\le 4 \Delta / s$ edges (counting multiplicity) for $z$, with $\ge \frac{\delta}{n^2}$ probability. Let $X_1,\ldots,X_\Delta$ be random variables, where $X_i$ is the number of edges that are sent to $\cI^{(i)}$ when the $j$th star adjacent to $z$ is removed. If the stream ends before an $j$th star is removed, then $X_j = 0$. Because Line \ref{step:va-to-ea-conversion-hide-highmult} removes all edges with multiplicity $> \tau$ in $P$, $z$ will have at most $\tau$ edges between it and the center of the $j$th star, so $X_j \le \tau$. Furthermore, at the time the $j$th star is selected, the algorithm makes a random decision on Line \ref{step:va-to-ea-conversion-pick-star} to choose which sketch will receive it. Because $s = 2\ceil{\sqrt{\Delta}}$, and each star has root degree only $\ceil{\sqrt{\Delta}}$, there will always be $\ge s/2$ instances that have not received a given vertex as the root of a star, so the probability that $\cI^{(i)}$ will receive the $j$th star is $\le \frac{2}{s}$. Thus $\EE[X_j | X_1,\ldots,X_{j-1}] \le 2 \tau / s$. This bound holds even if the input stream is produced by an adaptive adversary. Since the degree of $z$ will be less than $\le \Delta$ at the end of the stream, we also have $\EE[\sum_{j \in \Delta} X_j] \le 2 \Delta / s$.
  
  We now apply the multiplicative (Chernoff-like) form of Azuma's inequality, on the $[0,1]$ random variables $Y_1,\ldots,Y_\Delta$, defined by $Y_j := X_j / \tau$. Let $\alpha = \EE[\sum_{j \in \Delta} Y_j]$.
  \begin{align*}
    \Pr\left[\sum_{i \in [\Delta]} X_j \ge 4 \Delta / s\right] &= \Pr\left[ \sum_{i \in [\Delta]} Y_j \ge \frac{4 \Delta}{s \tau} \right] = \Pr\left[ \sum_{i \in [\Delta]} Y_j \ge \left(1 + \left(\frac{4 \Delta}{s \tau \alpha} - 1\right)\right) \alpha \right] \\
      &\le \exp\left( - \frac{1}{3} \left(\frac{4 \Delta}{s \tau \alpha} - 1\right) \cdot \alpha \right) \qquad\qquad \text{since $\frac{4 \Delta}{s \tau \alpha} - 1 > 1$} \\
      &\le \exp\left( - \frac{1}{3} \frac{2 \Delta}{s \tau \alpha}\cdot \alpha \right) =  \exp\left( - \frac{2}{3} \frac{\Delta}{s \tau} \right) \\
      &= \exp\left( - \frac{2 \Delta}{3 \cdot 2 \ceil{\sqrt{\Delta}} \floor{\sqrt{\Delta} / (9 \ln(n/\delta))}} \right) \\
      &\le \exp\left( - 2 \ln(n/\delta)\right) \le \frac{\delta}{n^2} \qquad\qquad \text{since $2 \ceil{\sqrt{\Delta}} \le 3 \sqrt{\Delta}$ and $1/\floor{x} \ge 1/x$}
  \end{align*}

  By a union bound over all $\le n$ vertices $v \in B$, and all $\le n$ instances in $\{\cI^{(j)}\}_{j\in [s]}$, we have that the total probability of any vertex $v$ in an instance $\cI^{(i)}$ receiving more than $4 \Delta / s$ edges is $\le \delta$.
\end{proof}

% If we restrict ourselves to \emph{simple} graphs, algorithm \Cref{alg:va-to-ea-conversion} can likely be derandomized to work deterministically. 

Combining \Cref{lem:va-to-ea-conversion} with \Cref{lem:online-bpt-rand-va-edgecol}, and then applying \Cref{lem:cvt-gen-to-bpt-ea} proves the following.

\begin{theorem}[Formal version of \Cref{thm:Wstream-rand-ea-edgecol}]\label{thm:Wstream-rand-ea-formal}
  There is a randomized W-streaming algorithm for $O(\Delta)$ edge coloring on edge arrival streams for multigraphs which uses $O(n \sqrt{\Delta} (\log (n \Delta))^2)$ bits of space, with error $\le 1/\poly(n)$ against any adaptive adversary. The algorithm also requires $\tO(n\Delta)$ bits of oracle randomness.
\end{theorem}

The following online edge coloring algorithms will both use the same core primitive; a pool of
random colors, which is periodically refreshed, along with data to keep track of which colors
in the pool have been used so far. The times at which the pool are refreshed only
depend on the number of colors that were used, and not which colors where used; this property makes
the primitive easier to handle in proofs.

\begin{algorithm}[htb!]
  \caption{Storing free regions from a permutation\label{alg:ds-rand-blocks}}
  \begin{algorithmic}[1]
    \Statex \underline{$F \gets $\textbf{InitFreeTracker}($C$,$s$,$\Delta$,$\sigma$\textbf{):}} \Comment{Assume $C$,$s$,$\Delta$ are powers of two, and $\sigma$ permutation of $[C]$, and $C \ge \Delta$}
    \State $H \gets [s]$ be a subset of $[s]$
    \State $b \gets 1$ be a counter between $1$ and $C/s$
    %\State $\sigma$ is referenced
    \State \textcolor{PineGreen}{Optional: $Q \gets \emptyset$ is a set of references to objects}
%     \State Define $P_i := \sigma[[s] + (i -1) s]$ for $i=1,\ldots,C/s$ \label{step:ds-rand-blocks-set-pi}% notation: \sigma[x,....,y] vs \sigma[x + [y-x] - 1]

    \Statex
    
    \Statex \underline{\textbf{Interpreting $F$ as subset of $[C]$}}
    \State \textbf{return} $\sigma[H + (b -1) s]$
    \Statex
    
    % todo: it may be better to define this as a generic function?
    \Statex \underline{$F$\textbf{.RemoveAndUpdate(}$c$, \textcolor{PineGreen}{optional: $o$}\textbf{)}} \Comment{Requires $c \in F_v$}
    \State $H \gets H \setminus \{ \sigma^{-1}(c)  \}$
    \State \textcolor{PineGreen}{Optional: Add a reference to $o$, and store it in $Q$}
    \If{$|H| \le s - s\Delta /C$} \Comment{Switch to next block}
      \State $H \gets [s]$
      \State $b \gets b + 1$
      \State \textcolor{PineGreen}{Optional: Drop all references in $Q$ and set $Q \gets \emptyset$}
    \EndIf
  \end{algorithmic}
\end{algorithm}

We are now ready to state and prove the formal version of \Cref{thm:online-rand-ea-edgecol}.

\begin{theorem}[Formal version of \Cref{thm:online-rand-ea-edgecol}]\label{thm:online-rand-ea-formal}
Given any adversarial edge-arrival stream of a simple graph, there is a randomized algorithm for online $O(\Delta)$-edge-coloring using $O(n \sqrt{\Delta \log n})$ bits of space
and $\tO(n \sqrt{\Delta})$ oracle random bits.
\end{theorem}

% Similar to: URN4-51
\begin{proof} We will show that \Cref{alg:rand-ea-edgecol} satisfies the claims of the lemma, if $\Delta = \Omega(\log (n/\delta))$. (For smaller values of $\Delta$, fall back to \Cref{alg:simple-greedy}.) In the following argument, we shall assume that the permutations $(\sigma_v)_{v\in S}$ are $s$-wise independent. The pseudocode states $(\epsilon,s)$-wise independence, since that is attainable per \Cref{lem:fast-permutations} using only $O(s \poly(\log 1/\epsilon, \log s))$ bits of randomness per permutation. This will not affect the validity of the proof, since it at most increases the probabilities of events $H_{\{u,v\}}$ and $J_{u,v}$ defined later by $\epsilon$, which is polynomially smaller than the losses in the argument due to bounding the number of events by $n^2$ instead of $\binom{n}{2}$ or $n^2 - n$. We also assume that $\Delta$ is a power of two; if not, we can increase $\Delta$ to the nearest power of two, and the algorithm will still give an $O(\Delta)$ coloring.

  Each color tracker $F_v$ can be stored using $O(\log \Delta)$ bits for $b$, and $O(s)$ bits for $H$; thus 
  \Cref{alg:rand-ea-edgecol}  will use $O(n (s + \log \Delta)) = O(n \sqrt{\Delta \log (n/\delta)})$ bits in total.
  For $\Delta = O(\log (n/\delta))$, the \Cref{alg:simple-greedy} uses $O(n \Delta)$ bits, which is also $O(n \sqrt{\Delta \log (n/\delta)})$.

  \begin{algorithm}[htb!]
    \caption{Randomized algorithm for $O(\Delta)$ edge coloring for simple graph edge arrival streams\label{alg:rand-ea-edgecol}}
    \begin{algorithmic}[1]
    \Statex \textbf{Input}: Stream of vertex arrivals $n$-vertex graph $G=(A \sqcup B, E)$
    \Statex Assume $\Delta$ is a power of two, and $\Delta = \Omega(\log(n/\delta))$
      \Statex
      \Statex \underline{\textbf{Initialize}:}
      \State Let $C = 128 \Delta$
      \State Let $s$ be the least power of two which is $\ge 128 \sqrt{\Delta \log (n/\delta)}$
      \State Let $\cH$ be an ($\epsilon,s$)-wise independent distribution of permutations on $[C]$, with $\epsilon \le \exp(-s^2/C) \le (\delta/n)^{128}$
      % Note: construction via switching network needs to be proven: Needs extra work to check this -- i.e, for any gate configuration mapping k elements to k elements, the probability of the d*k gates involved in the configuration is exactly that expected from randomness. Then add up over all configurations.
      \For{$v \in B$}
        \State Let $\sigma_v$ be a random permutation from $\cH$
        \State $F_v \gets \textsc{InitFreeTracker}(C, s, \Delta, \sigma_v)$, without reference count tracking
      \EndFor
    \Statex 
      
    \Statex 
    \Statex\underline{\textbf{Process}(edge $\{x,y\}$) $\rightarrow$ \textbf{color}}
      \If{$F_x \cap F_y = \emptyset$}
        \State abort
      \EndIf
      
      \State Let $c$ be chosen uniformly at random from $F_x \cap F_y$.\label{step:rand-ea-edgecol-pick-color}
      \State $F_x.\textsc{RemoveAndUpdate(c)}$
      \State $F_y.\textsc{RemoveAndUpdate(c)}$
      \State \textbf{return} color $c$
    \end{algorithmic}
  \end{algorithm}
  
  For each $v \in V$, $i \in [C/s]$, write $P_{v,i}$ for the set $\sigma[[s] + (i -1) s]$ of the free region tracker $F_v$ for vertex $v$. (See \Cref{alg:ds-rand-blocks}.) Since we are assuming the $\sigma_v$ are $s$-wise independent, the set $P_{v,i}$ will be uniformly distributed over over $\binom{[C]}{s}$.
  
  Consider a fixed input stream $e_1,e_2,\ldots$, where the edges of the stream together form the simple graph $G$. Write $b_{x,\{u,v\}}$ for the value of the counter $b$ inside $F_x$ just before the algorithm processed edge $\{u,v\}$. Let $D_{\{u,v\}} := P_{u, b_{u,\{u,v\}}} \cap P_{v, b_{v,\{u,v\}}}$. Also define $M_{u,v} := \{ x : \{x,u\} \in G \land b_{u,\{x,u\}} = b_{u,\{u,v\}} \land \{x,u\} \prec \{u,v\}\}$; this is the set of vertices which were adjacent to $u$, for which the edge $\{x,u\}$ was added before $\{u,v\}$, and while the value of the counter $b$ inside $F_u$ for vertex $u$ was the same as it was at the time $\{u,v\}$ was added. This is the set of vertices whose color
  choices \emph{might} reduce the size of $F_u \cap F_v$ at the time $\{u,v\}$ is added. 
  
  We will first show that of the following two classes of $m$ events, the probability that any of the events is true is $\le \delta / 2$.
  \begin{align*}
    \forall \{u,v\} \in G: H_{\{u,v\}} & := \left\{ D_{\{u,v\}} \le \frac{1}{2} s^2/C \right\} \\
    \forall (u,v) \text{ where } \{u,v\} \in G: J_{u,v} & := \left\{ \sum_{x \in M_{u,v}} |D_{\{u,v\}} \cap P_{x,b_{x,\{x,u\}}}| \ge  2 \cdot \frac{3}{2} \frac{s^3}{C^2} \frac{\Delta s}{C} \right\} \\
  \end{align*}
  
  To bound the probability of $H_{\{u,v\}}$, we let $X_1,\ldots,X_C$ be indicator random variables where $X_i = 1$ iff $i \in P_{u, b_{u,\{u,v\}}}$.  Since the $X_i$ are negatively associated\cite{JoagDevP83}, the proof of the Chernoff bound holds, and 
  \begin{align*}
    \Pr[ |D_{\{u,v\}}| \le \frac{1}{2} s^2/C ] &= \Pr[ D_{\{u,v\}} \le \frac{1}{2} \EE[D_{\{u,v\}}] ] \\
       &= \Pr[ \sum_{i \in [P_{v, b_{v,\{u,v\}}}]} X_i \le \frac{1}{2} \EE[D_{\{u,v\}}] ] \\
       &\le \exp( - \frac{1}{8} \EE[D_{\{u,v\}}])= \exp(- \frac{s^2}{8 C} )  \\
       &\le \exp(- 16 \log\frac{n}{\delta}) \le \frac{\delta}{2 n^2}
  \end{align*}
  
  To bound the probability of the events $\{J_{u,v}\}$, we will show that $|D_{u,v}|$ is not too large w.h.p, and conditioned on that, the sum $\sum_{x \in M_{u,v}} |D_{\{u,v\}} \cap P_{x,b_{x,\{x,u\}}}|$ is not too large w.h.p. With $\{X_i\}_{i\in[C]}$ as defined above:
  \begin{align*}
    \Pr[ |D_{\{u,v\}}| \ge \frac{3}{2} s^2/C ] &= \Pr[ D_{\{u,v\}} \ge \frac{3}{2} \EE[D_{\{u,v\}}] ] \\
       &\le  \Pr[ \sum_{i \in [P_{v, b_{v,\{u,v\}}}]} X_i \ge \frac{3}{2} \EE[D_{\{u,v\}}] ] \\
       &\le \exp( - \frac{1}{10} \EE[D_{\{u,v\}}]) = \exp(- \frac{s^2}{10 C} ) \\
       &\le \exp(- \frac{128}{10} \log\frac{n}{\delta})\le \frac{\delta}{4 n^2}
  \end{align*}
  The permutations $\{\sigma_x\}_{x \in M_{u,v}}$ are independent of $\sigma_u$ and $\sigma_v$. For each $x \in M_{u,v}$, let $Y_{1,x},\ldots,Y_{C,x}$ be indicator random variables
  where $Y_{i,x}$ is 1 iff $i \in P_{x, b_{x,\{x,u\}}}$, and zero otherwise. Due to the frequency of free color buffer refreshing, $|M_{u,v}| \le s \Delta / C$; and since $|P_{x, b_{x,\{x,u\}}}| = s$, $\EE Y_{i,x} = s/C$. Since the $\{Y_{i,x}\}_{i \in D_{\{u,v\}}, x \in M_{u,v}}$ are negatively associated, we can apply a Chernoff bound. If we assume that $|D_{\{u,v\}}| \le \frac{3 s^2}{2 C}$, then we have:
  \begin{align*}
    \Pr\left[ \sum_{x \in M_{u,v}} |D_{\{u,v\}} \cap P_{x,b_{x,\{x,u\}}}| \ge 2 \frac{\Delta s^2}{C^2} \frac{3 s^2}{2 C}\right]
      &= \Pr\left[ \sum_{i \in D_{\{u,v\}}, x \in M_{u,v}} Y_{i,x} \ge 2 \frac{\Delta s^2}{C^2} \frac{3 s^2}{2 C} \right] \\
      &\le \exp\left( - \frac{1}{8} \frac{\Delta s^2}{C^2} \frac{3 s^2}{2 C} \right) \\
      &\le \exp(- 12 (\log(n/\delta))^2) \le \frac{\delta}{4 n^2}
  \end{align*}
  Thus, the probability that either $|D_{\{u,v\}}| \ge \frac{3 s^2}{2 C}$ or event $J_{u,v}$ does not hold is $\frac{\delta}{2 n^2}$.
    
  For the rest of the proof, we will consider the case where none of the events $J_{u,v}$ or $H_{\{u,v\}}$ holds; this happens with probability $\ge 1 - \delta/2$. Fix values of the $(\sigma_v)_{v\in V}$ satisfying none of the events. The only other random decisions made
  by the algorithm are the choices made on Line \ref{step:rand-ea-edgecol-pick-color}, randomly choosing the edge color $\chi_{\{u,v\}}$ for $\{u,v\}$ from $F_u \cap F_v$. We will prove by induction on
  the number of edges processed that the probability of $|F_u \cap F_v| \le \frac{1}{4} s^2 / C$ holding at the time Line \ref{step:rand-ea-edgecol-pick-color} is executed, in total over all $t$ edges so far is, $\le \delta \cdot t / (2n^2)$.
  
  To do this, we will use the following lower bound:
  \begin{align}
    |F_u \cap F_v| \ge |D_{\{u,v\}}| - \sum_{x \in M_{u,v}} W_{x,u} - \sum_{x \in M_{u,v}} W_{x,v} \label{eq:fuv-lb}
  \end{align}
  Here $W_{x,u}$ is the indicator random variable for the event that the color chosen for
  $\{x,u\}$ was in $D_{\{u,v\}}$. The lower bound overcounts the number of colors in $D_{\{u,v\}}$ that have been removed from $F_u \cap F_v$.
  
  The base case of the induction (0 edges) is immediate. Assume that we are processing
  edge $\{u,v\}$, and that all edges $\{x,y\}$ earlier in the stream, when they were processed,
  had $|F_x \cap F_y| \ge \frac{s^2}{4 C}$. For each $x \in M_{u,v}$, the color $\chi_{\{x,u\}}$ was drawn uniformly at random from \emph{some} set $F_x \cap F_u$, which we assume satisfies $|F_x \cap F_u| \ge \frac{s^2}{4 C}$. For any subset $H$ of $D_{\{x,u\}}$ of size $\frac{s^2}{4 C}$, if $\hat{\chi}$ is chosen u.a.r. from $H$, then 
  \begin{align*}
    \Pr[\hat{\chi} \in D_{\{u,v\}}] \le \frac{|H \cap D_{\{u,v\}}|}{|H|} \le \frac{4 C}{s^2} |P_{x,b_{x,\{x,u\}}} \cap D_{\{u,v\}}|
  \end{align*}
  Conditioned on the color choices of all earlier edges, we thus have $\EE W_{x,u} \le \frac{4 C}{s^2} |P_{x,b_{x,\{x,u\}}} \cap D_{\{u,v\}}|$. Thus
  \begin{align*}
    \EE\left[\sum_{x \in M_{u,v}} W_{x,u} + \sum_{x \in M_{v,u}} W_{x,v}\right] &\le \frac{4 C}{s^2} \left(\sum_{x \in M_{u,v}}  |P_{x,b_{x,\{x,u\}}} \cap D_{\{u,v\}}| + \sum_{x \in M_{v,u}}  |P_{x,b_{x,\{x,v\}}} \cap D_{\{u,v\}}| \right) \\
      &\le \frac{4 C}{s^2} \cdot 2 \frac{\Delta s^2}{C^2} \frac{3 s^2}{2 C} = 12 \frac{\Delta}{C} \frac{s^2}{C} < \frac{s^2}{8 C} & \text{since $\Delta \le C / 128$}
  \end{align*}
  
  By the multiplicative/Chernoff-like formulation of Azuma's inequality,
  \begin{align*}
    \Pr&\left[\sum_{x \in M_{u,v}} W_{x,u} + \sum_{x \in M_{v,u}} W_{x,v} \ge \frac{s^2}{4 C}\right] \\
      &\le \exp( -\frac{1}{3} \frac{s^2}{4 C}) \le \exp\left( - \frac{32}{3} \log\frac{n}{\delta} \right) \le \frac{\delta}{2 n^2}
  \end{align*}
  By a union bound over all edges, the probability that any edge $\{u,v\}$ has $|F_u \cap F_v| \le \frac{s^2}{4 C}$ is $\le \delta / 2$.
  
  We have shown that, in total, the probability of the algorithm aborting because
  $F_u \cap F_v = \emptyset$ is $\le \delta$.
\end{proof}

Algorithm \ref{alg:rand-ea-edgecol} can be generalized to produce $O(\Delta^2 / t)$ edge colorings using $\tO(n \sqrt{t})$ bits of space, by increasing the parameters $C$ and $s$ while ensuring that $s^2 / C = \Omega(\log (n/\delta))$. Then as at most $s \Delta / C$ colors are removed from each free color tracker, it will be possible to store each free color tracker using $\tO(s \Delta / C)$ bits of space. However, further adjustment would be necessary to make the algorithm use $\tO(n \sqrt{t})$ random bits. We suspect that picking $(\epsilon,O(s^2/C))$-wise independent distributions will be sufficient. As proving this would be tedious, and the following \Cref{thm:online-det-ea-edgecol} already provides a color-space tradeoff for the edge arrival setting, we do not do so.

We now introduce a technical lemma which will be useful in the proof of \Cref{thm:online-det-ea-edgecol}

\begin{lemma}\label{lem:offline-ea-coloring}
  Let $V$ be a set of size $n$, $\delta \in (0,1)$, and let $\Delta$ be a power of two, satisfying $\Delta \ge 256 \log \frac{n}{\delta}$. Define $C = 32 \Delta$, and let $s$ be the least power of two which is $\ge 512 \sqrt{\Delta \log \frac{n}{\delta}}$. Let $(\sigma_v)_{v \in V}$ be randomly chosen permutations from an $(\epsilon, s)$-wise independent family, where $\epsilon \le \exp(- s^2/C) \le (\delta/n)^{1024}$. For $i \in [C / s]$, $v \in V$, let $P_{v,i} := \sigma[s (i-1) + [s]]$.
  
  We say that the permutations $(\sigma_v)_{v \in V}$ are \emph{good} if, for all simple graphs $H$ on $V \times [C/s]$ for which, for any $u,w\in V$ and $i \in [C/s]$, there is at most one $j$ for which edge $\{(u,i),(v,j)\}$ is in $H$, and the max degree of $H$ is $\le s \Delta / C$; that the graph $H$ can be list-edge colored where edge $\{(u,i),(v,j)\}$ may only use colors in $P_{u,i} \cap P_{v,j}$.
  
  The probability that the $(\sigma_v)_{v \in V}$ are \emph{good} is $\ge 1 - \delta$.
\end{lemma}

\begin{proof}[Proof of \Cref{lem:offline-ea-coloring}]
  We will prove this in two steps. First, define a specific property U 
  that the $(\sigma_v)_{v\in V}$ should satisfy with probability $\ge 1 - \delta$; second, prove that if this property holds, then any graph $H$ can be colored.
  
  The permutations $(\sigma_v)_{v\in V}$ satisfy property U if:
  \begin{itemize}
    \item For all pairs $(u,i),(v,j) \in V \times [C/s]$, with $u \ne v$, we have $|P_{u,i} \cap P_{v,i}| \ge \frac{s^2}{2 C}$.
    \item For each $(u,i) \in V \times [C/s]$, $S \subseteq (V \setminus \{s\}) \times [C/s]$ where $|S| \le s \Delta / C$ and $S$ includes no two vertices $(v,i),(u,j)$ with $v = u$, and all $T \in \binom{P_{u,i}}{|S| - 1}$, there exists some $(x,j) \in S$ for which $|P_{x,j} \cap T| < \frac{1}{10} |P_{x,j} \cap P_{u,i}|$. (This is, in effect, a stronger version of Hall's condition).
  \end{itemize}
  
  For the first part of property $U$, it is straightforward to bound the probability that it does not hold. For a given pair $(u,i),(v,j) \in V \times [C/s]$, $u\ne v$, because the permutations are $(\epsilon,s)$-wise independent, the sets $P_{u,i}$ and $P_{v,j}$ are within $\epsilon$-total-variation distance of being uniformly random subsets of $[C]$ of size $s$, we can apply a Chernoff bound for the number of elements in $P_{u,i}$ that lie in $P_{v,j}$:
  \begin{align}
    \Pr[ P_{u,i} \cap P_{v,j} \le \frac{1}{2} \frac{s^2}{C} ] \le \exp( -\frac{1}{8}  \frac{s^2}{C} ) + \epsilon \le \exp( - 2^{10} \log (n/\delta)) + \epsilon \le \frac{\delta}{2 n^2} \label{eq:puv-size-lb}
  \end{align}
  (The additive factor $\epsilon$ accounts for the maximum difference in probabilities for this event between the case where $P_{u,i}$ is exactly uniform and the case where it is $\epsilon$-far from such.)
  
  For the second part, consider a specific combination $(u,i,S,T)$, and fix $P_{u,i}$. Then the probability that this combination violates property U is:
  \begin{align}
    \Pr\left[ \bigwedge_{(x,j) \in S} \left\{ |P_{x,j} \cap T| \ge \frac{1}{10} |P_{x,j} \cap P_{u,i}\right\} \right] \le \prod_{(x,j) \in S} \Pr_{P_{x,j}}\left[ |P_{x,j} \cap T| \ge \frac{1}{10} |P_{x,j} \cap P_{u,i} | \right] \label{eq:superhall-specific}
  \end{align}
  since the $P_{x,j} \in S$ are all independent, since $S$ contains at most one entry for each $v \in V$. Since $P_{x,j}$ is a uniformly random subset $[C]$,
  $P_{x,j} \cap P_{u,i}$ is symmetrically distributed over $P_{u,i}$. Now let $\hat{T}$ be a uniformly random element of $\binom{P_{u,i}}{s}$, and define indicator random variables $\{Y_k\}_{k \in P_{u,i}}$ so that $Y_k = 1$ iff $k \in \hat{T}$; these are negatively associated and $\EE[Y_k] = \frac{|S| - 1}{s}$. Thus, if we assume $|P_{x,j} \cap P_{u,i}| = h$:
  \begin{align*}
    \Pr_{P_{x,j}}&\left[ |P_{x,j} \cap T| \ge \frac{1}{10} |P_{x,j} \cap P_{u,i} |  \Big| |P_{x,j} \cap P_{u,i}| = h\right] \\
      &\le \Pr_{\hat{T}}\left[ |P_{x,j} \cap \hat{T}| \ge \frac{1}{10} h \Big| |P_{x,j} \cap P_{u,i}| = h \right] + \epsilon \\
      &= \Pr_{\{Y_k\}_{k \in P_{u,i}}}\left[ \sum_{k \in P_{u,i}} Y_k \ge \frac{1}{10} h \Big| |P_{x,j} \cap P_{u,i}| = h | \right] + \epsilon \\
      &\le \exp\left( -2 \left(\frac{1}{10} - \frac{|S| - 1}{s} \right)^2 h \right) + \epsilon \\
      &\le \exp\left( -2 \left(\frac{1}{10} - \frac{\Delta}{C} \right)^2 h \right) + \epsilon \le \exp\left( - h / 200 \right) + \epsilon  \qquad \text{since $C = 32 \Delta$}
  \end{align*}
  This bound is useful only if $h$ is large enough. By the law of total probability, and using the bound from Eq. \ref{eq:puv-size-lb} to handle the case where $h$ is small:
  \begin{align*}
    \Pr_{P_{x,j}}&\left[ |P_{x,j} \cap T| \ge \frac{1}{10} |P_{x,j} \cap P_{u,i} | \right] \\
      &\le \Pr_{P_{x,j}}\left[ |P_{x,j} \cap T| \ge \frac{1}{10} |P_{x,j} \cap P_{u,i} | \Big| |P_{x,j} \cap P_{u,i}| \ge \frac{s^2}{2 C} \right]\Pr\left[|P_{x,j} \cap P_{u,i}| \ge \frac{s^2}{2 C}\right]
        + \Pr\left[|P_{x,j} \cap P_{u,i}| \le \frac{s^2}{2 C}\right] \\
        &\le (\exp\left( - \frac{s^2}{400C} \right) + \epsilon)\cdot 1 + (\exp(-\frac{s^2}{8 C}) + \epsilon)
        \le 2 \exp\left( - \frac{s^2}{400C} \right) \,.
  \end{align*}
  
  Substituting this result into Eq. \ref{eq:superhall-specific} gives:
  \begin{align*}
    \Pr\left[ \bigwedge_{(x,j) \in S} \left\{ |P_{x,j} \cap T| \ge \frac{1}{10} |P_{x,j} \cap P_{u,i}\right\} \right] \le 2^{|S|} \exp\left( - |S| \frac{s^2}{400C} \right) \,.
  \end{align*}
  Taking a union bound over all $(u,i,S,T)$ tuples gives:
  \begin{align*}
    \Pr&[\text{second part of Property U fails}] \\
      &\le \sum_{(u,i) \in V \times [C/s]} \sum_{k=1}^{s \Delta / C} \sum_{\text{valid $S$ with $|S| = k$}} \sum_{T \subseteq \binom{P_{u,i}}{k-1}} 2^{k} \exp\left( - k \frac{s^2}{400C}\right) \\
      &\le n \cdot \sum_{k=1}^{s \Delta / C} \cdot \binom{n - 1}{k} \left(\frac{C}{s}\right)^k \cdot \binom{s}{k - 1} \cdot 2^{k} \exp\left( - k \frac{s^2}{400C}\right) \\
      &\le n \sum_{k=1}^{s \Delta / C} \left( 2 n C s \exp\left(- \frac{s^2}{400 C}\right) \right)^k \\
      &\le 2 n \cdot 2 n C s \exp\left(- \frac{s^2}{400 C}\right) \qquad\qquad \text{for large enough $s^2/C$} \\
      & \le 4 n^2 \cdot (32 n)^2 \exp(- \frac{2048}{100} \log(n/\delta)) \le \frac{\delta}{2} \qquad\qquad \text{since $s \le C = 32 \Delta$}
  \end{align*}
  Combining this with a union bound over Eq. \ref{eq:puv-size-lb} implies that property U fails to hold with probability $\le \delta$.
  
  For the second stage of the proof, we consider the following iterative process to color any graph $H$ satisfying the given conditions. Consider an arbitrary ordering $v_1,\ldots,v_n$ of the vertices in $V \times [C/s]$. For a given vertex $v_t$, let $A(v_t)$ be the set of vertices in $\{v_1,\ldots,v_{t-1}\}$ which are adjacent to $(v_t,i)$, and let $B(v_t,i)$ be the set of vertices in 
  $\{v_{t+1},\ldots,v_{n}\}$ which are adjacent to $(v_t,i)$. For each $z \in A(v_t) \cup B(v_t)$, define $U_{t,v_t,z}$ to be the set of colors in $P_{v_t} \cap P_z$ that were already used by edges to vertices in $A(v_t)$ just \emph{after} step $t$. The color assignment chosen will maintain the invariant W that $|U_{t,z,(v_{t'})}| \le \frac{1}{3} |P_z \cap P_{v_{t'}}|$ for all $t' > t$ and $z \in A(v_{t'})$. In other words, that when it is time to color the edges from a future vertex $(v_{t'})$ to $A(v_{t'}$, only a $\le 1/3$ fraction of the initially possible color options will have been used.
  
  Invariant W automatically holds when $t = 0$, since no edges have been colored. Say the invariant holds at time $t - 1$. Then we are guaranteed that $|U_{t-1,z,v_t}| \le \frac{1}{3} |P_z \cap P_{v_t}|$ for all $z \in A(v_t)$, and want to find color assignments for the edges from $A(v_t)$ to $z$ so that $|U_{t,z,v_t}| \le \frac{1}{3} |P_z \cap P_{v_t}|$ for all $z \in B(v_t)$. To do this, we will first pick a set $F \subseteq P_{v_t}$ that satisfies:
  \begin{align}
    \forall x \in A(v_t):&\quad |P_x \cap P_{v_t} \setminus U_{t-1,x,v_t} \setminus F| \ge \frac{1}{10} |P_x \cap P_{v_t}| \label{eq:F-backward} \\
    \forall x \in B(v_t):&\quad |(P_y \cap P_{v_t}) \setminus F| \le \frac{1}{3} |P_y \cap P_{v_t}| \label{eq:F-forward} 
  \end{align}
  That such a set $F$ exists follows by the probabilistic method; say $F$ were chosen so that each element of $P_{v_t}$ is included u.a.r with probability $\frac{7}{10}$. For $i \in P_{v,t}$, let $X_i$ be the indicator random variable for the event that $i \in F$. Then the probability of Eq. \ref{eq:F-backward} is bounded by:
  \begin{align*}
    \Pr&\left[ |(P_x \cap P_{v_t}) \setminus U_{t-1,x,v_t} \setminus F| \ge \frac{1}{10} |P_x \cap P_{v_t}| \right] \\
      &\le   \Pr\left[ |(P_x \cap P_{v_t}) \setminus U_{t-1,x,v_t} \setminus F|      \ge \frac{3}{20} |(P_x \cap P_{v_t}) \setminus U_{t-1,x,v_t}| \right] \quad \text{ since $|(P_x \cap P_{v_t}) \setminus U_{x,v}| \ge \frac{2}{3} |P_x \cap  P_{v_t}|$ } \\
      &\le \Pr\left[ \sum_{i \in (P_x \cap P_{v_t}) \setminus U_{t-1,x,v_t}} X_i \ge \frac{17}{20} |(P_x \cap P_{v_t}) \setminus U_{t-1,x,v_t}| \right] \\
      &\le \exp\left(- \frac{9}{200} |P_x \cap P_{v_t} \setminus U_{t-1,x,v_t}| \right) \qquad \text{by Chernoff bound, since $\EE$ of LHS is $\frac{14}{20} |(P_x \cap P_{v_t}) \setminus U_{t-1,x,v_t}|$} \\
      &\le \exp\left(- \frac{3}{100} |P_x \cap P_{v_t}| \right) \qquad \text{since $|(P_x \cap P_{v_t}) \setminus U_{t-1,x,v_t}| \ge \frac{2}{3} |P_x \cap P_{v_t}|$} \\
      & \le \exp\left(-\frac{3 s^2}{200 C}\right) = \exp\left(- \frac{3072}{25} \log(n/\delta)\right) < \frac{1}{2 n}
  \end{align*}
  And for Eq. \ref{eq:F-forward}:
  \begin{align*}
    \Pr\left[ |(P_x \cap P_{v_t}) \setminus F| \ge \frac{1}{3} |P_x \cap P_{v_t}| \right] 
      &= \Pr\left[ \sum_{i \in (P_x \cap P_{v_t})} X_i \le \frac{2}{3} |P_x \cap P_{v,t}| \right] \\
      &\le \exp\left( -2 \left(\frac{7}{10} - \frac{2}{3}\right)^2 |P_x \cap P_{v,t}| \right) \\
      &\le \exp\left(-\frac{1}{450} |P_y \cap P_v|\right) = \exp\left(-\frac{s^2}{900 C}\right) = \exp\left(- \frac{2048}{225} \log(n/\delta)\right) < \frac{1}{2n}
  \end{align*}
  Applying a union bound for the complements of Eq. \ref{eq:F-forward} and Eq. \ref{eq:F-backward} over all applicable $z$, we find that both conditions hold with positive probability, so a suitable $F$ exists.
  
  Now that $F$ has been chosen, we will select the colors for the edges from $v_t$ to each $x \in A(v_t)$ from the set $(P_x \cap P_{v_t}) \setminus U_{t-1,x,v_t} \setminus F$. Since no colors in $F$ are chosen, for any $z \in B(v_t)$, $U_{t,z,v_t}$ will not contain any element of $F$; thus $|U_{t,z,v_t}| \le |(P_x \cap P_{v_t}) \setminus F| \le \frac{1}{3} |P_x \cap P_{v_t}|$.
  
  Construct the bipartite graph $J$ between $A(v_t)$ and $P_{v_t}$, where $x \in A(v_t)$ has an edge to each of the colors in $(P_x \cap P_{v_t}) \setminus U_{t-1,x,v_t} \setminus F$. We claim there is an an $A(v_t)$-saturating matching $M$ of $J$; given this matching, we assign to edge $\{x,v_t\}$ its matched color $M(x)$. For all $x \in A(v_t)$, we will have $M(x) \in (P_x \cap P_{v_t}) \setminus U_{t-1,x,v_t} \setminus F$; since $M(x) \notin U_{t-1,x,v_t}$, the edge color for $\{x,v_t\}$ will not have been used before by any edge adjacent to $x$.
  As the matching assigns a unique color to each edge, the edge coloring constraint will also be satisfied for $v_t$.
  
  To prove there exists a matching $M$ in $J$, we verify that Hall's condition holds. For any subset $S$ of the vertices in $A(v_t)$, we want to show that
  \begin{align*}
    \left| \bigcup_{x \in S} ((P_x \cap P_{v_t}) \setminus U_{t-1,x,v_t} \setminus F) \right| \ge |S|
  \end{align*}
  The construction of $H$ ensures that $A(v_t)$ and all subsets thereof satisfy the conditions for the second part of Property U (Specifically, $|A(v_t)| \le s \Delta / C$, and that for any $b \in V$ there is at most one $j$ for which $(b,j) \in A(v,t)$.) By this property, we are guaranteed that for all $T \subseteq P_{v_t}$ of size $k-1$, that there exists some $x \in S$ for which $|P_x \cap T| \le \frac{1}{10} |P_x \cap P_{v_t}$. If Hall's condition does not hold for $S$, then there must exist some $T \subseteq P_{v,t}$ for which:
  \begin{align*}
  &\bigcup_{x \in S} ((P_x \cap P_{v_t}) \setminus U_{t-1,x,v_t} \setminus F) \subseteq T \\
  \implies \quad \forall x \in S:\quad & (P_x \cap P_{v_t}) \setminus U_{t-1,x,v_t} \setminus F \subseteq T \\
  \implies \quad \forall x \in S:\quad & (P_x \cap P_{v_t}) \setminus U_{t-1,x,v_t} \setminus F \subseteq T \cap P_x \\
  \implies \quad \forall x \in S:\quad & |(P_x \cap P_{v_t}) \setminus U_{t-1,x,v_t} \setminus F| \le |T \cap P_x| \\
  \implies \quad \forall x \in S: \quad& \frac{1}{10} |P_x \cap P_{v_t}| \le |T \cap P_x| \qquad \text{by Eq. \ref{eq:F-backward}}
  \end{align*}
  But by the second part of Property U, there must exist an $x \in S$ for which $\frac{1}{10} |P_x \cap P_{v_t}| > |T \cap P_x|$; this is a contradiction, so it follows that Hall's condition does hold for $S$. Since Hall's condition holds for all $S \subseteq A(v_t)$, $J$ will contain a matching, and step $t$ will ensure invariant W holds for step $t+1$.

  By induction, it follows that invariant W holds for all $t \in [n]$, and thus that the process to color the edges of the graph will always work.
\end{proof}

At the core of the algorithm used by \Cref{thm:online-det-ea-edgecol} will be an algorithm for partial coloring of input streams. We will prove that this inner algorithm \Cref{alg:det-ea-partial} works
for a specific class of edge arrival streams. These are categorized by a Property Z, which is 
closely linked to the way the free color tracker (\Cref{alg:ds-rand-blocks}) refreshes its
pool of colors.

\begin{definition}  For each edge $\{u,v\}$ that arrives at time $t$,
  let $d_{u,t}$ and $d_{v,t}$ be the degrees of $u$ and $v$ respectively in the multigraph
  formed by all stream edges up to $t$. The edges adjacent to each $x \in V$ are assigned
  to blocks depending on the degree of $x$ after they were added; thus edge $\{x,y\}$ arriving
  at time $t$ is assigned to block number $b_{x,t} := \ceil{d_{x,t} \cdot \frac{C}{s \Delta}}$.
  Note that $b_{x,t} \in [C/s]$. The stream satisfies Property Z when, for all
  $v \in V$, $i \in [C/s]$, and $w \in V \setminus \{v\}$, the stream contains at most
  one edge $\{v,w\}$, added at time $t$, for which $b_{v,t} = i$.
\end{definition}

\begin{algorithm}[htb!]
  \caption{Partial coloring algorithm: $(1/3)$-partial $O(\Delta)$ edge coloring for graph edge arrival streams satisfying Property Z, plus reference counting\label{alg:det-ea-partial}}
  \begin{algorithmic}[1]
  \Statex \textbf{Input}: Stream of edge arrivals $n$-vertex graph $G=(V,E)$. 
    \Statex Assume $\Delta$ is a power of two % (todo: convention? Or let $\hat{\Delta}$ be the next power of two)

    \Statex
    \Statex \underline{\textbf{Initialize($\ell$,$\xi$,$\Delta$,$C$,$(\sigma_v)_{v \in V}$, $s$, $R$):}}
      \Statex Input $\Delta$ is the maximum degree of the input graph stream
      \Statex Input $C$ the number of colors this sketch will use
      \Statex Input $s$ is block size parameter, and $(\sigma_v)_{v \in V}$ are $s$-wise almost independent permutations
      \Statex Input $R$ is a reference counted pool of edges
      \Statex Each edge $e \in R$ will have associated counter $M^\pell_{e} \in [2^\ell]$ and color class $\chi^\pell_{e} \in \{0,\ldots,\ceil{\log_{3/2}{\Delta^\pell}}\} \times [C^\pell]$

      \For{$v \in V$}
        \State $F_v \gets \textsc{InitFreeTracker}(C,s,\Delta,\sigma_v)$. \Comment{Also referred to as: $F^\pellxi_v$}
      \EndFor
    
  \Statex 
  \Statex\underline{\textbf{Process}(edge $\{x,y\}$) $\rightarrow$ \textbf{Option<color>} $\in \{\bot\} \cup [C]$}
    \If{$F_x \cap F_y \ne \emptyset$}\label{step:layer-color-test}
      \State Choose $c \in F_x \cap F_y$ arbitrarily\label{step:layer-color-choice}
      \State $F_x.\textsc{RemoveAndUpdate(c, $\{x,y\}$)}$ \Comment{This will increase $\{x,y\}$'s refcount in $R$}
      \State $F_y.\textsc{RemoveAndUpdate(c, $\{x,y\}$)}$
      \State $\chi^\pell_{\{x,y\}} \gets (\xi,c)$
      \State $M^\pell_{\{x,y\}} \gets 1$
      \State \textbf{return} color $c$
    \Else
      \State \textbf{return} $\bot$
    \EndIf
  \end{algorithmic}
\end{algorithm}
  
\begin{lemma}\label{lem:ea-partial-coloring}
   \Cref{alg:det-ea-partial} properly edge-colors a $\ge 1/3$ fraction of the edges in its input stream, if the permutations it is given are good according to \Cref{lem:offline-ea-coloring}, and the input stream satisfies Property Z. 
\end{lemma}

\begin{proof}[Proof of \Cref{lem:ea-partial-coloring}]
  % todo: for easier explanation, might be better to introduce an inner-inner algorithm, which just does partial coloring?
  Say that the input stream for \Cref{alg:det-ea-partial} satisfies Property Z. To each edge $\{x,y\}$, arriving at time $t$, we can associate a set of possible colors $P_{x,b_{x,t}} \cap P_{y,b_{y,t}}$, where $P_{x,i} := \sigma_x[(i-1)s + [s]$ indicates the $i$th set of colors used by the free color tracker $F_x$. Of course, as the algorithm progresses some of the colors in $P_{x,b_{x,t}} \cap P_{y,b_{y,t}}$ may be used by other edges adjacent to $x$ and $y$; the color trackers
  $F_x$ and $F_y$ precisely record these.
  
  Let $H$ be the simple graph on $V \times [C/s]$ formed
  by mapping each edge $\{x,y\}$ arriving in the input stream at time $t$ to the edge $\{(x,b_{x,t}),(y,b_{y,t})\}$. Because the stream satisfies Property Z, for any $u,i,v$, there is at most one $j$ for which $\{(u,i),(v,j)\}$ is in $H$. Thus, by \Cref{lem:offline-ea-coloring}, with probability $\ge 1 - \delta$ over randomly chosen advice, the permutations $(\sigma_v)_{v \in V}$ are good, and there exists an edge coloring $\chi$ of $H$ where each edge $\{(x,i),(y,j)\}$ is given a color  from $P_{x,i} \cap P_{y,j}$. This implies that, \emph{if}
  the color chosen at Line \ref{step:layer-color-choice} were to exactly match the color from $\chi$ at each step, it would be possible to for the first layer to assign a color to every edge.
  
  However, Line \ref{step:layer-color-choice}, when processing edge $\{x,y\}$ at time $t$, chooses a color arbitrarily from the set of available colors in $P_{x,b_{x,t}} \cap P_{y,b_{y,t}}$. As a result, the algorithm may select colors so that at some point, a given edge has $F_x \cap F_y = \emptyset$ on Line \ref{step:layer-color-test}, and cannot be colored. We claim that nevertheless, it will color a $\ge \frac{1}{3}$ fraction of all input edges. Consider any fixed input graph stream of length $T$ for \Cref{alg:det-ea-partial}, whose edges form multiset $E$. Consider a run of this algorithm on the stream. At each time $t$, let $\rho_t : E \mapsto [C] \cup \{\bot\}$ indicate the partial coloring produced by the stream after $t$ edges were processed. Let $\chi : E \mapsto [C]$ be the coloring produced by  \Cref{lem:offline-ea-coloring}. Call an edge $e$ "good" at time $t$ if $\rho_t(e) = \bot$ and it is possible to assign color $\chi(e)$ to $e$. (In other words, there is no edge $f$ incident on one of $e$'s endpoints for which $\rho_{t-1}(f) = \chi(e)$.) Initially, all edges in $E$ are good. Each time $t$ that the algorithm processes an edge $\{u,v\}$, it will either fail to color the edge, or set $\rho_t(\{u,v\}) = c$ for some color $c$. If $c = \chi(\{u,v\}$, then the
  number of "good" edges will be reduced by 1, because $\chi$ is a valid edge
  coloring. If $c \ne \chi(\{u,v\}$, then the number of "good" edges will be reduced by 3; $\{u,v\}$ will no longer be good, and there are at most two
  edges $f$ that are incident to either $u$ or $v$ and have $\chi(f) = c$. If
  the algorithm fails to color edge $\{u,v\}$, then this means $\chi(\{u,v\})$
  was not available (because some edge incident on $u$ or $v$ used that color);
  so in all cases, after the algorithm processes an edge, it will no longer be "good". Thus, at the end of the stream, there will be no "good" edges remaining;
  and since the number of "good" edges is reduced by at most 3 per edge that was
  colored, the number of colored edges must be at least $|E|/3$.
\end{proof}

Finally, we state and prove the formal version of \Cref{thm:online-det-ea-edgecol}.

\begin{theorem}[Formal version of \Cref{thm:online-det-ea-edgecol}]\label{thm:online-det-ea-formal}
  There is a deterministic algorithm for online $O(\Delta (\log \Delta)^2)$ edge coloring in edge arrival streams for multigraphs, using $O(n \sqrt{\Delta} (\log n)^{2.5} (\log \Delta)^3)$ bits of space, and $\tO(n \sqrt{\Delta})$ bits of advice. (By picking a uniformly random advice string, the same algorithm can alternatively be used as a robust algorithm with $1/\poly(n)$ error; the advice can also be verified and computed in exponential time.)
\end{theorem}

% URN4-52, URN4-53, and URN4-57
\begin{proof} We claim that \Cref{alg:det-ea-edgecol} satisfies the conditions of the theorem. This algorithm uses $O(n \sqrt{\Delta})$ bits of advice, for which we do not know of an efficient polynomial time construction. $\delta \in (0,1)$ is a parameter which gives, if the advice string is chosen uniformly at random, an upper bound on the probability that the advice string does not work for all possible inputs. If we
ran this algorithm with a random advice string, it would be robust to adversarially generated inputs, with failure probability $\le \delta$.
  
  \Cref{alg:det-ea-edgecol} runs $O(\log \Delta)$ instances of an inner algorithm, \Cref{alg:det-ea-inner}, which is designed to give correct edge colorings
  for graph streams with a specific low-repetition guarantee: that for any vertex $v \in V$, if one considers the sequence of edges adjacent to $v$ in stream order,
  and splits them into $O(\sqrt{\Delta / \log n})$ contiguous lists, that no lists
  will include a given edge more than once. To handle edges that are more
  commonly repeated, \Cref{alg:det-ea-edgecol} will keep track of all edges which,
  if added, might violate the guarantee, and send them to another instance of
  \Cref{alg:det-ea-inner} which handles graph streams where no contiguous lists
  of edges adjacent to a vertex includes a given edge twice; and if an edge
  in the stream might violate that condition, the algorithm sends it to another
  copy of \Cref{alg:det-ea-inner}, and so on.

  \begin{algorithm}[ht!]
    \caption{Inner algorithm:  $O(\Delta \log \Delta)$ edge coloring for graph edge arrival streams which have certain substreams satisfying Property Z, plus reference counting\label{alg:det-ea-inner}}
    \begin{algorithmic}[1]
    \Statex \textbf{Input}: Stream of edge arrivals $n$-vertex graph $G=(V,E)$. 
      \Statex Assume $\Delta$ is a power of two
      \Statex Superscript $\cdot^\pell$ indicates \emph{level} of algorithm
      
      \Statex
      \Statex \underline{\textbf{Initialize($\ell$, $\Delta^\pell$, $R$)}:}
        \Statex Input $\Delta^\pell$ is the maximum degree of the input graph
        \Statex Let $C^\pell = 32 \Delta^\pell$
        \State $D^\pell \gets \emptyset$ be a set of $O(n \log n)$ "overflow" edges % or technically edge class?
        
        \If{ $\Delta^\pell \ge 256 \log(n/\delta))$}
          \Comment{Condition for \Cref{lem:offline-ea-coloring} to apply}
          \State Let $s^\pell$ satisfy constraints of \Cref{lem:offline-ea-coloring}
          \State Advice: $\{\sigma_v^\pell\}_{v \in V}$ are permutations over $[C^\pell]$, "good" for \Cref{lem:offline-ea-coloring}
          
          \For{each \emph{layer} $\xi \in [\ceil{\log_{3/2}{\Delta^\pell}}]$}
            \State $\mathfrak{I}^\pellxi \gets$ \textsc{Initialize}($\ell$,$\xi$,$\Delta$,$C$,$s$,$\{\sigma_v^\pell\}_{v\in V}$,$R$) from \Cref{alg:det-ea-partial}
          \EndFor
        \EndIf
      
    \Statex 
    \Statex\underline{\textbf{Process}(edge $\{x,y\}$) $\rightarrow$ \textbf{color} $\in \left[\ceil{\log_{3/2}{\Delta^\pell}}\right] \times \left[C^\pell\right]$}
        \If{$\Delta^\pell \ge 256 \log(n/\delta))$}
          \For{$\xi \in \left[\ceil{\log_{3/2}{\Delta^\pell}}\right]$}\label{step:loop-over-layers}
            \State Let $c \gets \mathfrak{I}^\pellxi.\textsc{Process}(\{x,y\})$ from \Cref{alg:det-ea-partial}
            \If{$c \ne \bot$}
              \State \textbf{return} color $(\xi, c)$
            \EndIf
          \EndFor
        \EndIf
        \State $D^\pell \gets D^\pell \cup \{\{x,y\}\}$\label{step:record-unpartialcolored-edge}
        \State Increase reference count for $\{x,y\}$ in $R$
        \State Greedily pick a color class $c \in [C^\pell]$ not used by any edge in $D^\pell$ adjacent to $x$ or $y$
        \State $\chi^\pell_{\{x,y\}} \gets (0,c)$
        \State $M^\pell_{\{x,y\}} \gets 1$
        \State \textbf{return} color $(0,  c)$
    \end{algorithmic}
  \end{algorithm}
  
  We will first prove that the inner algorithm, \Cref{alg:det-ea-inner}, works.
  This algorithm runs a number of instances of \Cref{alg:det-ea-partial}, which perform a greedy partial coloring of the input stream, with constraints on the set of colors that it can use for any edge, as per \Cref{lem:offline-ea-coloring}. This greedy-coloring can be performed
  in $O(n \sqrt{\Delta})$ space, and by \Cref{lem:ea-partial-coloring} is guaranteed 
  to color at least $1/3$ of
  the edges in the input stream. By sending all edges which the greedy procedure
  did not color to an independent greedy coloring instance, the number of uncolored
  edges can be reduced further; after $O(\log \Delta)$ iterations of this, an $O(1/\Delta)$ fraction of the input stream has not been colored; this part of the stream can be stored entirely and colored with $2 \Delta - 1$ colors.
  
  In \Cref{alg:det-ea-inner}, the for loop at Line \ref{step:loop-over-layers} sends the edge $\{x,y\}$ being processed to each of the $\ceil{\log_{3/2} \Delta^\pell}$ instances of \Cref{alg:det-ea-partial}, until either one of them assigns a color to the edge, or all the instances fail to color the edge. Since each instance is guaranteed to color a $\ge 1/3$ fraction of the edges it processes, only a $(2/3)^\ceil{\log_{3/2} \Delta^\pell} \le \frac{1}{\Delta^\pell}$ fraction of the edges received by \Cref{alg:det-ea-inner} will reach Line \ref{step:record-unpartialcolored-edge} of the algorithm and be stored in $D^\pell$; since there are only $O(n)$ such edges, storing them does not significantly affect the space usage of the algorithm. These edges will be greedily colored using a fresh set of colors.
  
  \begin{algorithm}[ht!]
    \caption{Deterministic algorithm for $O(\Delta (\log \Delta)^2)$ edge coloring for multigraph edge arrival streams\label{alg:det-ea-edgecol}}
    \begin{algorithmic}[1]
    \Statex \textbf{Input}: Stream of edge arrivals $n$-vertex graph $G=(V,E)$
      \Statex Assume $\Delta$ is a power of two
      \Statex
      \Statex \underline{\textbf{Initialize}:}
      
      \State Let $R \gets \emptyset$ be a reference counted pool of edges. (This will be a set of "recent" edges for each layer in each level, including edges that are either in $D^\pell$ or have a reference from one of the $F_x^{(\ell,\xi)}$
      \State Each edge $e \in R$ will have, per level $\ell$, one associated counter $M^\pell_{e} \in [0, 2^\ell]$ and color class $\chi^\pell_{e} \in \{0,\ldots,\ceil{\log_{3/2}{\Delta^\pell}}\} \times [C^\pell]$. When an edge is added to the pool, $M^\pell_{e} = 0$ for all layers.
      
      \For{each \emph{level} $\ell$ in $0,\ldots,\log \Delta$}
        \State Define $\Delta^\pell = \Delta / 2^\ell$
        \State $\mathfrak{K}^\pell \gets$ \textsc{Initialize}($\ell$,$\Delta^\pell$) from \Cref{alg:det-ea-inner}
      \EndFor
      
    \Statex 
    \Statex\underline{\textbf{Process}(edge $\{x,y\}$) $\rightarrow$ \textbf{color}}
      
      \For{$\ell$ in $0,\ldots,\log \Delta$}\label{step:det-ea-edgecol-outer-loop}
        \If{$\{x,y\}$ is in $R$ and $M^\pell_{\{x,y\}} > 0$}
          \If{$M^\pell_{\{x,y\}} = 2^\ell$}
            \State \textbf{continue}\label{step:skip-to-next-level}
          \Else
            \State $M^\pell_{\{x,y\}} \gets M^\pell_{\{x,y\}} + 1$
            \State Let $(i,j) \gets \chi^\pell_{\{x,y\}}$
            \State \textbf{return} color $\left(\ell, i,  (j - 1) 2^\ell + M^\pell_{\{x,y\}}\right)$
          \EndIf
        \Else
          \State Let $(\xi, c) \gets \mathfrak{K}^\pell.\textsc{Process}$($\{x,y\}$) for level $\ell$ algorithm
          \State \textbf{return} color $\left(\ell, \xi,  (c - 1) 2^\ell + 1\right)$
        \EndIf
      \EndFor
    \end{algorithmic}
  \end{algorithm}
  
  To handle general multigraphs, \Cref{alg:det-ea-edgecol} divides the stream into
  a series of levels, for $\ell$ from $0$ to $\log \Delta$. The higher levels
  process very common edges, which are assigned blocks of $2^\ell$ colors for layer $\ell$. Each level uses an instance of \Cref{alg:det-ea-inner} to assign color
  blocks for the edges it receives. If a copy $e$ of a given edge $\{u,v\}$ arrives,
  and the color block for $\{u,v\}$ is not full, then the $e$ will be assigned the
  next available color in the block. The specific scheme, as we shall show, ensures that every edge is either colored from an existing color block, or passed to an instance of \Cref{alg:det-ea-inner}; and in the latter case, ensures that property
  Z holds for the stream sent to \Cref{alg:det-ea-inner}.
  
  \Cref{alg:det-ea-edgecol} maintains a global reference counted pool $R$, which
  keeps track of every edge $\{x,y\}$ that arrives for some amount of time. Each
  level $\ell$ can provide references for the edge; $\{x,y\}$ will only be dropped
  from $R$ if no levels have a reference. At level $\ell$, we have two cases,
  depending on how $\{x,y\}$ was processed by  \Cref{alg:det-ea-inner}. 
  If $\{x,y\}$ was not colored by any layer $\xi$, it will be stored in $D^\pell$ for the rest of the stream, and any further copies of that edge will not be sent by \Cref{alg:det-ea-edgecol} to the level $\ell$ instance of \Cref{alg:det-ea-inner}. If $\{x,y\}$ was successfully colored by layer $\xi$, the edge will be recorded until both of the free color trackers $F_x^\pellxi$ and $F_y^\pellxi$
  have been refreshed. This only happens if the block numbers $b_{x,t}$ and $b_{y,t}$ of $\{x,y\}$ with respect to the substream received by the layer $\xi$ instance of \Cref{alg:det-ea-partial} have increased. Consequently, that substream will satisfy Property Z -- if \Cref{alg:det-ea-edgecol} receives a second copy of edge $\{x,y\}$ at time $t'$ while either $b_{x,t'} = b_{x,t}$ or $b_{y,t} = b_{y,t}$, because $\{x,y\}$ will be still be in $R$, \Cref{alg:det-ea-edgecol} will not send the second copy to the level $\ell$ instance of \Cref{alg:det-ea-inner}.
  
  We now check that the level $\ell$ instance of \Cref{alg:det-ea-inner} does not receive a graph stream of degree more than $\Delta^\pell$. When a copy of an edge $\{u,v\}$ arrives, the loop at Line \ref{step:det-ea-edgecol-outer-loop} only continues from level $\ell$ to to level $\ell + 1$ if $M_{\{u,v\}}^\pell = 2^\ell$ was true. Thus, for an edge $e$ to be processed by the level $\ell$ instance of \Cref{alg:det-ea-inner}, the edge $e$ must have arrived at least $2^0 + 2^1 \ldots + 2^{\ell-1} = 2^{\ell} - 1$ times before in the stream, and $2^\ell$ times in total, since the last time $e$ was dropped from $R$. Since the maximum degree of the graph is $\Delta$, and an edge must be added $\ge 2^\ell$ times for each time that
  it is sent to the level $\ell$ instance of \Cref{alg:det-ea-inner}, the level $\ell$ instance will
  receive at most $\Delta/2^\ell$ edges adjacent to any given vertex.

  When $\ell = \log \Delta$ in the for loop, Line \ref{step:skip-to-next-level} will not be executed;
  because for $M_{\{x,y\}}^\pell$ to equal $2^\ell = \Delta$, this must be the $\Delta + 1$st copy of
  edge $\{x,y\}$ to arrive. Thus \Cref{alg:det-ea-edgecol} will assign a color to every edge in
  the stream.
  
  The total space usage of \Cref{alg:det-ea-edgecol} is dominated by the sets $D^\pell$, free color trackers $F^\pellxi_v$, and the pool $R$ (along with its linked properties $\chi^\pell_e,M^\pell_e$, for each $\ell \in \{0,\ldots,\log \Delta\}$.) Over all levels $\ell$, layers $\xi$, and vertices in $V$, there are
   $O((\log \Delta)^2 n)$ color trackers, each of which uses $O(s^\pell \log n+ \log \Delta) = O(\sqrt{\Delta \log(n/\delta)} \log n)$ bits of space to store colors and $O(\log n)$-bit references to edges.
  Each $D^\pell$ is guaranteed to contain $O(n \log n)$ edges, at most, and needs
  $O(n (\log n)^2)$ bits of space. Finally, the $\ell$th level references at most  $|D^\pell| + O(n s^\pell \log \Delta)$ edges in $R$, so in total $R$ will have $O( n \log n + n \sqrt{\Delta \log(n/\delta)} \log \Delta )$ edges. Each each needs $O(\log n)$ bits to identify, and there are
  $O((\log \Delta)^2)$ bits of associated information in the $(\chi^\pell_e,M^\pell_e)_{\ell}$.
  Thus, in total, \Cref{alg:det-ea-edgecol} uses:
  \begin{align*}
    O(n \sqrt{\Delta \log(n /\delta)} (\log \Delta)^2  \log n \log(n\Delta))
  \end{align*}
  bits of space.
  
  If the advice $(\sigma_v^\pell)_{v \in V, \ell \in \{0,\ldots,\log \Delta\}}$ was chosen randomly
  using \Cref{lem:fast-permutations}, then $O(s \poly(\log\Delta,\log n))$ truly random bits per permutation
  would be needed for a level of accuracy ($\le 1/\poly(n)$ total variation variation distance from 
  uniformity) under which the proof of \Cref{lem:offline-ea-coloring} works. At $\delta = 1/2$,
  this is $O(n \sqrt{\Delta} \poly(\log\Delta,\log n)$ bits. Given exponential time, the
  advice can also be computed on demand, since checking that Property U from the proof of \Cref{lem:offline-ea-coloring} can be done in exponential time. 
  
\end{proof}

% It may be possible to improve the simple-to-multigraph overhead to $O(\log \log \Delta)$ by adjusting the smaller levels to use proportionally larger $s^\ell$.}

\section{Lower bounds for deterministic edge coloring}

% URN4-42
\begin{lemma}\label{lem:lb-rand-part-matching}
  Let $B$ be a set of $n$ vertices; let $\Delta$ be an integer, and let $C \in [\Delta, 2 \Delta - 1]$ be another integer. Define $\beta := C / \Delta$. Consider the case where $(2 - \beta) n \ge 32 C$.
  
  For each $v \in B$, say we have a nonempty set $S_v \subseteq [C]$ of possible colors, where $\sum_{i \in B} |S_i| \le \beta n$. If $G$ is a uniformly random bipartite graph between a set $A$ of size $\floor{n/\Delta}$, then the probability $p(n, \Delta, \beta)$ that $G$ has a valid edge coloring where each edge $(a,b)$ is given a color from $S_b$ is:
  \begin{align}
      p(n, \Delta, \beta) \le \exp(- \frac{1}{2^{13}} (2-\beta)^3 n) \label{eq:rgmprob}
  \end{align}
\end{lemma}

\begin{proof}[Proof of \Cref{lem:lb-rand-part-matching}]
  The graph $G$ can be interpreted as a random partition of $B$ into sets $P_1,\ldots,P_{\floor{n/\Delta}}$, plus a possible $P_{\text{remainder}}$ set, where all sets $P_i$ except for the remainder have size $\Delta$. Note that $P_1$ is a uniformly random subset of size $\Delta$ in $B$, $P_2$ is a uniformly random subset of size $\Delta$ in $B \setminus P_1$, and so on.
  For each $i \in \floor{n/\Delta}$, let $C_i$ be the event that there is a $P_i$-saturating
  matching in the bipartite graph between $P_i$ and $[C]$, where each $v \in P$ is adjacent to all $c \in S_v$. Let $\gamma = 2 - \beta$. Then we have:
  \begin{align}
    p(n, \Delta, \beta) \le \prod_{i = 1}^{\ceil{\frac{1}{2} \gamma n / \Delta}} \Pr[C_i | C_1,\ldots,C_{i-1}] \label{eq:strat-pndb-ub}
  \end{align}
  We will prove Eq. \ref{eq:rgmprob} by proving a upper bounds on $\Pr[C_i | C_1,\ldots,C_{i-1}]$, for all $i \in [\ceil{\frac{1}{2} \gamma n / \Delta}]$, and then applying Eq. \ref{eq:strat-pndb-ub}.
  
  By Markov's inequality, the fraction of vertices in $B$ for which $|S_v| - 1 \ge 1$ is $\le \beta - 1 = 1 - \gamma$, so $\Pr_{v \sim B}[|S_v| = 1] \ge \gamma$. For each $i \in [\floor{n/\Delta}]$, let $T_i = B \setminus \bigcup_{j < i} P_j$. Then for all $i \le \ceil{\frac{1}{2} \gamma n / \Delta}$:
  \begin{align*}
    \frac{|\{v \in T_i : |S_v| = 1\}|}{|T_i|} \ge \frac{|\{v \in B : |S_v| = 1\}| - \Delta (i-1)}{|B| - \Delta i} \ge \frac{\gamma n - (i-1) \Delta}{n - \Delta (i-1)} \ge \frac{\gamma n - \frac{1}{2} \gamma n}{n - \frac{1}{2} \gamma n} \ge \frac{1}{2} \gamma
  \end{align*}
  Consequently, conditioned on $P_1,\ldots,P_{i-1}$, the set $P_i$ will be drawn uniformly at random from a set $T_i$ of vertices for which at least a $\gamma / 2$ fraction have singleton color sets (have $|S_v| = 1$).
  
  For a given $i$, we remark that if $P_i$ contains two vertices $v,w$ for which $|S_v| = |S_w| = 1$ and $S_v = S_w$, then it is not possible to match the vertices in $P_i$ to colors, as $v$ and $w$ would conflict. Let us bound the probability that this occurs. Let $\hatn = |T_i|$; note that this is $\ge n / 2$. To make the distribution of singleton sets drawn from $\{S_v\}_{v \in T_i}$ appear more uniform, we construct $C$ disjoint sets $L_1,\ldots,L_C$ in $T_i$, so that for each $L_i$, the associated color sets are all a singleton $|\bigcup_{j\in L_i} S_j| = 1$; and for which $|L_i| \ge \gamma \hatn / 4 C$. (If each singleton color set were equally likely, we could get $|L_i| \ge \gamma \hatn / 2 C$ -- but it is possible that $\{1\}$ is rare, while $\{2\}$ more common than average. One way to construct the $L_1,\ldots,L_C$ is by iteratively removing sets of $\gamma \hatn / 4 C$ vertices from $T_i$ whose associated color sets are all the same singleton set.)
  
  For each $j \in T_i$, let $X_j$ be the indicator random variable for the event that $j \in P_i$. For each $k \in [C]$, let $Y_k = \sum_{j \in L_k} X_j$. Since the random variables $\{X_j\}_{j \in T_i}$ are negatively associated, sums of disjoint sets of them, the $\{Y_k\}_{k \in [C]}$, are also negatively associated. (See \cite{JoagDevP83}.) We have:
  \begin{align*}
    \Pr&[P_i \text{ has no two elements from same $L_k$}] \\
      &\le \Pr[\bigwedge_{k \in [C]} \{ Y_k \le 1 \} ] \\
      &\le \prod_{k \in [C]} \Pr[Y_k \le 1] \qquad \text{since negative association $\implies$ negative orthant dependence}
  \end{align*}
  We calculate $\Pr[Y_k \le 1]$ exactly, and then prove an upper bound on it.
  \begin{align*}
    \Pr[Y_k \le 1]
      &= \frac{\binom{\hatn - |L_k|}{\Delta} + |L_k| \binom{\hatn - |L_k|}{\Delta - 1}}{\binom{\hatn}{\Delta}} \\
      &= \frac{\frac{\hatn - |L_k| - \Delta + 1}{\Delta} \binom{\hatn - |L_k|}{\Delta - 1} + |L_k| \binom{\hatn - |L_k|}{\Delta - 1}}{\frac{\hatn}{\Delta}\binom{\hatn - 1}{\Delta - 1}} \\
      &= \frac{\frac{\hatn - |L_k| - \Delta + 1}{\Delta} + |L_k| }{\hatn/\Delta} \cdot \frac{\binom{\hatn - |L_k|}{\Delta - 1}}{\binom{\hatn-1}{\Delta -1}} \\
      &\le \frac{\frac{\hatn - |L_k|- \Delta + 1}{\Delta} + |L_k| }{\hatn/\Delta} \cdot \left( \frac{\hatn - |L_k|}{\hatn-1} \right)^{\Delta -1} \qquad \text{since $\binom{a}{c}/\binom{b}{c} \le (a/b)^c$ if $c \le a \le b$}\\ 
      &= \left(1 + \frac{|L_k| - 1}{\hatn} (\Delta - 1)\right) \left(1 - \frac{|L_k| -1}{\hatn - 1}\right)^{\Delta - 1} \\
      &\le \exp\left( \ln\left(1 + \frac{|L_k| - 1}{\hatn} (\Delta - 1) \right)   - \frac{|L_k| -1}{\hatn - 1} (\Delta - 1) \right) \qquad \text{since $1-x \le \exp(x)$} \\
      &\le \exp\left( \frac{|L_k| - 1}{\hatn} (\Delta - 1) - \frac{1}{4} \left(\frac{|L_k| - 1}{\hatn} (\Delta - 1)\right)^2  - \frac{|L_k| -1}{\hatn - 1} (\Delta - 1) \right) \qquad \text{since $\ln(1+x) \le x - x^2/4$ for $x \le 1$} \\
      &\le \exp\left( -\frac{(|L_k|-1)(\Delta - 1)}{\hatn (\hatn - 1)} - \frac{1}{4} \left(\frac{|L_k| - 1}{\hatn} (\Delta - 1)\right)^2 \right) \\
      &\le \exp\left( - \frac{1}{4} \frac{\gamma^2}{2^{10}} \right) = \exp(- \gamma^2 / 2^{12})
  \end{align*}
  In the last inequality, we used the fact that $(|L_k| -1)(\Delta - 1) / \hatn \ge (\gamma \hatn / 4 C - 1)(\Delta - 1) / \hatn \ge \gamma / 32 C$.
  % Note: L_k indirection is still required for this derivation because the ln(1 + x) <= x - x^2/4
  % inequality will break if (|L_k| -1)(\Delta - 1)/n is Ω(1). Working around this is not worth the time.
  
  Having bounded $\Pr[Y_k \le 1]$, it follows that:
  \begin{align*}
    p(n, \Delta, \beta) &\le \prod_{i = 1}^{\ceil{\frac{1}{2} \gamma n / \Delta}} \Pr[C_i | C_1,\ldots,C_{i-1}] \le \left( \exp\left(- \gamma^2 / 2^{12}\right)^C \right)^{\ceil{\frac{1}{2} \gamma n / \Delta}} \\
      &\le  \exp\left(- \gamma^2 / 2^{12} \cdot C \cdot \ceil{\frac{1}{2} \gamma n / \Delta} \right) \le \exp(-\gamma^3 n / 2^{13})
  \end{align*}
\end{proof}

We now formally restate \Cref{thm:det-online-va-lb} and prove it.
\begin{theorem}[Formal version of \Cref{thm:det-online-va-lb}]\label{thm:det-online-va-lb-formal}
  For all $\beta \in (1,2)$, and integers $n,\Delta$ satisfying $\Delta \le n  (2-\beta) / (32 \beta)$, every deterministic online streaming algorithm for edge-coloring that uses $\beta \Delta$ colors requires $\Omega((2-\beta)^3 n)$ bits of space.
\end{theorem}

% URN4-26
\begin{proof} Say we have an algorithm $\cA$ to provide an online $(\beta \Delta)$ edge-coloring of an input stream, presented in one-sided vertex arrival order, using $S$ bits of space. We assume
  that $\Delta | n$; if this is not the case, we can reduce $n$ to the nearest multiple
  of $\Delta$, weakening our final lower bound by at most a factor of $2$.
  With the algorithm, we can implement a protocol for a $\Delta$-player one-way
  communication game in which each message uses $\le S$ bits. We will then prove a communication lower bound
  for this game.
  
  Specifically, let $P_1,\ldots,P_\Delta$ be the players of the game. Let $A_1,\ldots,A_\Delta$ and $B$ be sets of vertices, where for each $i\in[\Delta]$,
  $|A_i| = n/\Delta$, and $|B| = n$. For each $i\in[\Delta]$, the player $P_i$ is given a regular bipartite graph $G_i$ from $A_i$ to $B$, in which each vertex in $A_i$ has degree $\Delta$, and each vertex in $B_i$ has degree $1$. Player $P_1$ starts the communication game by
  outputing an edge coloring $\chi_1$ of $G_1$, using colors in $[\beta \Delta]$; and
  then it sends a message $m_1$ to Player $P_2$. For each $i \in \{2,\ldots,\Delta\}$,
  the player $P_i$ will receive a message $m_{i-1}$ from its predecessor, output an edge
  coloring $\chi_i$ of $G_1$ which is compatible with the edge colorings $\chi_1,\ldots,\chi_{i-1}$ made by the earlier players, and then (if $i < \Delta$) send a message $m_i$
  to the next player.
  
  The conversion from an algorithm $\cA$ to a protocol for this game is straightforward;
  $P_1$ initializes an instance $A$ of $\cA$, uses it to process $G_1$ in arbitrary order, and reports the colors the algorithm output; then it encodes the state of the instance $A$ as an $S$-bit
  message $m_1$. $P_2$ receives this message, and uses it to continue running the instance
  $A$, this time having it process $G_2$; $P_2$ outputs the results, and sends the new state
  of $A$ to $P_3$ as $m_2$. The players continue in this way until $P_\Delta$ produces output.
  
  For each $i \in \{1,\ldots,\Delta-1\}$, and message $m_i$, we define $\cS_{m_i} = (S_{m_i,v})_{v \in B}$. Here $S_{m_i,v}$ is the set of all colors which players $P_1,\ldots,P_i$ could 
  have assigned to edges incident on $v$ for executions of the protocol in which $P_i$ sent message $m_i$. If player $P_{i+1}$ receives message $m_i$, the coloring $\chi_{i+1}$ that it
  outputs must be disjoint from $\cS_{m_i}$; specifically, if we view $\chi_{i+1}$ as a vector
  in $[\beta \Delta]^B$ whose $v$th entry gives the color assigned to the edge incident to vertex $v$, then $\forall v \in B : \chi_{i+1, v} \notin S_{m_i,v}$. If this were not the case, and
  there was a vertex $x$ for which $\chi_{i+1, x} \in S_{m_i,x}$, then there would be an execution of the protocol on which some player $P_j$ output color $\chi_{i+1, x}$ for 
  an edge incident on $x$, later $P_i$ sent message $m_i$, and now $P_{i+1}$'s assignment of
  $\chi_{i+1, x}$ to the edge incident on $x$ violates the edge coloring constraint.
  
  Let $p(n, \Delta, \beta)$ be the probability from \Cref{lem:lb-rand-part-matching}.
  We claim that there exists an input for which some player must send a message with more than
  $\log 1/p(n, \Delta, \beta)$ bits. If this is not the case, then we shall construct
  an input on which the protocol must give an incorrect output, a contradiction.
   
  Let $M_1$ be the set of all messages that player $P_1$ can send. Let $H$
  be chosen uniformly at random from the set of $(\Delta,1)$-regular bipartite graphs
  from $A_1$ to $B$, and let $m_1(H)$ be the message $P_1$ would send if $G_1 = H$.
  Then, if for all $m \in M_1$, we were to have $\sum_{v\in B} |S_{m_i,v}| \le \beta n$, 
  \begin{align*}
    1 = \sum_{m \in M_1} \Pr[m = m_1(H)] \le |M_1| p(n, \Delta, \beta) 
  \end{align*}
  But since we have assumed messages in $M_1$ need $< \log 1/p(n, \Delta, \beta)$
  bits, and hence $|M_1| < 1 / p(n, \Delta, \beta)$, the above equation would
  imply $1 < 1$; thus there must be some $m_1^\star \in M_1$ for which
  $\sum_{v\in B} {S_{m_i,v}} \ge \beta n$. Let $\cG_1$ be the set of graphs
  for which $H \in \cG_1 \iff m_1(H) = m_1^\star$; then on being given 
  any graph in $\cG_1$, player $P_1$ will output $m_1^\star$.
  
  We will now iterate over $i \in \{2,\ldots,\Delta - 1\}$ and build a sequence
  of messages $m_1^\star, m_2^\star, \ldots, m_\Delta^\star$, along with sets
  of input graphs $\cG_1,\ldots,\cG_{\Delta-1}$ on which the protocol will send
  these messages. For each $i \in [\Delta]$, define $\cT_{m,i,i} = (T_{v, m, i})_{v \in B}$,
  where $T_{v,m,i} := S_m \setminus S_{m_{i-1}^\star,i}$. Any coloring $\chi_i$ 
  that $P_i$ outputs which is compatible with all inputs leading to $m_i^\star$ 
  will satisfy $\chi_{i,v} \in T_{v,m,i}$. (If $\chi_{i,v} \in S_{m_{i-1}^\star,i}$,
  then as noted above there is a set of inputs where this will violate
  the edge coloring constraint for $v$.) As argued for $M_1$, there must be some
  message $m_i^\star \in M_i$ for which $\sum_{v \in B} |T_{v, m, i}| \ge \beta n$.
  
  Finally, for each $v \in B$, define $R_v = [\beta \Delta] \setminus S_{m_{\Delta-1}^\star,v}$.
  Since $S_{m_{\Delta-1}^\star,v} = \sqcup_{i = 1}^{\Delta - 1} T_{v, m_i^\star, i}$
  we will have $|R_v| \le \beta \Delta - \beta (\Delta - 1) = \beta$. On receiving
  $m_{\Delta -1}^\star$, player $P_\Delta$ can only assign edge colors so that
  edges incident on $v$ use colors in $R_v$; for any other color, there is a
  input which uses it and which makes $P_{\Delta - 1}$ send $m_{\Delta -1}^\star$.
  If $G_\Delta$ were chosen uniformly at random from its set of possible values,
  then the probability that $G_\Delta$ has an edge coloring compatible with
  $\{R_v\}_{v\in m}$ is $\le p(n, \Delta, \beta)$. Since this is $< 1$, there must
  exist a specific graph $G_\Delta^\dagger$ on which the protocol uses a color
  not in $R_v$ for some $v \in B$. We have thus shown that if the protocol always
  uses fewer than $\log 1/p(n, \Delta, \beta)$ bits for its messages, it
  will give incorrect outputs for some input.
  
  We conclude:
  \begin{align*}
    S \ge \log \frac{1}{p(n, \Delta, \beta)} = \Omega((2-\beta)^3 n)
  \end{align*}
\end{proof}

\section{Supporting lemmas}

\begin{corollary}[Practical high-rate-distance-product binary codes, via \cite{SipserS96}]\label{cor:practical-binary-codes}
  For sufficiently large $t$, there is a binary code
  of length $t$ with rate $\ge \frac{1}{4} t$, and distance $\ge \frac{1}{400} t$. The code can be implemented with
  $poly(t)$ initial setup time and space, and $O(t^2)$ encoding time.
  
  % todo, if time permits: look for some paper giving explicit minimum threshold, or even a "for all" guarantee
\end{corollary}

\begin{proof}[Proof of \Cref{cor:practical-binary-codes}]
  Let $\epsilon = \sqrt{2}/20$. By \cite{SipserS96} Theorem 19, there is a polynomial time constructible family of codes of rate $1 - 2 H(\epsilon) = 0.2628\ldots \ge 1/4$, and relative distance approaching $\epsilon^2 = 1/200$. In particular, for sufficiently large code length $t$, the rate will be $\ge 1/4$ and the relative distance will be $\ge 1/400$.
  
  The expander codes described by \cite{SipserS96}, at code length $t$, require $\poly(t)$ time to construct the expander graph used, and can
  be encoded in $O(t^2)$ time.
\end{proof}

\begin{lemma}\label{lem:moment-gen-func}
   Let $W$ be a nonnegative integral random variable where, for all $k \in \NN$, $\Pr[W \ge k] \le 1/2^k$. Then $\EE[W] \le 1$; and furthermore, for all $t$ for which $e^t \in [1,2)$:
   \begin{align*}
        \EE[e^{t W}] \le \frac{1}{2 - e^t}
   \end{align*}
\end{lemma}

\begin{proof}[Proof of \Cref{lem:moment-gen-func}]
  First,
  \begin{align*}
    \EE[W] = \sum_{k=0}^{\infty} \Pr[W \ge k] \le \sum_{k=0} 1/2^{k + 1} = 1
  \end{align*}
  Next,
  \begin{align*}
    \EE[e^{t W}] &= \sum_{k=0}^{\infty} e^{t k} \Pr[W = k] \le \sum_{k=0}^{\infty} e^{t k} \frac{1}{2^{k+1}} = \frac{1}{2} \sum_{k=0}^{\infty} \left(\frac{e^t}{2}\right)^k = \frac{1}{2}\cdot\frac{1}{1 - \frac{1}{2}e^t} = \frac{1}{2 - e^t}
  \end{align*}
  The inequality step follows because:
  % proof approach is a standard sequence manipulation trick
  \begin{align*}
     \sum_{k=0}^{\infty} e^{t k} \left(\frac{1}{2^{k+1}} - \Pr[W = k]\right)
        &= \sum_{k=0}^{\infty} e^{t k} \left((\frac{1}{2^{k}} - \frac{1}{2^{k+1}}) -(\Pr[W \ge k] - \Pr[W \ge k - 1])\right) \\
        &= \sum_{k=0}^{\infty} e^{t k} \left(\frac{1}{2^{k}} - \Pr[W \ge k]\right)
          - \sum_{k=0}^{\infty} e^{t k} \left(\frac{1}{2^{k + 1}} - \Pr[W \ge k - 1]\right) \\
        &= \sum_{k=0}^{\infty} e^{t k} \left(\frac{1}{2^{k}} - \Pr[W \ge k]\right)
          - \sum_{k=1}^{\infty} e^{t (k - 1)} \left(\frac{1}{2^{k}} - \Pr[W \ge k]\right) \\
        &= \left(\frac{1}{2^{0} - \Pr[W \ge 0]}\right) + \sum_{k=1}^{\infty} (e^t - e^{t (k - 1)}) \left( \frac{1}{2^{k} - \Pr[W \ge k]}\right) \\
        &= 0 + \sum_{k=1}^{\infty} e^{t (k -1)}(e^t - 1) \left(\frac{1}{2^{k}} - \Pr[W \ge k]\right) \ge 0
  \end{align*}
\end{proof}

By \cite{Morris13} plus some algebra, for any $\epsilon > 0$, a sequence of $O\left(d^3 + d \ln(\frac{1}{\epsilon}) \right)$ Thorp shuffle steps will produce a permutation on $[2^d]$ whose distribution has total variation distance at most $\epsilon$ away from the uniform distribution. % constant approx is good enough
\footnote{While there exist more efficient switching networks that also permute sets whose sizes are not powers of two, we use the result of \cite{Morris13} here because it is simple to work with. In particular, see the results claimed by \cite{Czumaj15}, although we could not find the full version of that paper.}

%\mscomment{TODO: use state of the art switching networks here? Need approx 'depth'*k-wise independence, with this construction; best might be O(log n log log n) for depth, vs a less uniform O(n log n} by gate count}

\begin{lemma}[Random permutations through switching networks]\label{lem:fast-permutations}
  For any $C$ which is a power of 2, there is an explicit construction of an $(\epsilon,s)$-wise independent
  random permutation, using $r = O(s (\log C)^4 \log \frac{1}{\epsilon})$ bits. Furthermore, we can evaluate $\sigma(i)$ and $\sigma^{-1}(i)$ in $O(s (\log C)^4 \log \frac{1}{\epsilon} \log C)$ time.
\end{lemma}

\begin{proof}[Proof of \Cref{lem:fast-permutations}]
  Let $k = O(d^3 + d \ln(1/\epsilon))$ be the constant for which $k$ Thorp shuffle steps would permute $[C] = [2^d]$ within total variation distance of $\epsilon$ of the uniform distribution over permutations on $[2^d]$.
  
  % issue: would need to define switching network
  
  The switching network $\cN$ corresponding to the $k$ Thorp shuffle steps has depth $k$ and uses exactly $k C / 2$ gates. Assign each gate a unique number in $[kC/2]$. Then given a uniformly random bit vector $x \in \{0,1\}^{k C / 2}$, we evaluate the switching network by having the gate numbered $i$ switch its inputs iff $x_i = 1$. A key property of switching networks is that one can evaluate their action on a single
  input by only evaluating one gate per layer -- for this network, only $k$ gates. Reversing the order in which the layers are applied will produce the inverse of the original permutation. Thus, one can evaluate
  $\cN(i)$ by reading only $k$ entries of $x$, and similarly for and $\cN^{-1}{i}$.
  
  Now, say the bits of $x$ are the output of a hash function drawn from a $k s$-wise independent hash family. (For example, using a family of \cite{WegmanC81}, let $h = \ceil{\log_2{k C / 2}}$, take the family of random polynomials of degree $k s - 1$ inside $\FF_{2^{h}}$, and output the least bit of the output. The polynomial coefficients can be encoded using $k s h = O(k s \log(C))$ bits, and the polynomials evaluated at any point in $O(k s h^2 ) = O(k s (\log C)^2 )$ time.)
  
  If $\pi$ is a uniformly random permutation on $[C]$, then for all lists of distinct $h_1,\ldots,h_s$, and all lists of distinct $j_1,\ldots,j_s$, we have
  \begin{align*}
    \Pr[\bigwedge_{i\in[s]} \pi(h_i) = j_i] = \frac{1}{\prod_{i \in [s]} (C - i + 1)}
  \end{align*}
  
  Now, let $f : \{0,1\}^{k C/2} \times [C] \mapsto [C]^k$ be the function which maps the gate-controlling vector $x \in \{0,1\}^{k C / 2}$ and an input $a \in [C]$ to the path $b_1,\ldots,b_k$ that $a$ takes through the switching network $\cN$ if its gates are configured according to $x$. The last node of this path, $f(x, a)_k$ is the output of $\cN$ given $x$ and $a$. Each path $P = (a,b_1,\ldots,b_k)$ through the switching network corresponds to a restriction $R_P \in \{0,1,\star\}^{k C / 2}$ which has value $\star$ on gates not touched by the path, and for each gate traversed by the path assigns either $0$ or $1$ depending on whether a straight or switched configuration of the gate is compatible with $P$. Since all paths through the network have length $k$, $R_P$ only sets $k$ coordinates. Now, for each pair $(a,b) \in [C]^2$, let 
  \begin{align*}
    \cF_{a,b} = \{R_P : P = (a,b_1,\ldots,b_{k-1},b) \text{ is a possible path}\}
  \end{align*}
  % define restriction refinement?
  Then, for lists $(h_1,\ldots,h_s)$, $(j_1,\ldots,j_s)$, define 
  \begin{align*}
    \cK_{h_1,\ldots,h_s,j_1,\ldots,j_s} = \{ R \in \{0,1,\star\}^{k C / 2} : \forall i \in [k], \exists T \in \cF_{h_i,j_i} \text{where $R$ is a minimal refinement of $T$}  \}\,,
  \end{align*}
  i.e., the set of minimal restrictions for vectors in $\{0,1\}^{k C / 2}$ which
  completely determine the paths through the switching network of inputs $(h_1,\ldots,h_s)$. 
  
  Now, let $Y \in \{0,1\}^{k C / 2}$ be $k s$-wise independent, and $X \in \{0,1\}^{k C / 2}$ be fully independent. We have:
  \begin{align*}
    \Pr[\bigwedge_{i\in[s]} f(Y, h_i)_k = j_i] &= \sum_{R \in \cK_{h_1,\ldots,h_s,j_1,\ldots,j_s}} \Pr[Y \text{ compatible with $R$}] \\
    &= \sum_{R \in \cK_{h_1,\ldots,h_s,j_1,\ldots,j_s}} \Pr[X \text{ compatible with $R$}] \\
    &= \Pr[\bigwedge_{i\in[s]} f(X, h_i)_k = j_i]
  \end{align*}
  where the second inequality follows because each restriction $R \in \cK_{h_1,\ldots,h_s,j_1,\ldots,j_s}$ only constrains $s k / 2$ coordinates
  corresponding to the gates on the paths in the switching network from $h_1,\ldots,h_s$ to $j_1,\ldots,j_s$ that it fixes.
  
  Thus,
  \begin{align*}
      &\frac{1}{2} \sum_{\text{distinct $b_1,\ldots,b_k$ in $[C]$,}} \left| \Pr\left[\bigwedge_{i \in [s]} \{ f(X, h_i)_k = j_i \} \right] - \frac{1}{\prod_{i \in [s]} (C - i + 1)} \right| = \\
      &\frac{1}{2} \sum_{\text{distinct $b_1,\ldots,b_k$ in $[C]$,}} \left| \Pr\left[\bigwedge_{i \in [s]} \{ f(Y, h_i)_k = j_i \} \right] - \frac{1}{\prod_{i \in [s]} (C - i + 1)} \right| \le \epsilon
  \end{align*}
  which proves that the switching network evaluated on $X$ produces outputs that are $(\epsilon,s)$-wise independent.
\end{proof}

%%%%%%%%%%%%%%%%%%%%%%%%%%%%%%%%%%%%%%%%%%%%%%%%%%%%%%%%%%%%%%%
%%%%%%%%%%%%%%%%%%%%%%%%%%%%%%%%%%%%%%%%%%%%%%%%%%%%%%%%%%%%%%%
%%%%%%%%%%%%%%%%%%%%%%%%%%%%%%%%%%%%%%%%%%%%%%%%%%%%%%%%%%%%%%%

\section{Regarding implementation}

Several of the algorithms in this paper rely on the availability of oracle randomness (i.e, having a long read only random string) in order to avoid
the space penalty of explicitly storing many independent random permutations.
In practice (where we assume cryptographic pseudo-random number generators exist),
it is straightforward to generate the bits of the oracle random string on demand,
ensuring that computationally bounded systems essentially cannot
produce hard inputs for the algorithm.

The randomized algorithms \Cref{alg:bpt-rand-va-edgecol} and \Cref{alg:rand-ea-edgecol} both use the same idea of trying and discarding (making unavailable for future use), either immediately or periodically, a set of fresh colors chosen by iterating over a random permutation. This construction has the downside that, since many colors are discarded, the total number $C$ of colors in the algorithm might use must be large. Instead of discarding colors,
a possibly more efficient approach is to retain, for each vertex, a pool of all
the colors that were tried but not used; this ensures that colors are only removed from consideration when they have been used. The downside of retaining unused colors
is an increased space usage that is harder to prove upper bounds for. For the following two algorithms, \Cref{alg:conj-va} and \Cref{alg:conj-ea}, for one-sided vertex arrival, and edge arrival streams,
we conjecture that they require $\tO(n)$ and $\tO(n\sqrt{\Delta})$ bits of space
with high probability, but have not been able to prove this. The second algorithm
in particular is rather similar to an online edge coloring algorithm conjectured to
use $\Delta + O(\sqrt{\Delta} \log n)$ colors by \cite{BarNoyMN92}, in which each
edge is assigned a uniformly random color from the set of colors that no edges
incident to its endpoints have used.

\begin{algorithm}[htb!]
  \caption{A simple randomized algorithm for $(2\Delta-1)$-edge coloring in the one-sided vertex arrival model which is conjectured to use $O(n \log \Delta)$ space w.h.p.; and uses $\tO(n\Delta)$ oracle random bits\label{alg:conj-va}}
  \begin{algorithmic}[1]
  \Statex \textbf{Input}: Stream of one-sided vertex arrivals on $n$-vertex graph $G=(A \sqcup B)$. 
    \Statex Let $C := 2 \Delta - 1$.
    \Statex
    \Statex \underline{\textbf{Initialize():}}
      \For{$z \in B$}
        \State $\sigma_z \gets$ uniformly randomly chosen permutation over $[C]$
        \State $h_z \gets 1$
        \State $F_z \gets \emptyset$
      \EndFor
    
  \Statex 
  \Statex\underline{\textbf{Process}(vertex $a \in A$, adjacent edges $M_a$)}
    \State $S \gets \emptyset$
    \For{$\{a,b\} \in M_a$, in random order}
      \While{$F_b \subseteq S$}
        \State $F_b \gets F_b \cup \sigma_b[h_b]$
        \State $h_b \gets h_b + 1$
      \EndWhile
      \State Let $c$ be random color from $F_b \setminus S$
      \State Assign color $c$ to edge $\{a,b\}$
      \State $F_b \gets F_b \setminus \{c\}$
      \State $S \gets S \cup \{c\}$
    \EndFor
  \end{algorithmic}
\end{algorithm}

\begin{algorithm}[htb!]
  \caption{A simple randomized algorithm for $(2\Delta-1)$-edge coloring in the edge arrival model which is conjectured to use $O(n \sqrt{\Delta} \log \Delta)$ space w.h.p.; and uses $\tO(n\Delta)$ oracle random bits\label{alg:conj-ea}}
  \begin{algorithmic}[1]
  \Statex \textbf{Input}: Stream of edge arrivals on $n$-vertex graph $G=(V, E)$. 
    \Statex Let $C := 2 \Delta - 1$.
    \Statex
    \Statex \underline{\textbf{Initialize:}}
      \For{$v \in B$}
        \State $\sigma_v \gets$ uniformly randomly chosen permutation over $[C]$
        \State $h_v \gets 1$
        \State $F_v \gets \emptyset$
      \EndFor
    
  \Statex 
  \Statex\underline{\textbf{Process}(edge $\{x,y\}$):}
    \While{$F_x \cap F_y = \emptyset$}
        \State $F_x \gets F_x \cup \sigma_x[h_y]$
        \State $h_x \gets h_x + 1$
        \State $F_y \gets F_y \cup \sigma_y[h_y]$
        \State $h_y \gets h_y + 1$
    \EndWhile
    \State Let $c$ be random color from $F_x \cap F_y$
    \State Assign color $c$ to edge $\{x,y\}$
    \State $F_x \gets F_x \setminus \{c\}$
    \State $F_y \gets F_y \setminus \{c\}$
    \State $S \gets \emptyset$
  \end{algorithmic}
\end{algorithm}

\bibliographystyle{alpha}
\bibliography{refs}

\newcommand{\etalchar}[1]{$^{#1}$}
\begin{thebibliography}{JURdW16}

\bibitem[AL12]{AlonL12}
Noga Alon and Shachar Lovett.
\newblock Almost k-wise vs. k-wise independent permutations, and uniformity for
  general group actions.
\newblock In {\em Proc. 16th International Workshop on Randomization and
  Approximation Techniques in Computer Science}, pages 350--361. Springer,
  2012.

\bibitem[AMSZ03]{aggarwal2003switch}
Gagan Aggarwal, Rajeev Motwani, Devavrat Shah, and An~Zhu.
\newblock Switch scheduling via randomized edge coloring.
\newblock In {\em 44th Annual IEEE Symposium on Foundations of Computer Science
  (FOCS), 2003,}, pages 502--512. IEEE, 2003.

\bibitem[ASZZ22]{AnsariSZ22}
Mohammad Ansari, Mohammad Saneian, and Hamid Zarrabi-Zadeh.
\newblock {Simple Streaming Algorithms for Edge Coloring}.
\newblock In Shiri Chechik, Gonzalo Navarro, Eva Rotenberg, and Grzegorz
  Herman, editors, {\em 30th Annual European Symposium on Algorithms (ESA
  2022)}, volume 244 of {\em Leibniz International Proceedings in Informatics
  (LIPIcs)}, pages 8:1--8:4, Dagstuhl, Germany, 2022. Schloss Dagstuhl --
  Leibniz-Zentrum f{\"u}r Informatik.

\bibitem[BDH{\etalchar{+}}19]{BehnezhadDHKS19}
Soheil Behnezhad, Mahsa Derakhshan, MohammadTaghi Hajiaghayi, Marina Knittel,
  and Hamed Saleh.
\newblock Streaming and massively parallel algorithms for edge coloring.
\newblock In {\em 27th Annual European Symposium on Algorithms, {ESA} 2019,
  September 9-11, 2019, Munich/Garching, Germany}, volume 144 of {\em LIPIcs},
  pages 15:1--15:14. Schloss Dagstuhl - Leibniz-Zentrum f{\"{u}}r Informatik,
  2019.

\bibitem[BGW21]{BhattacharyaGW21}
Sayan Bhattacharya, Fabrizio Grandoni, and David Wajc.
\newblock Online edge coloring algorithms via the nibble method.
\newblock In {\em Proceedings of the 2021 {ACM-SIAM} Symposium on Discrete
  Algorithms, {SODA} 2021, Virtual Conference, January 10 - 13, 2021}, pages
  2830--2842. {SIAM}, 2021.

\bibitem[BJWY20]{BenEliezerJWY20}
Omri Ben{-}Eliezer, Rajesh Jayaram, David~P. Woodruff, and Eylon Yogev.
\newblock A framework for adversarially robust streaming algorithms.
\newblock In {\em Proc. 39th ACM Symposium on Principles of Database Systems},
  page 63–80, 2020.

\bibitem[BMM12]{BahmaniMM12}
Bahman Bahmani, Aranyak Mehta, and Rajeev Motwani.
\newblock Online graph edge-coloring in the random-order arrival model.
\newblock {\em Theory of Computing}, 8(1):567--595, 2012.

\bibitem[BMN92]{BarNoyMN92}
Amotz Bar{-}Noy, Rajeev Motwani, and Joseph Naor.
\newblock The greedy algorithm is optimal for on-line edge coloring.
\newblock {\em Information Processing Letters}, 44(5):251--253, 1992.

\bibitem[BS23]{BehnezhadS23}
Soheil Behnezhad and Mohammad Saneian.
\newblock Streaming edge coloring with asymptotically optimal colors.
\newblock {\em arXiv preprint arXiv:2305.01714}, 2023.

\bibitem[CL21]{CharikarL21}
Moses Charikar and Paul Liu.
\newblock Improved algorithms for edge colouring in the {W}-streaming model.
\newblock In {\em 4th Symposium on Simplicity in Algorithms, {SOSA} 2021,
  Virtual Conference, January 11-12, 2021}, pages 181--183. {SIAM}, 2021.

\bibitem[CMZ23]{ChechikMZ23}
Shiri Chechik, Doron Mukhtar, and Tianyi Zhang.
\newblock Streaming edge coloring with subquadratic palette size.
\newblock {\em arXiv preprint arXiv:2305.07090}, 2023.

\bibitem[CPW19]{CohenPW19}
Ilan~Reuven Cohen, Binghui Peng, and David Wajc.
\newblock Tight bounds for online edge coloring.
\newblock In {\em 60th {IEEE} Annual Symposium on Foundations of Computer
  Science, {FOCS} 2019, Baltimore, Maryland, USA, November 9-12, 2019}, pages
  1--25. {IEEE} Computer Society, 2019.

\bibitem[Czu15]{Czumaj15}
Artur Czumaj.
\newblock Random permutations using switching networks.
\newblock In {\em Proc. 47th Annual ACM Symposium on the Theory of Computing},
  pages 703--712, 2015.

\bibitem[DEMR10]{DemetrescuEMR10}
Camil Demetrescu, Bruno Escoffier, Gabriel Moruz, and Andrea Ribichini.
\newblock Adapting parallel algorithms to the w-stream model, with applications
  to graph problems.
\newblock {\em Theoretical Computer Science}, 411(44):3994--4004, 2010.

\bibitem[DFR06]{DemetrescuFR06}
Camil Demetrescu, Irene Finocchi, and Andrea Ribichini.
\newblock Trading off space for passes in graph streaming problems.
\newblock In {\em Proceedings of the Seventeenth Annual {ACM-SIAM} Symposium on
  Discrete Algorithms, {SODA} 2006}, pages 714--723. {ACM} Press, 2006.

\bibitem[EFKM10]{EhmsenFKM10}
Martin~R. Ehmsen, Lene~M. Favrholdt, Jens~S. Kohrt, and Rodica Mihai.
\newblock Comparing first-fit and next-fit for online edge coloring.
\newblock {\em Theor. Comput. Sci.}, 411(16-18):1734--1741, 2010.

\bibitem[EJ01]{ErlebachJ01}
Thomas Erlebach and Klaus Jansen.
\newblock The complexity of path coloring and call scheduling.
\newblock {\em Theoretical Computer Science}, 255(1):33--50, 2001.

\bibitem[FM18]{FavrholdtM18}
Lene~M. Favrholdt and Jesper~W. Mikkelsen.
\newblock Online edge coloring of paths and trees with a fixed number of
  colors.
\newblock {\em Acta Informatica}, 55(1):57--80, 2018.

\bibitem[FN03]{FavrholdtN03}
Lene~M. Favrholdt and Morten~N. Nielsen.
\newblock On-line edge-coloring with a fixed number of colors.
\newblock {\em Algorithmica}, 35(2):176--191, 2003.

\bibitem[GDP05]{GandhamDP05}
S.~Gandham, M.~Dawande, and R.~Prakash.
\newblock Link scheduling in sensor networks: distributed edge coloring
  revisited.
\newblock In {\em Proceedings IEEE 24th Annual Joint Conference of the IEEE
  Computer and Communications Societies.}, volume~4, pages 2492--2501 vol. 4,
  2005.

\bibitem[GSS22]{GlazikSS22}
Christian Glazik, Jan Schiemann, and Anand Srivastav.
\newblock A one pass streaming algorithm for finding euler tours.
\newblock {\em Theory of Computing Systems}, pages 1--23, 12 2022.

\bibitem[Hol81]{Holyer81}
Ian Holyer.
\newblock The np-completeness of edge-coloring.
\newblock {\em SIAM Journal on Computing}, 10(4):718--720, 1981.

\bibitem[JP83]{JoagDevP83}
Kumar {Joag-Dev} and Frank Proschan.
\newblock Negative association of random variables, with applications.
\newblock {\em Ann. Stat.}, 11(1):286--295, 1983.

\bibitem[JURdW16]{JanuarioURD16}
Tiago Januario, Sebastián Urrutia, Celso~C. Ribeiro, and Dominique de. Werra.
\newblock Edge coloring: A natural model for sports scheduling.
\newblock {\em European Journal of Operational Research}, 254(1):1--8, 2016.

\bibitem[KLS{\etalchar{+}}22]{KulkarniLSST22}
Janardhan Kulkarni, Yang~P. Liu, Ashwin Sah, Mehtaab Sawhney, and Jakub
  Tarnawski.
\newblock Online edge coloring via tree recurrences and correlation decay.
\newblock In {\em {STOC} '22: 54th Annual {ACM} {SIGACT} Symposium on Theory of
  Computing, Rome, Italy, June 20 - 24, 2022}, pages 104--116. {ACM}, 2022.

\bibitem[LS11]{LauraS11}
Luigi Laura and Federico Santaroni.
\newblock Computing strongly connected components in the streaming model.
\newblock In Alberto Marchetti-Spaccamela and Michael Segal, editors, {\em
  Theory and Practice of Algorithms in (Computer) Systems}, pages 193--205.
  Springer Berlin Heidelberg, 2011.

\bibitem[MG92]{MisraG92}
Jayadev Misra and David Gries.
\newblock A constructive proof of vizing's theorem.
\newblock {\em Information Processing Letters}, 41(3):131--133, 1992.

\bibitem[Mik15]{Mikkelsen15}
Jesper~W. Mikkelsen.
\newblock Optimal online edge coloring of planar graphs with advice.
\newblock In {\em Algorithms and Complexity - 9th International Conference,
  {CIAC} 2015, Paris, France, May 20-22, 2015. Proceedings}, volume 9079 of
  {\em Lecture Notes in Computer Science}, pages 352--364. Springer, 2015.

\bibitem[Mik16]{Mikkelsen16}
Jesper~W. Mikkelsen.
\newblock Randomization can be as helpful as a glimpse of the future in online
  computation.
\newblock In {\em 43rd International Colloquium on Automata, Languages, and
  Programming, {ICALP} 2016, July 11-15, 2016, Rome, Italy}, volume~55 of {\em
  LIPIcs}, pages 39:1--39:14. Schloss Dagstuhl - Leibniz-Zentrum f{\"{u}}r
  Informatik, 2016.

\bibitem[Mor13]{Morris13}
Ben Morris.
\newblock Improved mixing time bounds for the thorp shuffle.
\newblock {\em Combinatorics, Probability and Computing}, 22(1):118--132, 2013.

\bibitem[NSW23]{NaorSW23}
Joseph Naor, Aravind Srinivasan, and David Wajc.
\newblock Online dependent rounding schemes.
\newblock {\em CoRR}, abs/2301.08680, 2023.

\bibitem[RU94]{RaghavanU94}
Prabhakar Raghavan and Eli Upfal.
\newblock Efficient routing in all-optical networks.
\newblock In {\em Proceedings of the twenty-sixth annual ACM symposium on
  Theory of computing (STOC)}, pages 134--143, 1994.

\bibitem[Sha49]{Shannon49}
Claude~E. Shannon.
\newblock A theorem on coloring the lines of a network.
\newblock {\em Journal of Mathematics and Physics}, 28(1-4):148--152, 1949.

\bibitem[SS96]{SipserS96}
Michael Sipser and Daniel~A. Spielman.
\newblock Expander codes.
\newblock {\em IEEE Transactions on Information Theory}, 42(6):1710--1722,
  1996.

\bibitem[SW21]{SaberiW21}
Amin Saberi and David Wajc.
\newblock The greedy algorithm is not optimal for on-line edge coloring.
\newblock In {\em 48th International Colloquium on Automata, Languages, and
  Programming, {ICALP} 2021, July 12-16, 2021, Glasgow, Scotland (Virtual
  Conference)}, volume 198 of {\em LIPIcs}, pages 109:1--109:18. Schloss
  Dagstuhl - Leibniz-Zentrum f{\"{u}}r Informatik, 2021.

\bibitem[Viz64]{Vizing64}
V.~G. Vizing.
\newblock On an estimate of the chromatic class of a p-graph.
\newblock {\em Discret Analiz}, 3:25--30, 1964.

\bibitem[WC81]{WegmanC81}
Mark~N. Wegman and Larry Carter.
\newblock New hash functions and their use in authentication and set equality.
\newblock {\em J. Comput. Syst. Sci.}, 22(3):265--279, 1981.

\end{thebibliography}

\end{document}